\documentclass[11pt]{article}
\usepackage{xifthen}
\usepackage[colorlinks=true,
            citecolor=black,
            linkcolor=black,
            urlcolor=black]{hyperref}
\usepackage{fullpage}
\usepackage{latexsym,amsmath,amsfonts,amssymb,stmaryrd,amsthm}
\usepackage{url}
\usepackage{proof}
\usepackage[square,authoryear,semicolon]{natbib}
\usepackage{cleveref}
\usepackage[all]{xy}
\usepackage{manfnt}

\theoremstyle{definition}
\newtheorem{definition}{Definition}

\theoremstyle{remark}
\newtheorem{remark}[definition]{Remark}
\newtheorem{corollary}[definition]{Corollary}

\theoremstyle{plain}
\newtheorem{lemma}[definition]{Lemma}
\newtheorem{theorem}[definition]{Theorem}
\crefname{remark}{Remark}{Remarks} 
\crefname{lemma}{Lemma}{Lemmas} 
\crefname{theorem}{Theorem}{Theorems} 
\crefname{definition}{Definition}{Definitions} 
\crefname{section}{Section}{Sections} 
\crefname{appendix}{Appendix}{Appendices} 

\newcommand{\G}{\ensuremath{\Gamma}}
\newcommand{\e}{\ensuremath{\varepsilon}}
\newcommand{\eb}{\overline{\e}}
\newcommand{\eq}{\ensuremath{\mathbin{\doteq}}}
\newcommand{\fresh}{\ \#\ }
\newcommand{\fd}[1]{\ensuremath{\mathsf{FD}(#1)}}

\newcommand{\oft}[2]{#1\mathbin{:}#2}
\newcommand{\J}{\ensuremath{\mathcal{J}}}
\newcommand{\cube}{\ensuremath{\text{\mancube}}}

\makeatletter
\def\rightharpoonupfill@{\arrowfill@\relbar\relbar\rightharpoonup}
\newcommand{\overrightharpoonup}{%
\mathpalette{\overarrow@\rightharpoonupfill@}}
\makeatother

\newcommand{\subst}[3]{\ensuremath{#1 [#2 / #3]}}
\newcommand{\dsubst}[3]{\ensuremath{#1 \langle{#2}/{#3}\rangle}}

\newcommand{\arr}[2]{\ensuremath{#1 \to #2}}
\newcommand{\picl}[3]{\ensuremath{({#1}{:}{#2}) \to #3}}
\newcommand{\lam}[2]{\ensuremath{\lambda{#1}.{#2}}}
\newcommand{\app}[2]{\ensuremath{\mathsf{app}({#1},{#2})}}

\newcommand{\sigmacl}[3]{\ensuremath{({#1}{:}{#2}) \times #3}}
\newcommand{\pair}[2]{\ensuremath{\langle #1,#2\rangle}}
\newcommand{\fst}[1]{\ensuremath{\mathsf{fst}(#1)}}
\newcommand{\snd}[1]{\ensuremath{\mathsf{snd}(#1)}}

\newcommand{\Id}[3]{\ensuremath{\mathsf{Id}_{#1}(#2,#3)}}
\newcommand{\dlam}[2]{\ensuremath{\langle #1 \rangle #2}}
\newcommand{\dapp}[2]{\ensuremath{#1 @ #2}}

\newcommand{\bool}{\ensuremath{\mathsf{bool}}}
\newcommand{\notb}[1]{\ensuremath{\mathsf{not}_{#1}}}
\newcommand{\notf}[1]{\ensuremath{\mathsf{not}(#1)}}
\newcommand{\notel}[2]{\ensuremath{\mathsf{notel}_{#1}(#2)}}
\newcommand{\true}{\ensuremath{\mathsf{true}}}
\newcommand{\false}{\ensuremath{\mathsf{false}}}
\newcommand{\ifsym}{\ensuremath{\mathsf{if}}}
\newcommand{\ifb}[4]{\ensuremath{\ifsym_{#1}(#2;#3,#4)}}
\newcommand{\ifbdots}[2]{\ensuremath{\ifsym_{#1}(#2;\dots)}}

\newcommand{\C}{\ensuremath{\mathbb{S}^1}}
\newcommand{\base}{\ensuremath{\mathsf{base}}}
\newcommand{\lp}[1]{\ensuremath{\mathsf{loop}_{#1}}}
\newcommand{\Celimsym}{\ensuremath{\mathbb{S}^1\mathsf{\text{-}elim}}}
\newcommand{\Celim}[4]{\ensuremath{\Celimsym_{#1}({#2};{#3},{#4})}}

\newcommand{\coesym}{\ensuremath{\mathsf{coe}}}
\newcommand{\coe}[4]{\ensuremath{\mathsf{coe}_{#1}^{#2 \rightsquigarrow #3}(#4)}}
\newcommand{\coegeneric}[1]{\coe{#1}{r}{r'}{M}}

\newcommand{\etc}[1]{\ensuremath{\overrightharpoonup{#1}}}
\newcommand{\hcomsym}{\ensuremath{\mathsf{hcom}}}
\newcommand{\hcom}[6]{\ensuremath{%
\hcomsym_{#2}^{#1}(#3\rightsquigarrow #4,#5;#6)}}
\newcommand{\hcomgeneric}[2]{\hcom{#1}{#2}{r}{r'}{M}{\etc{y.N^\e_i}}}
\newcommand{\comsym}{\ensuremath{\mathsf{com}}}
\newcommand{\com}[6]{\ensuremath{%
\comsym_{#2}^{#1}(#3\rightsquigarrow #4,#5;#6)}}
\newcommand{\comgeneric}[2]{\com{#1}{#2}{r}{r'}{M}{\etc{y.N^\e_i}}}

\newcommand{\steps}{\ensuremath{\longmapsto}}
\newcommand{\evals}{\ensuremath{\Downarrow}}
\newcommand{\isval}[1]{\ensuremath{#1\ \mathsf{val}}}

\newcommand{\wftm}[2][\Psi]{\ensuremath{{#2}\ \mathsf{tm}\ [#1]}}
\newcommand{\wfval}[2][\Psi]{\ensuremath{#2\ \mathsf{val}\ [#1]}}

\newcommand{\msubsts}[3]{\ensuremath{{#2} : {#1} \to {#3}}}
\newcommand{\psitd}{\msubsts{\Psi'}{\psi}{\Psi}}
\newcommand{\td}[2]{\ensuremath{{#1}{#2}}}
\newcommand{\id}[1][\Psi]{\ensuremath{\mathsf{id}_{#1}}}

\newcommand{\veqper}[3][\Psi]
  {\ensuremath{{#2} \approx^{#1} {#3}}}
\newcommand{\veqperfour}[5][\Psi]
  {\ensuremath{{#2} \approx^{#1} {#3} \approx^{#1} {#4} \approx^{#1} {#5}}}
\newcommand{\eqper}[3][\Psi]
  {\ensuremath{{#2} \sim^{#1} {#3}}}
\newcommand{\eqperfour}[5][\Psi]
  {\ensuremath{{#2} \sim^{#1} {#3} \sim^{#1} {#4} \sim^{#1} {#5}}}

\newcommand{\vinper}[4][\Psi]
  {\ensuremath{{#2} \approx^{#1}_{#4} {#3}}}
\newcommand{\vinperfour}[6][\Psi]
  {\ensuremath{{#2} \approx^{#1}_{#6} {#3}
   \approx^{#1}_{#6} {#4} \approx^{#1}_{#6} {#5}}}

\newcommand{\inper}[4][\Psi]
  {\ensuremath{{#2} \sim^{#1}_{#4} {#3}}}

\newcommand{\inperfour}[6][\Psi]
  {\ensuremath{{#2}\sim^{#1}_{#6} {#3} \sim^{#1}_{#6} {#4} \sim^{#1}_{#6} {#5}}}

\newcommand{\judgctx}[3][\Psi]
  {\ensuremath{{#2} \gg {#3}\ [#1]}}
\newcommand{\judg}[2][\Psi]
  {\ensuremath{{#2}\ [#1]}}
\newcommand{\judgres}[3][\Psi]
  {\ensuremath{{#3}\ [#1 \mid #2]}}

\newcommand{\cpretype}[2][\Psi]
  {\ensuremath{{#2}\ \mathsf{pretype}\ [#1]}}
\newcommand{\cpretyperes}[3][\Psi]
  {\ensuremath{{#3}\ \mathsf{pretype}\ [#1 \mid #2]}}
\newcommand{\ceqpretype}[3][\Psi]
  {\ensuremath{{#2}\eq {#3} \ \mathsf{pretype}\ [#1]}}
\newcommand{\ceqpretyperes}[4][\Psi]
  {\ensuremath{{#3}\eq {#4} \ \mathsf{pretype}\ [#1 \mid #2]}}
\newcommand{\pretype}[3][\Psi]
  {\ensuremath{{#2} \gg {#3} \ \mathsf{pretype}\ [#1]}}
\newcommand{\pretyperes}[4][\Psi]
  {\ensuremath{{#3} \gg {#4} \ \mathsf{pretype}\ [#1 \mid #2]}}
\newcommand{\eqpretype}[4][\Psi]
  {\ensuremath{{#2} \gg {#3} \eq {#4}\ \mathsf{pretype}\ [#1]}}
\newcommand{\eqpretyperes}[5][\Psi]
  {\ensuremath{{#3} \gg {#4} \eq {#5}\ \mathsf{pretype}\ [#1 \mid #2]}}

\newcommand{\wfctx}[2][\Psi]
  {\ensuremath{{#2}\ \mathsf{ctx}\ [#1]}}
\newcommand{\wfctxres}[3][\Psi]
  {\ensuremath{{#3}\ \mathsf{ctx}\ [#1 \mid #2]}}

\newcommand{\cwftype}[2][\Psi]
  {\ensuremath{{#2}\ \mathsf{type}\ [#1]}}
\newcommand{\cwftyperes}[3][\Psi]
  {\ensuremath{{#3}\ \mathsf{type}\ [#1 \mid #2]}}
\newcommand{\wftype}[3][\Psi]
  {\ensuremath{{#2} \gg {#3}\ \mathsf{type}\ [#1]}}
\newcommand{\ceqtype}[3][\Psi]
  {\ensuremath{{#2} \eq {#3}\ \mathsf{type}\ [#1]}}
\newcommand{\ceqtyperes}[4][\Psi]
  {\ensuremath{{#3} \eq {#4}\ \mathsf{type}\ [#1 \mid #2]}}
\newcommand{\eqtype}[4][\Psi]
  {\ensuremath{{#2} \gg {#3} \eq {#4}\ \mathsf{type}\ [#1]}}

\newcommand{\coftype}[3][\Psi]
  {\ensuremath{{#2} \in {#3}\ [#1]}}
\newcommand{\coftyperes}[4][\Psi]
  {\ensuremath{{#3} \in {#4}\ [#1 \mid #2]}}
\newcommand{\oftype}[4][\Psi]
  {\ensuremath{{#2} \gg {#3} \in {#4}\ [#1]}}
\newcommand{\eqtm}[5][\Psi]
  {\ensuremath{{#2} \gg {#3} \eq {#4} \in {#5}\ [#1]}}
\newcommand{\eqtmres}[6][\Psi]
  {\ensuremath{{#3} \gg {#4} \eq {#5} \in {#6}\ [#1 \mid #2]}}
\newcommand{\ceqtm}[4][\Psi]
  {\ensuremath{{#2} \eq {#3} \in {#4}\ [#1]}}
\newcommand{\ceqtmres}[5][\Psi]
  {\ensuremath{{#3} \eq {#4} \in {#5}\ [#1 \mid #2]}}
\newcommand{\ceqtmtab}[4][\Psi]
  {\ensuremath{{#2} \eq\ & {#3} \in {#4}\ [#1]}}

\newcommand{\sbool}{\ensuremath{\mathsf{sbool}}}
\newcommand{\VPER}{\ensuremath{\mathsf{VPER}}}
\newcommand{\vperleq}{\ensuremath{\sqsubseteq}}
\newcommand{\CTS}{\ensuremath{\mathsf{CTS}}}
\newcommand{\ctsleq}{\ensuremath{\sqsubseteq}}
\newcommand{\relcts}[2]{\ensuremath{{#1} \models \left({#2}\right)}}

\newcommand{\ia}[2]{\ensuremath{\mathsf{ia}_{#1}(#2)}}
\newcommand{\iain}[2]{\ensuremath{\mathsf{iain}_{#1}(#2)}}
\newcommand{\iaout}[2]{\ensuremath{\mathsf{iaout}_{#1}(#2)}}
\newcommand{\Iso}[3][\Psi]{\ensuremath{\mathsf{Iso}_{#2}({#3})\ [#1]}}
\newcommand{\Isoeq}[4][\Psi]
  {\ensuremath{\mathsf{Iso}_{#2}({#3}) \eq \mathsf{Iso}_{#2}({#4})\ [#1]}}

\title{Computational Higher Type Theory II:\\Dependent Cubical Realizability}

\author{Carlo Angiuli\thanks{\texttt{cangiuli@cs.cmu.edu}}\\Carnegie Mellon University
  \and Robert Harper\thanks{\texttt{rwh@cs.cmu.edu}}\\Carnegie Mellon University}

\date{July, 2016}

\begin{document}

\maketitle{}

\begin{abstract}
  This is the second in a series of papers extending Martin-L\"{o}f's
  \emph{meaning explanation} of dependent type theory to account for
  higher-dimensional types.  We build on the cubical realizability framework for
  simple types developed in Part I, and extend it to a meaning explanation of
  dependent higher-dimensional type theory.  This extension requires
  generalizing the computational Kan condition given in Part I, and considering
  the action of type families on paths.  We define \emph{identification types},
  which classify identifications (paths) in a type, and dependent function and
  pair types.  The main result is a \emph{canonicity theorem}, which states
  that a closed term of boolean type evaluates to either true or false.  This
  result establishes the first computational interpretation of higher dependent
  type theory by giving a deterministic operational semantics for its programs,
  including operations that realize the Kan condition.
\end{abstract}

\section{Dependent Types}
\label{sec:dep}

In Part I of this series~\citep{ahw2016cubical}
we introduced \emph{abstract cubical
  realizability} to provide a ``meaning explanation'' of
higher-dimensional simple type theory in the style of
\citet{cmcp} and \citet{constableetalnuprl}.
In Part II we extend cubical realizability to
higher-dimensional \emph{dependent type theory}, which considers
type-indexed families of types such as the \emph{cubical
  identification type}~\citep{licata2014cubical,cohen2016cubical} and dependent
generalizations of the function and product types.  The construction
provides the first deterministic computational meaning explanation for
higher-dimensional dependent type theory:
\begin{theorem}[Canonicity]
If $\oftype[\cdot]{\cdot}{M}{\bool}$ then either
$M \evals \true$ or $M \evals \false$.
\end{theorem}
\noindent The proof of the canonicity theorem is straightforward, once
the extended realizability interpretation has been obtained.  All of
the effort is in formulating the meaning explanation and showing that
it has the requisite properties.

The development follows along the lines of Part I, but is extended to
account for type-indexed families of types.  The main effect of
dependency is to induce lines between the instance types of the
family.  Briefly, if $B$ is a family of types indexed by the type $A$,
then points in $A$ determine instances of $B$, and lines in $A$
determine lines between those instances.  The framework of cubical
realizability is defined so as to account for lines (and higher cells)
between types, and so dependency itself presents no further
complications.  However, accounting for the cubical identification
type, whose $n$-cells classify $n+1$-cells in a type, requires a
modest generalization of the Kan composition operation to admit
additional constraints.  Whereas in the simply typed case it suffices
to consider lines (in any dimension) constrained by lines, here it is
necessary to admit more general forms of composition problem.  This is
managed by extending the Kan composition operation itself to allow for
additional ``tube faces'', and enriching the notion of dimension
context to allow for equational constraints among dimension variables.
This extension is sufficient to account for the cubical identification
type, and closure under dependent forms of function and product types.
We have also clarified the definition of a cubical type system, and what it
means for a cubical type system to have certain type formers; this
reorganization has no material effect on our development.

The purpose of a meaning explanation is to provide a computational
semantics for typehood and membership judgments (more precisely, for
exact equality of types and exact equality of members of a type at all
finite dimensions).  The explanation is based on the acceptability of
basic predicative comprehension principles used to define the meanings
of types.  These principles seem scarcely deniable; principles of the
same or greater strength would be required to validate the sensibility
of a formal type theory defined by a collection of rules.  The meaning
explanation provides a concept of \emph{truth} in type theory grounded
in computation, following the principles elucidated by Martin-L\"{o}f
and Constable in the one-dimensional case.  From this point of view
the role of proof theory is to provide a window on the truth.  There
are few restrictions on the form of a proof theory other than that it
derive only true judgments.  For example, a proof theory for
computational higher type theory need not admit meta-theoretic
properties, such as cut elimination or decidability, that are
requisite for formal type theories.  However, the rules of formal
cubical type theory given by~\citet{licata2014cubical} and
\citet{cohen2016cubical} are valid; we give explicit proofs for some
representative cases.


\subsection*{Acknowledgements}

This paper is an extension of abstract cubical realizability as
introduced by~\citet{ahw2016cubical}, whose primary antecedents are
two-dimensional type theory~\citep{lh2dtt}, the uniform Kan cubical
model of homotopy type theory~\citep{bch} and the formal cubical type
theories~\citep{cohen2016cubical,licata2014cubical}.  We are greatly
indebted to Marc Bezem, Evan Cavallo, Kuen-Bang Hou (Favonia), Simon
Huber, Dan Licata, Ed Morehouse, and Todd Wilson for their
contributions and advice.  The authors gratefully acknowledge the
support of the Air Force Office of Scientific Research through MURI
grant FA9550-15-1-0053.  Any opinions, findings and conclusions or
recommendations expressed in this material are those of the authors
and do not necessarily reflect the views of the AFOSR.

\newpage
\section{Programming language}
\label{sec:opsem}

The programming language itself has two sorts, dimensions and terms,
and binders for both sorts.  Terms are an ordinary untyped lambda calculus
with constructors; dimensions are either dimension
constants ($0$ or $1$) or one of countably many dimension names
($x,y,\dots$) behaving like nominal constants~\citep{pittsnominal}.
Dimension terms
occur at specific positions in some terms; for example, $\lp{r}$ is a term
for any dimension term $r$.  The operational semantics is
defined on terms that are closed with respect to term variables but may
contain free dimension names.

Dimension names represent generic elements of an abstract interval whose end
points are notated $0$ and $1$. While one may sensibly substitute any dimension
term for a dimension name, terms are \emph{not} to be understood solely in terms
of their dimensionally-closed instances (namely, their end points). Rather, a
term's dependence on dimension names is to be understood generically;
geometrically, one might imagine additional unnamed points in the interior of
the abstract interval.

The language features two terms that are specific to higher type
theory.  The first, called \emph{coercion}, has the form
$\coegeneric{x.A}$, where $x.A$ is a type line, $r$ is the
\emph{starting} dimension and $r'$ is the \emph{ending} dimension.
Coercion transports a term $M$ from $\dsubst{A}{r}{x}$ to
$\dsubst{A}{r'}{x}$ using the type line $x.A$ as a guide.
Coercion from $r$ to itself has no effect, up to exact
equality.  Coercion from $0$ to $1$ or vice versa is transport, which
applies one direction of the equivalence induced by the type line.
Coercion from $0$ or $1$ to a dimension name $y$ creates a $y$-line in 
$\dsubst{A}{y}{x}$, and coercion from $y$ to $0$ or $1$ 
yields a line between one end point of the input $y$-line and
the transport of the opposite end point.
Finally, coercion from one dimension name to another reorients the
line from one dimension to another.

The second, called \emph{homogeneous Kan composition}, has the form
$\hcomgeneric{\etc{r_i}}{A}$, where $r_1,\dots,r_n$ are the \emph{extents}, $r$
is the \emph{starting} dimension, and $r'$ is the \emph{ending} dimension.
The term $M$ is called the \emph{cap}, and the terms $N^0_i$ and $N^1_i$
form the \emph{tube in the $r_i$ extent}.%
\footnote{If $r_i=x$ and $x$ occurs in $N^\e_i$, then the tube sides are actually
$\dsubst{N^\e_i}{\e}{x}$ for $\e=0,1$, representing that $x=\e$ on the $N^\e$
side of the composition problem. In the present description, we assume that $x$
does not occur in $N^\e$; the precise typing rules for Kan composition are given
in \cref{def:kan}.}
This composition is well-typed when the starting side of each $N^\e_i$ coincides
(up to exact equality) with the $r_i=\e$ side of the cap. When all the $r_i$ are
dimension names $x_i$, the composition results in an ($x_1,\dots,x_n$)-cube,
called the \emph{composite}, whose $\dsubst{-}{\e}{x_i}$ sides coincide with the
ending sides of each $N^\e_i$. The composite is easily visualized when $r=0$,
$r'=1$, and there is one extent $x$:

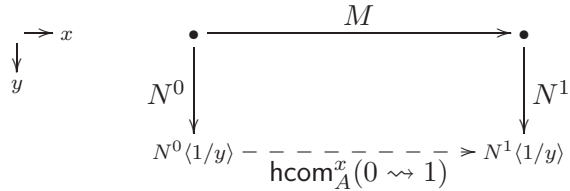
\[
\renewcommand{\objectstyle}{\scriptstyle}
\renewcommand{\labelstyle}{\textstyle}
\xymatrix@=0.75em{
  {} \ar[d] \ar[r] & x \\ y
}
\qquad
\xymatrix@C=8em@R=3em{
  \bullet \ar[d]_{N^0} \ar[r]^{M} &
  \bullet \ar[d]^{N^1} \\
  \dsubst{N^0}{1}{y} \ar@{-->}[r]_{\hcomsym^x_A(0\rightsquigarrow 1)} &
  \dsubst{N^1}{1}{y}}
\]
The case of $r=1$ and $r'=0$ is symmetric, swapping the roles of the cap and the
composite.

When the starting dimension is $r=0$ (or, analogously, $r=1$) and the ending
dimension is $r'=y$, where $y$ does not occur in $M$, the Kan composition yields
the interior of the $x,y$-square depicted above, called the \emph{filler}. One
may think of this composition as sweeping out that square by sliding the cap from
$y=0$ to any point in the $y$ dimension, much in the manner of opening a window
shade. The filler is simultaneously an $x$-line identifying the two tube sides
with each other, and a $y$-line identifying the cap with the composite.  

When $r=y$ and $r'$ is $0$ or $1$, the composition may be visualized as closing
a window shade, starting in the ``middle'' and heading towards the roll at one
end or the other.
When both $r$ and $r'$ are dimension names, the result is harder to visualize,
and is best understood formally, as is also the case where $r=0$ and $r'=y$ but
$y$ does occur in $M$.

Finally, there are two cases in which the composition scenario trivializes.
When $r=r'=0$ or $r=r'=1$, the composition is the cap itself, intuitively
because the window shade does not move from its starting position at the cap.
When an extent $r_i$ is $0$ (or $1$), rather than a dimension name, the
composition is simply $\dsubst{N^0_i}{r'}{y}$ (or $\dsubst{N^1_i}{r'}{y}$),
because the composition has no extent beyond that end point. These two cases are
important because they ensure, respectively, that the $y$ and $x$ end points of
the $x,y$-filler are as depicted above.

\subsection{Terms}

\[\begin{aligned}
M &:=
\picl{a}{A}{B} \mid
\sigmacl{a}{A}{B} \mid
\Id{x.A}{M}{N} \mid
\bool \mid
\notb{r} \mid
\C \\&\mid
\lam{a}{M} \mid
\app{M}{N} \mid
\pair{M}{N} \mid
\fst{M} \mid
\snd{M} \mid
\dlam{x}{M} \mid
\dapp{M}{r} \\&\mid
\true \mid
\false \mid
\ifb{a.A}{M}{N_1}{N_2} \mid
\notel{r}{M} \\&\mid
\base \mid
\lp{r} \mid
\Celim{a.A}{M}{N_1}{x.N_2} \\&\mid
\coegeneric{x.A} \mid
\hcomgeneric{\etc{r_i}}{A}
\end{aligned}\]

We use capital letters like $M$, $N$, and $A$ to denote terms, $r$,
$r'$, $r_i$ to denote dimension terms, $x$ to denote dimension names, $\e$
to denote dimension constants ($0$ or $1$), and $\eb$ to denote the
opposite dimension constant of $\e$.  We write $x.-$ for dimension
binders, $a.-$ for term binders, and $\fd{M}$ for the set of dimension
names free in $M$. (Additionally, in $\picl{a}{A}{B}$ and $\sigmacl{a}{A}{B}$,
$a$ is bound in $B$.)  Dimension substitution $\dsubst{M}{r}{x}$ and term
substitution $\subst{M}{N}{a}$ are defined in the usual way.

The superscript argument of $\hcomsym$ is a list of $n\geq 1$ dimension terms
$\etc{r_i} = r_1,\dots,r_n$; it then takes $2n$ term arguments with one
dimension binder each, $\etc{y.N^\e_i} = y.N^0_1,y.N^1_1,\dots,y.N^0_n,y.N^1_n$.
We use the $\etc{-}$ notation to abbreviate a list of the appropriate length,
or to abbreviate applying some term formers to each term in that list.

We employ two abbreviations in the operational semantics below:

\[\begin{aligned}
\notf{M} &:=
\ifb{\_.\bool}{M}{\false}{\true}
\\
\comgeneric{\etc{r_i}}{y.A} &:=
\hcom{\etc{r_i}}{\dsubst{A}{r'}{y}}{r}{r'}
  {\coe{y.A}{r}{r'}{M}}
  {\etc{y.\coe{y.A}{y}{r'}{N^\e_i}}}
\end{aligned}\]

\subsection{Operational semantics}

The following describes a deterministic weak head reduction evaluation
strategy for closed terms in the form of a transition system with two
judgments:
\begin{enumerate}
\item $\isval{E}$, stating that $E$ is a \emph{value}, or
  \emph{canonical form}.
\item $E\steps E'$, stating that $E$ takes \emph{one step of
    evaluation} to $E'$.
\end{enumerate}
These judgments are defined so that if $\isval{E}$, then
$E\not\steps$, but the converse need not be the case.  As usual, we
write $E\steps^* E'$ to mean that $E$ transitions to $E'$ in zero
or more steps.  We say $E$ evaluates to $V$, written $E \evals V$, when
$E\steps^* V$ and $\isval{V}$.

Most of the evaluation rules are standard, and evaluate only principal arguments
of elimination forms. The principal arguments of $\hcomsym$ and $\coesym$ are
their type subscripts, whose head constructors determine how those terms
evaluate.

Determinacy is a strong condition that implies that a term has at most
one value.
\begin{lemma}[Determinacy]
  If $M\steps M_1$ and $M\steps M_2$, then $M_1 = M_2$.
\end{lemma}

Stability states that evaluation does not introduce any new dimension names.
\begin{lemma}[Stability]
  If $M\steps M'$, then $\fd{M'}\subseteq\fd{M}$.
\end{lemma}

\paragraph{Types}

\[
\infer
  {\notb{\e} \steps \bool}
  {}
\]
\[
\infer
  {\isval{\picl{a}{A}{B}}}
  {}
\qquad
\infer
  {\isval{\sigmacl{a}{A}{B}}}
  {}
\qquad
\infer
  {\isval{\Id{a.A}{M}{N}}}
  {}
\qquad
\infer
  {\isval{\bool}}
  {}
\qquad
\infer
  {\isval{\notb{x}}}
  {}
\qquad
\infer
  {\isval{\C}}
  {}
\]

\paragraph{Hcom/coe}

\[
\infer
  {\coe{x.A}{r}{r'}{M} \steps \coe{x.A'}{r}{r'}{M}}
  {A\steps A'}
\]

\[
\infer
  {\hcomgeneric{\etc{r_i}}{A} \steps \hcomgeneric{\etc{r_i}}{A'}}
  {A\steps A'}
\]

\paragraph{Dependent function types}

\[
\infer
  {\app{M}{N} \steps \app{M'}{N}}
  {M \steps M'}
\qquad
\infer
  {\app{\lam{a}{M}}{N} \steps \subst{M}{N}{a}}
  {}
\qquad
\infer
  {\isval{\lam{a}{M}}}
  {}
\]

\[
\infer
  {\hcomgeneric{\etc{r_i}}{\picl{a}{A}{B}} \steps
   \lam{a}{\hcom{\etc{r_i}}{B}{r}{r'}{\app{M}{a}}{\etc{y.\app{N^\e_i}{a}}}}}
  {}
\]

\[
\infer
  {\coe{x.\picl{a}{A}{B}}{r}{r'}{M} \steps
   \lam{a}{\coe{x.\subst{B}{\coe{x.A}{r'}{x}{a}}{a}}
   {r}{r'}{\app{M}{\coe{x.A}{r'}{r}{a}}}}}
  {}
\]

\paragraph{Dependent pair types}

\[
\infer
  {\fst{M} \steps \fst{M'}}
  {M \steps M'}
\qquad
\infer
  {\snd{M} \steps \snd{M'}}
  {M \steps M'}
\qquad
\infer
  {\isval{\pair{M}{N}}}
  {}
\]

\[
\infer
  {\fst{\pair{M}{N}} \steps M}
  {}
\qquad
\infer
  {\snd{\pair{M}{N}} \steps N}
  {}
\]

\[
\infer
  {\begin{array}{c}
   \hcomgeneric{\etc{r_i}}{\sigmacl{a}{A}{B}} \\
   \steps \\
   \pair{\hcom{\etc{r_i}}{A}{r}{r'}{\fst{M}}{\etc{y.\fst{N^\e_i}}}}
        {\com{\etc{r_i}}{z.\subst{B}{F}{a}}
          {r}{r'}{\snd{M}}{\etc{y.\snd{N^\e_i}}}}
   \end{array}}
  {F = \hcom{\etc{r_i}}{A}{r}{z}{\fst{M}}{\etc{y.\fst{N^\e_i}}}}
\]~

\[
\infer
  {\coe{x.\sigmacl{a}{A}{B}}{r}{r'}{M} \steps
   \pair{\coe{x.A}{r}{r'}{\fst{M}}}
        {\coe{x.\subst{B}{\coe{x.A}{r}{x}{\fst{M}}}{a}}{r}{r'}{\snd{M}}}}
  {}
\]

\paragraph{Identification types}

\[
\infer
  {\dapp{M}{r} \steps \dapp{M'}{r}}
  {M \steps M'}
\qquad
\infer
  {\dapp{(\dlam{x}{M})}{r} \steps \dsubst{M}{r}{x}}
  {}
\qquad
\infer
  {\isval{\dlam{x}{M}}}
  {}
\]

\[
\infer
  {\hcomgeneric{\etc{r_i}}{\Id{x.A}{P_0}{P_1}} \steps
   \dlam{x}{\hcom{\etc{r_i},x}{A}{r}{r'}
     {\dapp{M}{x}}
     {\etc{y.\dapp{N^\e_i}{x}},\_.P_0,\_.P_1}}}
  {}
\]

\[
\infer
  {\coe{y.\Id{x.A}{P_0}{P_1}}{r}{r'}{M} \steps
   \dlam{x}{\com{x}{y.A}{r}{r'}{\dapp{M}{x}}{y.P_0,y.P_1}}}
  {}
\]

\paragraph{Booleans}

\[
\infer
  {\hcomgeneric{\etc{r_i}}{\bool} \steps \dsubst{N^\e_i}{r'}{y}}
  {\etc{r_i} = x_1,\dots,x_{i-1},\e,r_{i+1},\dots,r_n}
\qquad
\infer
  {\hcomgeneric{x_1,\dots,x_n}{\bool} \steps M}
  {r = r'}
\]

\[
\infer
  {\isval{\true}}
  {}
\qquad
\infer
  {\isval{\false}}
  {}
\qquad
\infer
  {\isval{\hcomgeneric{x_1,\dots,x_n}{\bool}}}
  {r \neq r'}
\]

\[
\infer
  {\ifb{a.A}{M}{T}{F} \steps \ifb{a.A}{M'}{T}{F}}
  {M \steps M'}
\quad
\infer
  {\ifb{a.A}{\true}{T}{F} \steps T}
  {}
\quad
\infer
  {\ifb{a.A}{\false}{T}{F} \steps F}
  {}
\]

\[
\infer
  {\begin{array}{c}
   \ifb{a.A}{\hcomgeneric{x_1,\dots,x_n}{\bool}}{T}{F} \\
   \steps \\
   \com{x_1,\dots,x_n}{z.\subst{A}{H}{a}}{r}{r'}{\ifb{a.A}{M}{T}{F}}
    {\etc{y.\ifb{a.A}{N^\e_i}{T}{F}}}
   \end{array}}
  {r \neq r' &
   H = \hcom{x_1,\dots,x_n}{\bool}{r}{z}{M}{\etc{y.N^\e_i}}}
\]~

\[
\infer
  {\coe{x.\bool}{r}{r'}{M} \steps M}
  {}
\]

\paragraph{Circle}

\[
\infer
  {\hcomgeneric{\etc{r_i}}{\C} \steps \dsubst{N^\e_i}{r'}{y}}
  {\etc{r_i} = x_1,\dots,x_{i-1},\e,r_{i+1},\dots,r_n}
\qquad
\infer
  {\hcomgeneric{x_1,\dots,x_n}{\C} \steps M}
  {r = r'}
\]

\[
\infer
  {\lp{\e} \steps \base}
  {}
\qquad
\infer
  {\isval{\base}}
  {}
\qquad
\infer
  {\isval{\lp{x}}}
  {}
\qquad
\infer
  {\isval{\hcomgeneric{x_1,\dots,x_n}{\C}}}
  {r \neq r'}
\]

\[
\infer
  {\Celim{a.A}{M}{P}{x.L} \steps \Celim{a.A}{M'}{P}{x.L}}
  {M \steps M'}
\]

\[
\infer
  {\Celim{a.A}{\base}{P}{x.L} \steps P}
  {}
\qquad
\infer
  {\Celim{a.A}{\lp{w}}{P}{x.L} \steps \dsubst{L}{w}{x}}
  {}
\]

\[
\infer
  {\begin{array}{c}
   \Celim{a.A}{\hcomgeneric{x_1,\dots,x_n}{\C}}{P}{x.L} \\
   \steps \\
   \com{x_1,\dots,x_n}{z.\subst{A}{F}{a}}{r}{r'}{\Celim{a.A}{M}{P}{x.L}}
    {\etc{y.\Celim{a.A}{N^\e_i}{P}{x.L}}}
   \end{array}}
  {r \neq r' &
   F = \hcom{x_1,\dots,x_n}{\C}{r}{z}{M}{\etc{y.N^\e_i}}}
\]~

\[
\infer
  {\coe{x.\C}{r}{r'}{M} \steps M}
  {}
\]

\paragraph{Not}

\[
\infer
  {\isval{\notel{x}{M}}}
  {}
\qquad
\infer
  {\notel{0}{M} \steps \notf{M}}
  {}
\qquad
\infer
  {\notel{1}{M} \steps M}
  {}
\]

\[
\infer
  {\coe{x.\notb{x}}{\e}{\eb}{M} \steps \notf{M}}
  {}
\qquad
\infer
  {\coe{x.\notb{x}}{\e}{\e}{M} \steps M}
  {}
\]

\[
\infer
  {\coe{x.\notb{x}}{0}{x}{M} \steps \notel{x}{\notf{M}}}
  {}
\qquad
\infer
  {\coe{x.\notb{x}}{1}{x}{M} \steps \notel{x}{M}}
  {}
\]

\[
\infer
  {\coe{x.\notb{x}}{x}{r}{M} \steps \coe{x.\notb{x}}{x}{r}{M'}}
  {M\steps M'}
\qquad
\infer
  {\coe{x.\notb{x}}{x}{r}{\notel{x}{M}} \steps \notel{r}{M}}
  {}
\qquad
\infer
  {\coe{x.\notb{y}}{r}{r'}{M} \steps M}
  {x \neq y}
\]

\[
\infer
  {\hcomgeneric{\etc{r_i}}{\notb{w}}
   \steps
   \notel{w}{\hcom{\etc{r_i}}{\bool}{r}{r'}{\coe{x.\notb{x}}{w}{1}{M}}{\etc{y.\coe{x.\notb{x}}{w}{1}{N^\e_i}}}}}
  {}
\]~


\newpage
\section{Meaning explanations}
\label{sec:meanings}

\begin{definition}\label{def:wftm}
We say $\wftm{M}$ when $M$ is a term with no free term variables, and
$\fd{M}\subseteq\Psi$.
\end{definition}

\begin{remark}\label{def:wfval}
We write $\wfval{M}$ when $\wftm{M}$ and $\isval{M}$. Being a value does not
depend on the choice of $\Psi$, so whenever $\wfval{M}$ and
$\fd{M}\subseteq\Psi'$, we also have $\wfval[\Psi']{M}$.
\end{remark}

\begin{definition}
A total dimension substitution $\psitd$ assigns to each dimension name in
$\Psi$ either $0$, $1$, or a dimension name in $\Psi'$. It follows that if
$\wftm{M}$ then $\wftm[\Psi']{\td{M}{\psi}}$.
\end{definition}

In this paper, we define the judgments of higher type theory as arising from two
families of partial equivalence relations on values: $\veqper[-]--$, which will
determine when two (pre)types are equal, and $\vinper[-]---$, which will
determine when two elements of a (pre)type are equal. We call such a pair of
relations a cubical type system. (We employ PERs as a convenient method of
describing sets equipped with an equivalence relation; elements of the
corresponding set are the values that are related to themselves.)

\begin{definition}\label{def:veqper}\label{def:vinper}
A \emph{cubical type system} consists of
\begin{enumerate}
\item
For every $\Psi$, a symmetric and transitive relation $\veqper--$ over values
$\wfval{A}$, and

\item
For every $\veqper{A}{B}$, symmetric and transitive relations $\vinper--{A}$ and
$\vinper--{B}$ over values $\wfval{M}$, such that $\vinper{M}{N}{A}$ if and only
if $\vinper{M}{N}{B}$.
\end{enumerate}
\end{definition}

\begin{remark}\label{def:eqper}\label{def:inper}
We write $\eqper{A}{B}$ when
$A\evals A_0$, $B\evals B_0$, and $\veqper{A_0}{B_0}$, and
we write $\inper{M}{N}{A}$ when 
$M\evals M_0$, $N\evals N_0$, $A\evals A_0$, and $\vinper{M_0}{N_0}{A_0}$.
\end{remark}

\subsection{Closed judgments}

We proceed by defining what it means for our core judgments to hold in any
cubical type system. (In \cref{sec:types} we use these judgments to describe
desirable properties of cubical type systems, for example, being closed under
certain type formers.)

The presuppositions of a judgment are the facts that must be true before one can
even sensibly state that judgment. For example, in \cref{def:ceqtm} below, we
presuppose that $A$ is a pretype when defining what it means for $M$ and $N$ to
be equal elements of $A$; if $A$ is not a pretype, then the PERs considered in
that definition may not even be defined.

Approximately, a term $A$ is a pretype at $\Psi$ when
$\eqper[\Psi']{\td{A}{\psi}}{\td{A}{\psi}}$ for every $\psitd$.
A term $M$ is an element of a pretype $A$ at $\Psi$ when 
$\inper[\Psi']{\td{M}{\psi}}{\td{M}{\psi}}{\td{A}{\psi}}$ for every $\psitd$.
We also demand that pretypes and their elements have \emph{coherent aspects}, a
technical condition implying that dimension substitutions can be taken
simultaneously or sequentially, before or after evaluating a term, without
affecting the outcome, up to PER equality.
(In our postfix notation for dimension substitutions, $\td{A}{\psi_1\psi_2}$
means $\td{(\td{A}{\psi_1})}{\psi_2}$.)

\begin{definition}\label{def:ceqpretype}
We say $\ceqpretype{A}{B}$, presupposing $\wftm{A}$ and $\wftm{B}$, when
for any $\msubsts{\Psi_1}{\psi_1}{\Psi}$ and
$\msubsts{\Psi_2}{\psi_2}{\Psi_1}$,
\begin{enumerate}
\item
$\td{A}{\psi_1}\evals A_1$, 
$\td{A_1}{\psi_2}\evals A_2$, 
$\td{A}{\psi_1\psi_2}\evals A_{12}$, 
\item
$\td{B}{\psi_1}\evals B_1$, 
$\td{B_1}{\psi_2}\evals B_2$, 
$\td{B}{\psi_1\psi_2}\evals B_{12}$, and
\item
$\veqperfour[\Psi_2]{A_2}{A_{12}}{B_2}{B_{12}}$.
\end{enumerate}
\end{definition}

\begin{definition}\label{def:ceqtm}
We say $\ceqtm{M}{N}{A}$, presupposing $\ceqpretype{A}{A}$,
$\wftm{M}$, and $\wftm{N}$, when for any
$\msubsts{\Psi_1}{\psi_1}{\Psi}$ and $\msubsts{\Psi_2}{\psi_2}{\Psi_1}$,
\begin{enumerate}
\item
$\td{M}{\psi_1}\evals M_1$, 
$\td{M_1}{\psi_2}\evals M_2$, 
$\td{M}{\psi_1\psi_2}\evals M_{12}$, 
\item
$\td{N}{\psi_1}\evals N_1$, 
$\td{N_1}{\psi_2}\evals N_2$, 
$\td{N}{\psi_1\psi_2}\evals N_{12}$, and
\item
$\vinperfour[\Psi_2]{M_2}{M_{12}}{N_2}{N_{12}}{A_{12}}$,
where $\td{A}{\psi_1\psi_2}\evals A_{12}$.
\end{enumerate}
\end{definition}

\begin{remark}\label{def:cpretype}\label{def:coftype}
We write $\cpretype{A}$ when $\ceqpretype{A}{A}$, and
we write $\coftype{M}{A}$ when $\ceqtm{M}{M}{A}$.
(We will similarly abbreviate further judgments without further comment.)
\end{remark}

\begin{remark}
The judgments $\ceqpretype{A}{B}$ and $\ceqtm{M}{N}{A}$ are symmetric and
transitive. Therefore, if $\ceqpretype{A}{B}$ then $\cpretype{A}$ and
$\cpretype{B}$, and if $\ceqtm{M}{N}{A}$ then $\coftype{M}{A}$ and
$\coftype{N}{A}$.
\end{remark}

If no terms in our programming language contained dimension subterms, then we
would have $M=\td{M}{\psi}$ for all $M$. The above meaning explanations would
therefore collapse into:
$\cpretype{A}$ whenever $\eqper[\Psi']{A}{A}$ for all $\Psi'$,
and $\ceqtm{M}{N}{A}$ whenever $\inper[\Psi']{M}{N}{A}$ for all $\Psi'$.
Disregarding $\Psi'$, these are precisely the ordinary meaning explanations for
computational type theory.

In order to accurately capture higher-dimensional structures, we restrict our
attention to pretypes satisfying the additional conditions of being
\emph{cubical} (ensuring their PERs are functorially indexed by the cube
category) and \emph{Kan} (ensuring they validate the $\hcomsym$ and $\coesym$
rules). This Kan condition is most easily expressed using judgments augmented by
\emph{dimension context restrictions}.

\begin{definition}\label{def:satisfies}
For any $\Psi$ and set of unoriented equations $\Xi = (r_1=r_1',\dots,r_n=r_n')$
in $\Psi$ (that is, $\fd{\etc{r_i},\etc{r_i'}}\subseteq\Psi$), we say that
$\psitd$ \emph{satisfies} $\Xi$ if $\td{r_i}{\psi} = \td{r_i'}{\psi}$ for each
$i\in [1,n]$.
\end{definition}

\begin{definition}\label{def:ceqpretyperes}
We say $\ceqpretyperes{\Xi}{A}{B}$, presupposing $\wftm{A}$, $\wftm{B}$, and
$\Xi$ is a set of equations in $\Psi$, when for any $\psitd$ satisfying $\Xi$,
$\ceqpretype[\Psi']{\td{A}{\psi}}{\td{B}{\psi}}$.
\end{definition}

\begin{definition}\label{def:ceqtmres}
We say $\ceqtmres{\Xi}{M}{N}{A}$, presupposing $\cpretyperes{\Xi}{A}$,
$\wftm{M}$, $\wftm{N}$, and $\Xi$ is a set of equations in $\Psi$, when for any
$\psitd$ satisfying $\Xi$,
$\ceqtm[\Psi']{\td{M}{\psi}}{\td{N}{\psi}}{\td{A}{\psi}}$.
\end{definition}

Notice that $\ceqpretype{A}{B}$ if and only if for all $\psitd$,
$\ceqpretype[\Psi']{\td{A}{\psi}}{\td{B}{\psi}}$. Thus if $\ceqpretype{A}{B}$,
then $\ceqpretyperes{\Xi}{A}{B}$, and the converse is true when $\Xi$ is
satisfied by all $\psitd$ (for example, when $\Xi$ is empty).

\begin{definition}\label{def:cubical}
We say $\cpretype{A}$ is \emph{cubical} if for any $\psitd$ and
$\vinper[\Psi']{M}{N}{A_0}$ (where $\td{A}{\psi}\evals A_0$),
$\ceqtm[\Psi']{M}{N}{\td{A}{\psi}}$.
\end{definition}

\begin{definition}\label{def:kan}
We say $A,B$ are \emph{equally Kan}, presupposing $\ceqpretype{A}{B}$, if the
following five conditions hold:
\begin{enumerate}
\item 
For any $\psitd$, if
\begin{enumerate}
\item $\ceqtm[\Psi']{M}{O}{\td{A}{\psi}}$,
\item $\ceqtmres[\Psi',y]{r_i=\e,r_j=\e'}{N^\e_i}{N^{\e'}_j}{\td{A}{\psi}}$
for any $i\in [1,n]$, $j\in [1,n]$, $\e=0,1$, and $\e'=0,1$,
\item $\ceqtmres[\Psi',y]{r_i=\e}{N^\e_i}{P^\e_i}{\td{A}{\psi}}$
for any $i\in [1,n]$ and $\e=0,1$, and
\item $\ceqtmres[\Psi']{r_i=\e}{\dsubst{N^\e_i}{r}{y}}{M}{\td{A}{\psi}}$
for any $i\in [1,n]$ and $\e=0,1$,
\end{enumerate}
then $\ceqtm[\Psi']{\hcomgeneric{\etc{r_i}}{\td{A}{\psi}}}%
{\hcom{\etc{r_i}}{\td{B}{\psi}}{r}{r'}{O}{\etc{y.P^\e_i}}}{\td{A}{\psi}}$.

\item 
For any $\psitd$, if
\begin{enumerate}
\item $\coftype[\Psi']{M}{\td{A}{\psi}}$,
\item $\ceqtmres[\Psi',y]{r_i=\e,r_j=\e'}{N^\e_i}{N^{\e'}_j}{\td{A}{\psi}}$
for any $i\in [1,n]$, $j\in [1,n]$, $\e=0,1$, and $\e'=0,1$, and
\item $\ceqtmres[\Psi']{r_i=\e}{\dsubst{N^\e_i}{r}{y}}{M}{\td{A}{\psi}}$
for any $i\in [1,n]$ and $\e=0,1$,
\end{enumerate}
then
$\ceqtm[\Psi']{\hcom{\etc{r_i}}{\td{A}{\psi}}{r}{r}{M}{\etc{y.N^\e_i}}}{M}{\td{A}{\psi}}$.

\item 
For any $\psitd$, under the same conditions as above,
if $r_i = \e$ for some $i$ then
\[
\ceqtm[\Psi']{\hcomgeneric{\etc{r_i}}{\td{A}{\psi}}}{\dsubst{N^\e_i}{r'}{y}}{\td{A}{\psi}}.
\]

\item 
For any $\msubsts{(\Psi',x)}{\psi}{\Psi}$, if
$\ceqtm[\Psi']{M}{N}{\dsubst{\td{A}{\psi}}{r}{x}}$, then
$\ceqtm[\Psi']{{\coe{x.\td{A}{\psi}}{r}{r'}{M}}}{{\coe{x.\td{B}{\psi}}{r}{r'}{N}}}{\dsubst{\td{A}{\psi}}{r'}{x}}$.

\item 
For any $\msubsts{(\Psi',x)}{\psi}{\Psi}$, if
$\coftype[\Psi']{M}{\dsubst{\td{A}{\psi}}{r}{x}}$, then
$\ceqtm[\Psi']{{\coe{x.\td{A}{\psi}}{r}{r}{M}}}{M}{\dsubst{\td{A}{\psi}}{r}{x}}$.
\end{enumerate}
\end{definition}

\begin{definition}\label{def:ceqtype}
We say $\ceqtype{A}{B}$, presupposing $\ceqpretype{A}{B}$, when $A$ and $B$ are
cubical and equally Kan.
\end{definition}

The judgment $\ceqtype{A}{B}$ is symmetric and transitive because the Kan
conditions are.

\subsection{Open judgments}

We extend these judgments to open terms by functionality, that is, an open
pretype (resp., element of a pretype) is an open term that sends equal elements
of the pretypes in the context to equal closed pretypes (resp., elements). The
open judgments are defined simultaneously, stratified by the length of the
context. (We assume the variables $a_1,\dots,a_n$ in a context are distinct.)

\begin{definition}\label{def:wfctx}
We say $\wfctx{(\oft{a_1}{A_1},\dots,\oft{a_n}{A_n})}$ when
\begin{gather*}
\cpretype{A_1}, \\
\pretype{\oft{a_1}{A_1}}{A_2}, \dots \\ \text{and}~
\pretype{\oft{a_1}{A_1},\dots,\oft{a_{n-1}}{A_{n-1}}}{A_n}.
\end{gather*}
\end{definition}

\begin{definition}\label{def:eqpretype}
We say $\eqpretype{\oft{a_1}{A_1},\dots,\oft{a_n}{A_n}}{B}{B'}$,
presupposing \\
$\wfctx{(\oft{a_1}{A_1},\dots,\oft{a_n}{A_n})}$, when for any $\psitd$ and any
\begin{gather*}
\ceqtm[\Psi']{N_1}{N_1'}{\td{A_1}{\psi}}, \\
\ceqtm[\Psi']{N_2}{N_2'}{\subst{\td{A_2}{\psi}}{N_1}{a_1}}, \dots\\\text{and}~
\ceqtm[\Psi']{N_n}{N_n'}
{\subst{\td{A_n}{\psi}}{N_1,\dots,N_{n-1}}{a_1,\dots,a_n}},
\end{gather*}
we have
$\ceqpretype[\Psi']
{\subst{\td{B}{\psi}}{N_1,\dots,N_n}{a_1,\dots,a_n}}
{\subst{\td{B'}{\psi}}{N_1',\dots,N_n'}{a_1,\dots,a_n}}$.
\end{definition}

\begin{definition}\label{def:eqtm}
We say $\eqtm{\oft{a_1}{A_1},\dots,\oft{a_n}{A_n}}{M}{M'}{B}$,
presupposing \\
$\pretype{\oft{a_1}{A_1},\dots,\oft{a_n}{A_n}}{B}$,
when for any $\psitd$ and any
\begin{gather*}
\ceqtm[\Psi']{N_1}{N_1'}{\td{A_1}{\psi}}, \\
\ceqtm[\Psi']{N_2}{N_2'}{\subst{\td{A_2}{\psi}}{N_1}{a_1}}, \dots\\\text{and}~
\ceqtm[\Psi']{N_n}{N_n'}
{\subst{\td{A_n}{\psi}}{N_1,\dots,N_{n-1}}{a_1,\dots,a_n}},
\end{gather*}
we have
$\ceqtm[\Psi']
{\subst{\td{M}{\psi}}{N_1,\dots,N_n}{a_1,\dots,a_n}}
{\subst{\td{M'}{\psi}}{N_1',\dots,N_n'}{a_1,\dots,a_n}}
{\subst{\td{B}{\psi}}{N_1,\dots,N_n}{a_1,\dots,a_n}}$.
\end{definition}

One should read $[\Psi]$ as extending across the entire judgment, as it
specifies the starting dimension at which to consider not only $B$ and $M$ but
$\G$ as well. 
The open judgments, like the closed judgments, are symmetric and transitive.
In particular, if $\eqpretype{\G}{B}{B'}$ then $\pretype{\G}{B}$.
As a result, the earlier hypotheses of each definition ensure that later
hypotheses are sensible; for example,
$\wfctx{(\oft{a_1}{A_1},\dots,\oft{a_n}{A_n})}$ and
$\coftype[\Psi']{N_1}{\td{A_1}{\psi}}$ ensure that
$\cpretype[\Psi']{\subst{\td{A_2}{\psi}}{N_1}{a_1}}$.

Finally, we extend the notions of context restriction and of being a type to
open pretypes and elements in a straightforward fashion. (\Cref{def:wfctxres}
requires the open judgments to be closed under dimension substitution, which we
prove in \cref{lem:td-judgments}.)

\begin{definition}
\label{def:wfctxres}\label{def:eqpretyperes}\label{def:eqtmres}~
\begin{enumerate}
\item
We say $\wfctxres{\Xi}{\G}$, presupposing $\Xi$ is a set of equations in $\Psi$,
when for any $\psitd$ satisfying $\Xi$, $\wfctx[\Psi']{\td{\G}{\psi}}$.

\item
We say $\eqpretyperes{\Xi}{\G}{B}{B'}$, presupposing $\wfctxres{\Xi}{\G}$ and
$\Xi$ is a set of equations in $\Psi$, when for any $\psitd$ satisfying $\Xi$,
$\eqpretype[\Psi']{\td{\G}{\psi}}{\td{B}{\psi}}{\td{B'}{\psi}}$.

\item
We say $\eqtmres{\Xi}{\G}{M}{M'}{B}$, presupposing $\wfctxres{\Xi}{\G}$,
$\pretyperes{\Xi}{\G}{B}$, and $\Xi$ is a set of equations in $\Psi$, when for
any $\psitd$ satisfying $\Xi$,
$\eqtm[\Psi']{\td{\G}{\psi}}{\td{M}{\psi}}{\td{M'}{\psi}}{\td{B}{\psi}}$.
\end{enumerate}
\end{definition}

\begin{definition}\label{def:eqtype}
We say $\eqtype{\oft{a_1}{A_1},\dots,\oft{a_n}{A_n}}{B}{B'}$,
presupposing \\
$\eqpretype{\oft{a_1}{A_1},\dots,\oft{a_n}{A_n}}{B}{B'}$,
when for any $\psitd$ and any
\begin{gather*}
\ceqtm[\Psi']{N_1}{N_1'}{\td{A_1}{\psi}}, \\
\ceqtm[\Psi']{N_2}{N_2'}{\subst{\td{A_2}{\psi}}{N_1}{a_1}}, \dots\\\text{and}~
\ceqtm[\Psi']{N_n}{N_n'}
{\subst{\td{A_n}{\psi}}{N_1,\dots,N_{n-1}}{a_1,\dots,a_n}},
\end{gather*}
we have
$\ceqtype[\Psi']
{\subst{\td{B}{\psi}}{N_1,\dots,N_n}{a_1,\dots,a_n}}
{\subst{\td{B'}{\psi}}{N_1',\dots,N_n'}{a_1,\dots,a_n}}$.
\end{definition}

\subsection{Basic lemmas}
\label{ssec:lemmas}

We prove some basic results about our core judgments before proceeding.

\begin{lemma}[Head expansion]\label{lem:expansion}
If $\ceqtm{M'}{N}{A}$ and for all $\psitd$,
$\td{M}{\psi} \steps^* \td{M'}{\psi}$, then $\ceqtm{M}{N}{A}$.
\end{lemma}
\begin{proof}
For any $\msubsts{\Psi_1}{\psi_1}{\Psi}$ and $\msubsts{\Psi_2}{\psi_2}{\Psi_1}$,
we know $\vinperfour[\Psi_2]{M'_2}{M'_{12}}{N_2}{N_{12}}{A_{12}}$
where $\td{A}{\psi_1\psi_2} \evals A_{12}$. Therefore it suffices to show
$\vinper[\Psi_2]{M'_{12}}{M_{12}}{A_{12}}$ and
$\vinper[\Psi_2]{M'_2}{M_2}{A_{12}}$.
The former is true because
$\td{M}{\psi_1\psi_2} \steps^* \td{M'}{\psi_1\psi_2} \evals M'_{12}$ and
$\vinper[\Psi_2]{M'_{12}}{M'_{12}}{A_{12}}$.
The latter is true because
$\td{M}{\psi_1} \steps^* \td{M'}{\psi_1} \evals M'_1$,
$\td{M'_1}{\psi_2} \evals M'_2$, and
$\vinper[\Psi_2]{M'_2}{M'_2}{A_{12}}$.
\end{proof}

A special case of this lemma is that if $\coftype{M'}{A}$ then $\coftype{M}{A}$.

\begin{lemma}\label{lem:coftype-ceqtm}
If $\coftype{M}{A}$, $\coftype{N}{A}$, and for all $\psitd$,
$\inper[\Psi']{\td{M}{\psi}}{\td{N}{\psi}}{\td{A}{\psi}}$, then
$\ceqtm{M}{N}{A}$.
\end{lemma}
\begin{proof}
For all $\msubsts{\Psi_1}{\psi_1}{\Psi}$ and $\msubsts{\Psi_2}{\psi_2}{\Psi_1}$,
by $\coftype{M}{A}$ we have $\td{M}{\psi_1}\evals M_1$ and
$\inper[\Psi_2]{\td{M_1}{\psi_2}}{\td{M}{\psi_1\psi_2}}{\td{A}{\psi_1\psi_2}}$,
and by $\coftype{N}{A}$ we have $\td{N}{\psi_1}\evals N_1$ and
$\inper[\Psi_2]{\td{N_1}{\psi_2}}{\td{N}{\psi_1\psi_2}}{\td{A}{\psi_1\psi_2}}$.
Therefore it suffices to show
$\inper[\Psi_2]{\td{M}{\psi_1\psi_2}}{\td{N}{\psi_1\psi_2}}{\td{A}{\psi_1\psi_2}}$,
which
follows from our assumption at $\psi = \psi_1\psi_2$.
\end{proof}

\begin{lemma}\label{lem:td-judgments}
For any $\psitd$,
\begin{enumerate}
\item if $\ceqtm{M}{N}{A}$ then
$\ceqtm[\Psi']{\td{M}{\psi}}{\td{N}{\psi}}{\td{A}{\psi}}$;
\item if $\ceqtype{A}{B}$ then $\ceqtype[\Psi']{\td{A}{\psi}}{\td{B}{\psi}}$;
\item if $\wfctx{\G}$ then $\wfctx[\Psi']{\td{\G}{\psi}}$;
\item if $\eqpretype{\G}{A}{B}$ then 
$\eqpretype[\Psi']{\td{\G}{\psi}}{\td{A}{\psi}}{\td{B}{\psi}}$;
\item if $\eqtm{\G}{M}{N}{A}$ then
$\eqtm[\Psi']{\td{\G}{\psi}}{\td{M}{\psi}}{\td{N}{\psi}}{\td{A}{\psi}}$; and
\item if $\eqtype{\G}{A}{B}$ then
$\eqtype[\Psi']{\td{\G}{\psi}}{\td{A}{\psi}}{\td{B}{\psi}}$.
\end{enumerate}
\end{lemma}
\begin{proof}
We have already observed that if $\cpretype{A}$ then
$\cpretype[\Psi']{\td{A}{\psi}}$. Exact equality is closed under dimension
substitution for the same reason: its definition quantifies over all dimension
substitutions. Similarly, equal types are equally Kan cubical pretypes, and both
of these conditions are closed under dimension substitution.

Propositions (3), (4), and (5) are proven simultaneously by induction on the
length of $\G$. If $\G=\cdot$, then (3) is trivial, and (4) and (5) follow
because the closed judgments are closed under dimension substitution. The
inductive steps for all three use all three inductive hypotheses. Proposition
(6) follows similarly.
\end{proof}

The open judgments satisfy the \emph{structural rules} of type theory, like
hypothesis and weakening.

\begin{lemma}[Hypothesis]\label{lem:hypothesis}
If $\wfctx{(\G,\oft{a_i}{A_i},\G')}$ then
$\oftype{\G,\oft{a_i}{A_i},\G'}{a_i}{A_i}$.
\end{lemma}
\begin{proof}
We must show for any $\psitd$ and equal elements $N_1,N_1',\dots,N_n,N_n'$ of
the pretypes in $(\td{\G}{\psi},\oft{a_i}{\td{A_i}{\psi}},\td{\G'}{\psi})$, that
$\ceqtm[\Psi']{N_i}{N_i'}{\td{A_i}{\psi}}$. But this is exactly our assumption
about $N_i,N_i'$.
\end{proof}

\begin{lemma}[Weakening]\label{lem:weakening}~
\begin{enumerate}
\item If $\eqpretype{\G,\G'}{B}{B'}$ and $\pretype{\G}{A}$, then
$\eqpretype{\G,\oft{a}{A},\G'}{B}{B'}$.
\item If $\eqtm{\G,\G'}{M}{M'}{B}$ and $\pretype{\G}{A}$, then
$\eqtm{\G,\oft{a}{A},\G'}{M}{M'}{B}$.
\end{enumerate}
\end{lemma}
\begin{proof}
For the first part, we must show for any $\psitd$ and equal elements
\begin{gather*}
\ceqtm[\Psi']{N_1}{N_1'}{\td{A_1}{\psi}}, \\
\ceqtm[\Psi']{N_2}{N_2'}{\subst{\td{A_2}{\psi}}{N_1}{a_1}}, \dots \\
\ceqtm[\Psi']{N}{N'}{\subst{\td{A}{\psi}}{N_1,\dots}{a_1,\dots}},
\dots\\\text{and}~
\ceqtm[\Psi']{N_n}{N_n'}
{\subst{\td{A_n}{\psi}}{N_1,\dots,N,\dots,N_{n-1}}{a_1,\dots,a,\dots,a_n}},
\end{gather*}
that the corresponding instances of $B,B'$ are equal closed pretypes.
By $\eqpretype{\G,\G'}{B}{B'}$ we know that $a\fresh \G',B,B'$---since the
contained pretypes become closed when substituting for $a_1,\dots,a_n$.
It also gives us
$\ceqpretype[\Psi']
{\subst{\td{B}{\psi}}{N_1,\dots}{a_1,\dots}}
{\subst{\td{B'}{\psi}}{N_1',\dots}{a_1,\dots}}$
which are the desired instances of $B,B'$ because $a\fresh B,B'$.
The second part follows similarly.
\end{proof}

The definition of equal pretypes was chosen to ensure that equal pretypes have
equal elements.

\begin{lemma}\label{lem:ceqpretype-ceqtm}
If $\ceqpretype{A}{B}$ and $\ceqtm{M}{N}{A}$ then $\ceqtm{M}{N}{B}$.
\end{lemma}
\begin{proof}
For any $\msubsts{\Psi_1}{\psi_1}{\Psi}$ and $\msubsts{\Psi_2}{\psi_2}{\Psi_1}$,
by the first hypothesis we have that
$\td{A}{\psi_1\psi_2}\evals A_{12}$, $\td{B}{\psi_1\psi_2}\evals B_{12}$, and
$\eqper[\Psi_2]{A_{12}}{B_{12}}$;
by the second hypothesis, we have that
$\vinperfour[\Psi_2]{M_2}{M_{12}}{N_2}{N_{12}}{A_{12}}$.
But this implies
$\vinperfour[\Psi_2]{M_2}{M_{12}}{N_2}{N_{12}}{B_{12}}$.
\end{proof}

\begin{lemma}\label{lem:eqpretype-eqtm}
If $\eqpretype{\G}{A}{B}$ and $\eqtm{\G}{M}{N}{A}$ then $\eqtm{\G}{M}{N}{B}$.
\end{lemma}
\begin{proof}
If $\G = (\oft{a_1}{A_1},\dots,\oft{a_n}{A_n})$ then $\eqtm{\G}{M}{N}{A}$ means
that for any $\psitd$ and equal elements $N_1,N_1',\dots,N_n,N_n'$ of the
pretypes in $\td{\G}{\psi}$, the corresponding instances of $M$ and $N$ are
equal in $\subst{\td{A}{\psi}}{N_1,\dots,N_n}{a_1,\dots,a_n}$. But
$\eqpretype{\G}{A}{B}$ implies this pretype is equal to
$\subst{\td{B}{\psi}}{N_1,\dots,N_n}{a_1,\dots,a_n}$, so the result follows by
\cref{lem:ceqpretype-ceqtm}.
\end{proof}

The context-restricted judgments have many of the same properties as the
ordinary judgments.

\begin{lemma}\label{lem:td-judgres}
For any $\psitd$, if $\judgres{\Xi}{\J}$ then
$\judgres[\Psi']{\td{\Xi}{\psi}}{\td{\J}{\psi}}$, where $\J$ is any
dimension-context-restricted judgment.
\end{lemma}
\begin{proof}
We are given that for any $\psitd$ satisfying $\Xi$,
$\judg[\Psi']{\td{\J}{\psi}}$; and want to show that for any $\psitd$ and any
$\msubsts{\Psi''}{\psi'}{\Psi'}$ satisfying $\td{\Xi}{\psi}$, that
$\judg[\Psi'']{\td{\J}{\psi\psi'}}$. It suffices to show that if
$\psi'$ satisfies $\td{\Xi}{\psi}$, then $\psi\psi'$ satisfies $\Xi$.
But these are both true if and only if for each equation $r_i=r_i'$ in $\Xi$,
$\td{r_i}{\psi\psi'} = \td{r_i'}{\psi\psi'}$.
\end{proof}

\begin{lemma}\label{lem:judg-judgres}
If $\judg{\J}$ then $\judgres{\Xi}{\J}$.
\end{lemma}
\begin{proof}
By \cref{lem:td-judgments}, we know that $\judg{\J}$ implies that for any
$\psitd$, $\judg[\Psi']{\td{\J}{\psi}}$. Therefore, for any $\psi$ satisfying
$\Xi$, $\judg[\Psi']{\td{\J}{\psi}}$.
\end{proof}

\begin{remark}\label{rem:judgres-dsubst}
Although we define $\judgres{\Xi}{\J}$ for general $\Xi$, \cref{def:kan} only
uses $18$ distinct $\Xi$ (given that we consider them modulo permutation and
duplication). Moreover, this class of $\Xi$ are closed under dimension
substitution. They fall into three categories:

\begin{enumerate}
\item Three $\Xi$ are satisfied by \emph{all} $\psi$, for example $(0=0,1=1)$.
In these cases, $\judgres{\Xi}{\J}$ if and only if $\judg[\Psi']{\td{\J}{\psi}}$
for all $\psi$, which by \cref{lem:td-judgments} holds if and only if
$\judg{\J}$.

\item Six are satisfied by \emph{no} $\psi$, for example $(0=x,1=x)$.
In these cases, $\judgres{\Xi}{\J}$ always.

\item The remaining nine can be reduced to a substitution instance of $\J$,
because any $\psi$ satisfying $\Xi$ can be factored through a one- or
two-variable dimension substitution.
\begin{enumerate}
\item For $(\e=x,\e'=\e')$ and $(\e=x)$, $\judgres[\Psi,x]{\Xi}{\J}$ if and only
if $\judg{\dsubst{\J}{\e}{x}}$.

\item For $(\e=x,\e'=y)$, $\judgres[\Psi,x,y]{\Xi}{\J}$ if and only if
$\judg{\dsubst{\dsubst{\J}{\e}{x}}{\e'}{y}}$.
\end{enumerate}
\end{enumerate}

Therefore one can think of the context-restricted judgments as merely a
notational device for avoiding the above case-split when expressing the Kan
condition.
\end{remark}

In \cref{def:kan} we defined closed Kan pretypes as ones with closed $\hcomsym$
and $\coesym$ elements; in \cref{def:eqtype} we defined open types as open
pretypes whose instances are all cubical and Kan. Because dimension and term
substitutions commute with $\hcomsym$ and $\coesym$, it follows that open types
have open $\hcomsym$ elements. (Open versions of the four other Kan conditions
also hold for open types, but we do not state them here.)

\begin{lemma}\label{lem:kan-open}
If $\eqtype{\G}{A}{B}$ then if
\begin{enumerate}
\item $\eqtm{\G}{M}{O}{A}$,
\item $\eqtmres[\Psi,y]{r_i=\e,r_j=\e'}{\G}{N^\e_i}{N^{\e'}_j}{A}$
for any $i\in [1,n]$, $j\in [1,n]$, $\e=0,1$, and $\e'=0,1$,
\item $\eqtmres[\Psi,y]{r_i=\e}{\G}{N^\e_i}{P^\e_i}{A}$
for any $i\in [1,n]$ and $\e=0,1$, and
\item $\eqtmres{r_i=\e}{\G}{\dsubst{N^\e_i}{r}{y}}{M}{A}$
for any $i\in [1,n]$ and $\e=0,1$,
\end{enumerate}
then $\eqtm{\G}{\hcomgeneric{\etc{r_i}}{A}}%
{\hcom{\etc{r_i}}{B}{r}{r'}{O}{\etc{y.P^\e_i}}}{A}$.
\end{lemma}
\begin{proof}
Let $\G = (\oft{a_1}{A_1},\dots,\oft{a_n}{A_n})$.
We need to show that for any $\psitd$ and
$\ceqtm[\Psi']{N_1}{N_1'}{\td{A_1}{\psi}}$, \dots, we have
\begin{small}\[
\def\sub#1{\subst{\td{#1}{\psi}}{N_1,\dots}{a_1,\dots}}
\ceqtm[\Psi']
{\hcom{\etc{\td{r_i}{\psi}}}{\sub{A}}{\td{r}{\psi}}{\td{r'}{\psi}}{\sub{M}}{\etc{y.\sub{N^\e_i}}}}
{\dots}{\sub{A}}.
\]\end{small}%
We prove this using the first Kan condition of the corresponding closed
instances of $A,B$; the only difficulty is showing that the open
context-restricted hypotheses of this lemma imply the necessary closed
context-restricted equalities.

Consider $\eqtmres[\Psi,y]{r_i=\e}{\G}{N^\e_i}{P^\e_i}{A}$. By
\cref{lem:td-judgres} we know
$\eqtmres[\Psi',y]{\td{r_i}{\psi}=\e}{\td{\G}{\psi}}
{\td{N^\e_i}{\psi}}{\td{P^\e_i}{\psi}}{\td{A}{\psi}}$.
Therefore, for any $\msubsts{\Psi''}{\psi'}{(\Psi',y)}$ satisfying
$\td{r_i}{\psi}=\e$,
$\eqtm[\Psi'']{\td{\G}{\psi\psi'}}
{\td{N^\e_i}{\psi\psi'}}{\td{P^\e_i}{\psi\psi'}}{\td{A}{\psi\psi'}}$.
By \cref{lem:td-judgments} we know that
$\ceqtm[\Psi'']{\td{N_1}{\psi'}}{\td{N_1'}{\psi'}}{\td{A_1}{\psi\psi'}}$, \dots,
so we get
\[
\ceqtm[\Psi'']
{\subst{\td{N^\e_i}{\psi\psi'}}{\td{N_1}{\psi'},\dots}{a_1,\dots}}
{\subst{\td{P^\e_i}{\psi\psi'}}{\td{N_1'}{\psi'},\dots}{a_1,\dots}}
{\subst{\td{A}{\psi\psi'}}{\td{N_1}{\psi'},\dots}{a_1,\dots}}
\]
or, by commuting substitutions,
\[
\ceqtm[\Psi'']
{\td{\subst{\td{N^\e_i}{\psi}}{N_1,\dots}{a_1,\dots}}{\psi'}}
{\td{\subst{\td{P^\e_i}{\psi}}{N_1',\dots}{a_1,\dots}}{\psi'}}
{\td{\subst{\td{A}{\psi}}{N_1,\dots}{a_1,\dots}}{\psi'}}.
\]
Since this holds for all $\psi'$ satisfying $\td{r_i}{\psi} = \e$, it implies
\[
\ceqtmres[\Psi',y]{\td{r_i}{\psi} = \e}
{\subst{\td{N^\e_i}{\psi}}{N_1,\dots}{a_1,\dots}}
{\subst{\td{P^\e_i}{\psi}}{N_1',\dots}{a_1,\dots}}
{\subst{\td{A}{\psi}}{N_1,\dots}{a_1,\dots}}
\]
which is exactly what we needed. The other hypotheses similarly follow.
\end{proof}

Finally, while the Kan conditions only directly define \emph{homogeneous}
composition, in the sense that the type $A$ must be degenerate in the bound
direction of the tubes, we can combine homogeneous composition and coercion to
obtain \emph{heterogeneous} composition, written $\comsym$ and defined:
\[
\comgeneric{\etc{r_i}}{y.A} :=
\hcom{\etc{r_i}}{\dsubst{A}{r'}{y}}{r}{r'}
  {\coe{y.A}{r}{r'}{M}}
  {\etc{y.\coe{y.A}{y}{r'}{N^\e_i}}}
\]
which satisfies the following properties.

\begin{theorem}\label{thm:com}
If $\ceqtype{A}{B}$, then:
\begin{enumerate}
\item 
For any $\msubsts{(\Psi',y)}{\psi}{\Psi}$, if
\begin{enumerate}
\item $\ceqtm[\Psi']{M}{O}{\dsubst{\td{A}{\psi}}{r}{y}}$,
\item $\ceqtmres[\Psi',y]{r_i=\e,r_j=\e'}{N^\e_i}{N^{\e'}_j}{\td{A}{\psi}}$
for any $i\in [1,n]$, $j\in [1,n]$, $\e=0,1$, and $\e'=0,1$,
\item $\ceqtmres[\Psi',y]{r_i=\e}{N^\e_i}{P^\e_i}{\td{A}{\psi}}$
for any $i\in [1,n]$ and $\e=0,1$, and
\item $\ceqtmres[\Psi']{r_i=\e}{\dsubst{N^\e_i}{r}{y}}{M}
{\dsubst{\td{A}{\psi}}{r}{y}}$
for any $i\in [1,n]$ and $\e=0,1$,
\end{enumerate}
then $\ceqtm[\Psi']{\comgeneric{\etc{r_i}}{y.\td{A}{\psi}}}%
{\com{\etc{r_i}}{y.\td{B}{\psi}}{r}{r'}{O}{\etc{y.P^\e_i}}}
{\dsubst{\td{A}{\psi}}{r'}{y}}$.

\item 
For any $\msubsts{(\Psi',y)}{\psi}{\Psi}$, if
\begin{enumerate}
\item $\coftype[\Psi']{M}{\dsubst{\td{A}{\psi}}{r}{y}}$,
\item $\ceqtmres[\Psi',y]{r_i=\e,r_j=\e'}{N^\e_i}{N^{\e'}_j}{\td{A}{\psi}}$
for any $i\in [1,n]$, $j\in [1,n]$, $\e=0,1$, and $\e'=0,1$,
\item $\ceqtmres[\Psi']{r_i=\e}{\dsubst{N^\e_i}{r}{y}}{M}
{\dsubst{\td{A}{\psi}}{r}{y}}$
for any $i\in [1,n]$ and $\e=0,1$,
\end{enumerate}
then
$\ceqtm[\Psi']{\com{\etc{r_i}}{y.\td{A}{\psi}}{r}{r}{M}{\etc{y.N^\e_i}}}{M}
{\dsubst{\td{A}{\psi}}{r}{y}}$.

\item 
For any $\msubsts{(\Psi',y)}{\psi}{\Psi}$,
under the same conditions as above, if $r_i = \e$ for some $i$ then
\[
\ceqtm[\Psi']{\comgeneric{\etc{r_i}}{\td{A}{\psi}}}{\dsubst{N^\e_i}{r'}{y}}{\dsubst{\td{A}{\psi}}{r'}{y}}.
\]
\end{enumerate}
\end{theorem}
\begin{proof}
By the fourth Kan condition, the conditions in (1) above imply that for all
$i,j,\e,\e'$,
\begin{enumerate}
\item $\ceqtm[\Psi']
{\coe{y.\td{A}{\psi}}{r}{r'}{M}}
{\coe{y.\td{B}{\psi}}{r}{r'}{O}}
{\dsubst{\td{A}{\psi}}{r'}{y}}$,
\item $\ceqtmres[\Psi',y]{r_i=\e,r_j=\e'}
{\coe{y.\td{A}{\psi}}{y}{r'}{N^\e_i}}
{\coe{y.\td{A}{\psi}}{y}{r'}{N^{\e'}_j}}
{\dsubst{\td{A}{\psi}}{r'}{y}}$,
\item $\ceqtmres[\Psi',y]{r_i=\e}
{\coe{y.\td{A}{\psi}}{y}{r'}{N^\e_i}}
{\coe{y.\td{B}{\psi}}{y}{r'}{P^\e_i}}
{\dsubst{\td{A}{\psi}}{r'}{y}}$, and
\item $\ceqtmres[\Psi']{r_i=\e}
{\dsubst{(\coe{y.\td{A}{\psi}}{y}{r'}{N^\e_i})}{r}{y}}
{\coe{y.\td{A}{\psi}}{r}{r'}{M}}
{\dsubst{\td{A}{\psi}}{r'}{y}}$.
\end{enumerate}
To prove this for the context-restricted judgments above, we use the fact that
$\ceqtm[\Psi'']{\td{N^\e_i}{\psi'}}{\td{N^{\e'}_j}{\psi'}}{\td{A}{\psi\psi'}}$
for all $\msubsts{\Psi''}{\psi'}{(\Psi',y)}$ satisfying $(r_i=\e,r_j=\e')$, and
so
\[
\ceqtm[\Psi'']
{\coe{y.\td{A}{\psi\psi'}}{\td{y}{\psi'}}{\td{r'}{\psi'}}{\td{N^\e_i}{\psi'}}}
{\coe{y.\td{A}{\psi\psi'}}{\td{y}{\psi'}}{\td{r'}{\psi'}}{\td{N^{\e'}_j}{\psi'}}}
{\dsubst{\td{A}{\psi}}{r'}{y}}
\]
which implies the tube adjacency condition.
The first Kan condition of $A,B$ therefore gives us the first condition above.

For the second condition above, the second Kan condition of $A$ gives us 
\[
\ceqtm[\Psi']{\com{\etc{r_i}}{y.\td{A}{\psi}}{r}{r}{M}{\etc{y.N^\e_i}}}
{\coe{y.\td{A}{\psi}}{r}{r}{M}}
{\dsubst{\td{A}{\psi}}{r}{y}}
\]
and by the fifth Kan condition,
$\ceqtm[\Psi']{\coe{y.\td{A}{\psi}}{r}{r}{M}}{M}{\dsubst{\td{A}{\psi}}{r}{y}}$.

For the third condition, the third Kan condition of $A$ gives us
\[
\ceqtm[\Psi']{\comgeneric{\etc{r_i}}{\td{A}{\psi}}}
{\dsubst{(\coe{y.\td{A}{\psi}}{y}{r'}{N^\e_i})}{r'}{y}}
{\dsubst{\td{A}{\psi}}{r'}{y}}
\]
and again by the fifth Kan condition,
$\ceqtm[\Psi']
{\dsubst{(\coe{y.\td{A}{\psi}}{y}{r'}{N^\e_i})}{r'}{y}}
{\dsubst{N^\e_i}{r'}{y}}{\dsubst{\td{A}{\psi}}{r'}{y}}$.
\end{proof}

An open version of this theorem also holds for open types.


\newpage
\section{Types}
\label{sec:types}

In \cref{sec:meanings} we explained how a cubical type system gives rise to the
judgments of higher type theory. However, our interest is in type systems with
(higher) inductive types, dependent functions and pairs, identification
types, and so forth. In this section we will explain what it means for a cubical
type system to have certain type formers, just as in category theory one
explains what it means for a category to have, say, finite products.

For each type former, we will then prove that any cubical type system with that
property validates the expected typing rules. For example, any cubical type
system with dependent pairs will validate the usual formation, introduction,
elimination, computation, and eta rules for that pretype, and moreover, the
dependent pair pretypes will be cubical and Kan.

Finally, we will want to exhibit a cubical type system with all our desired type
formers, and use it as a model for the rules in \cref{sec:proof-theory}. A
straightforward method is to produce the smallest cubical type system closed
under the type formers, by means of a fixed point construction
\citep{allen1987types,harper1992typesys}. This is possible because the meanings
of dependent function, pair, and identification types are parasitic on the
meanings of their constituent types.

Note, however, that \emph{any} such cubical type system suffices for our
purposes; none of our theorems hold only in the least such. It should therefore
be possible to extend our results with additional type formers (such as type
universes, or more higher inductive types) without needing to reprove
everything.

\subsection{Booleans}

We consider $\bool$ as a higher inductive type, meaning that we freely add Kan
composites as higher cells, rather than specifying that all its higher cells are
exactly $\true$ or $\false$. We do this to demonstrate the robustness of our
canonicity theorem and our treatment of $\notb{x}$; in \cref{appendix}, we
define a type of ``strict booleans'' which may be more useful in practice.

A cubical type system \emph{has booleans} if $\veqper{\bool}{\bool}$ for all
$\Psi$, and $\vinper[-]{-}{-}{\bool}$ is the least relation such that:
\begin{enumerate}
\item $\vinper{\true}{\true}{\bool}$,
\item $\vinper{\false}{\false}{\bool}$, and
\item $\vinper{\hcomgeneric{\etc{x_i}}{\bool}}{\hcom{\etc{x_i}}{\bool}{r}{r'}{O}{\etc{y.P^\e_i}}}{\bool}$ whenever $r\neq r'$,
\begin{enumerate}
\item $\ceqtm{M}{O}{\bool}$,
\item $\ceqtmres[\Psi,y]{x_i=\e,x_j=\e'}{N^\e_i}{N^{\e'}_j}{\bool}$
for all $i,j,\e,\e'$,
\item $\ceqtmres[\Psi,y]{x_i=\e}{N^\e_i}{P^\e_i}{\bool}$
for all $i,\e$, and
\item $\ceqtmres{x_i=\e}{\dsubst{N^\e_i}{r}{y}}{M}{\bool}$
for all $i,\e$.
\end{enumerate}
\end{enumerate}

Note that this relation is symmetric because (a) and (c) imply that (b) and (d)
also hold for $P^\e_i$ and $O$.
In this definition, each $\vinper{-}{-}{\bool}$ refers to all other
$\vinper[\Psi']{-}{-}{\bool}$, because each equality judgment does.
In the remainder of this subsection, we prove theorems about cubical type
systems that have booleans.

\paragraph{Pretype}
$\cpretype{\bool}$.

For all $\psi_1,\psi_2$,
$\td{\bool}{\psi_1}\evals\bool$,
$\td{\bool}{\psi_2}\evals\bool$,
$\td{\bool}{\psi_1\psi_2}\evals\bool$,
and $\veqper[\Psi_2]{\bool}{\bool}$.

\paragraph{Introduction}
$\coftype{\true}{\bool}$ and $\coftype{\false}{\bool}$.

For all $\msubsts{\Psi_1}{\psi_1}{\Psi}$
and $\msubsts{\Psi_2}{\psi_2}{\Psi_1}$,
$\td{\true}{\psi_1}\evals\true$,
$\td{\true}{\psi_2}\evals\true$,
$\td{\true}{\psi_1\psi_2}\evals\true$,
and $\vinper[\Psi_2]{\true}{\true}{\bool}$.
The $\false$ case is analogous.

\paragraph{Kan}
$\cpretype{\bool}$ is Kan.

We prove only the unary version of the first Kan condition, to lessen the
notational burden; the binary version follows easily. Show that for any $\Psi'$,
if
\begin{enumerate}
\item $\coftype[\Psi']{M}{\bool}$,
\item $\ceqtmres[\Psi',y]{r_i=\e,r_j=\e'}{N^\e_i}{N^{\e'}_j}{\bool}$
for all $i,j,\e,\e'$, and
\item $\ceqtmres[\Psi']{r_i=\e}{\dsubst{N^\e_i}{r}{y}}{M}{\bool}$
for all $i,\e$,
\end{enumerate}
then $\coftype[\Psi']{\hcomgeneric{\etc{r_i}}{\bool}}{\bool}$.
That is, for any $\msubsts{\Psi_1}{\psi_1}{\Psi'}$ and
$\msubsts{\Psi_2}{\psi_2}{\Psi_1}$, $\td{\hcomsym}{\psi_1}\evals H_1$ and
$\inper[\Psi_2]{\td{H_1}{\psi_2}}{\td{\hcomsym}{\psi_1\psi_2}}{\bool}$. We
proceed by case-analyzing $r$, $r'$, and $\etc{r_i}$ under $\psi_1$ and
$\psi_1\psi_2$ to determine how $\td{\hcomsym}{\psi_1}$ and
$\td{\hcomsym}{\psi_1\psi_2}$ step.

\begin{enumerate}
\item $\td{r_i}{\psi_1} = \e$
(where $\td{r_k}{\psi_1}$ is a dimension name for all $k < i$), and
$\td{r_j}{\psi_1\psi_2} = \e'$ (where this is again the smallest such $j$).

Then $\td{\hcomsym}{\psi_1} \steps
\dsubst{\td{N^\e_i}{\psi_1}}{\td{r'}{\psi_1}}{y} =
\td{\dsubst{N^\e_i}{r'}{y}}{\psi_1}$ and
$\td{\hcomsym}{\psi_1\psi_2} \steps
\dsubst{\td{N^{\e'}_j}{\psi_1\psi_2}}{\td{r'}{\psi_1\psi_2}}{y} =
\td{\dsubst{N^{\e'}_j}{r'}{y}}{\psi_1\psi_2}$.
We know $\coftyperes[\Psi',y]{r_i=\e}{N^\e_i}{\bool}$ and
$\td{\dsubst{r_i}{r'}{y}}{\psi_1}=\e$, so
$\coftype[\Psi_1]{\td{\dsubst{N^\e_i}{r'}{y}}{\psi_1}}{\bool}$ and thus
$\td{\dsubst{N^\e_i}{r'}{y}}{\psi_1} \evals H_1$ and
$\inper[\Psi_2]{\td{H_1}{\psi_2}}{\td{\dsubst{N^\e_i}{r'}{y}}{\psi_1\psi_2}}{\bool}$.

We also know that
$\ceqtmres[\Psi',y]{r_i=\e,r_j=\e'}{N^\e_i}{N^{\e'}_j}{\bool}$,
$\td{\dsubst{r_i}{r'}{y}}{\psi_1\psi_2}=\e$, and
$\td{\dsubst{r_j}{r'}{y}}{\psi_1\psi_2}=\e'$, so
$\ceqtm[\Psi_2]{\td{\dsubst{N^\e_i}{r'}{y}}{\psi_1\psi_2}}
{\td{\dsubst{N^{\e'}_j}{r'}{y}}{\psi_1\psi_2}}{\bool}$, and thus
$\inper[\Psi_2]{\td{\dsubst{N^\e_i}{r'}{y}}{\psi_1\psi_2}}
{\td{\dsubst{N^{\e'}_j}{r'}{y}}{\psi_1\psi_2}}{\bool}$.
The result follows by transitivity.

\item All $\td{r_i}{\psi_1}$ are dimension names,
$\td{r}{\psi_1} = \td{r'}{\psi_1}$, and
$\td{r_i}{\psi_1\psi_2} = \e$ (where this is the smallest $i$).

Then $\td{\hcomsym}{\psi_1} \steps \td{M}{\psi_1}$ and
$\td{\hcomsym}{\psi_1\psi_2} \steps
\dsubst{\td{N^\e_i}{\psi_1\psi_2}}{\td{r'}{\psi_1\psi_2}}{y} =
\dsubst{\td{N^\e_i}{\psi_1\psi_2}}{\td{r}{\psi_1\psi_2}}{y} =
\td{\dsubst{N^\e_i}{r}{y}}{\psi_1\psi_2}$.
We know $\coftype[\Psi']{M}{\bool}$ so $\td{M}{\psi_1} \evals H_1$ and
$\inper[\Psi_2]{\td{H_1}{\psi_2}}{\td{M}{\psi_1\psi_2}}{\bool}$.
We also know that
$\ceqtmres[\Psi']{r_i=\e}{\dsubst{N^\e_i}{r}{y}}{M}{\bool}$ and
$\td{r_i}{\psi_1\psi_2}=\e$, so
$\ceqtm[\Psi_2]{\td{\dsubst{N^\e_i}{r}{y}}{\psi_1\psi_2}}
{\td{M}{\psi_1\psi_2}}{\bool}$, and thus
$\inper[\Psi_2]{\td{\dsubst{N^\e_i}{r}{y}}{\psi_1\psi_2}}
{\td{M}{\psi_1\psi_2}}{\bool}$.
The result again follows by transitivity.

\item All $\td{r_i}{\psi_1}$ are dimension names,
$\td{r}{\psi_1} = \td{r'}{\psi_1}$, and
all $\td{r_i}{\psi_1\psi_2}$ are dimension names.

Then $\td{\hcomsym}{\psi_1} \steps \td{M}{\psi_1}$ and
$\td{\hcomsym}{\psi_1\psi_2} \steps \td{M}{\psi_1\psi_2}$
(since $\td{r}{\psi_1\psi_2} = \td{r'}{\psi_1\psi_2}$).
By $\coftype[\Psi']{M}{\bool}$, $\td{M}{\psi_1} \evals H_1$ and
$\inper[\Psi_2]{\td{H_1}{\psi_2}}{\td{M}{\psi_1\psi_2}}{\bool}$.

\item All $\td{r_i}{\psi_1}$ are dimension names,
$\td{r}{\psi_1} \neq \td{r'}{\psi_1}$, and
$\td{r_i}{\psi_1\psi_2} = \e$ (where this is the smallest $i$).

Then $\isval{\td{\hcomsym}{\psi_1}}$ and
$\td{\hcomsym}{\psi_1\psi_2} \steps
\dsubst{\td{N^\e_i}{\psi_1\psi_2}}{\td{r'}{\psi_1\psi_2}}{y} =
\td{\dsubst{N^\e_i}{r'}{y}}{\psi_1\psi_2}$.
In this case $H_1 = \td{\hcomsym}{\psi_1}$, so $\td{H_1}{\psi_2} =
\td{\hcomsym}{\psi_1\psi_2}$ and we must show
$\inper[\Psi_2]{\td{\hcomsym}{\psi_1\psi_2}}
{\td{\hcomsym}{\psi_1\psi_2}}{\bool}$.
We know $\coftyperes[\Psi',y]{r_i=\e}{N^\e_i}{\bool}$ and
$\td{\dsubst{r_i}{r'}{y}}{\psi_1\psi_2}=\e$, so
$\coftype[\Psi_2]{\td{\dsubst{N^\e_i}{r'}{y}}{\psi_1\psi_2}}{\bool}$ and thus
$\inper[\Psi_2]{\td{\dsubst{N^\e_i}{r'}{y}}{\psi_1\psi_2}}
{\td{\dsubst{N^\e_i}{r'}{y}}{\psi_1\psi_2}}{\bool}$.

\item All $\td{r_i}{\psi_1}$ are dimension names,
$\td{r}{\psi_1} \neq \td{r'}{\psi_1}$,
all $\td{r_i}{\psi_1\psi_2}$ are dimension names, and
$\td{r}{\psi_1\psi_2} = \td{r'}{\psi_1\psi_2}$.

Then $\isval{\td{\hcomsym}{\psi_1}}$ and
$\td{\hcomsym}{\psi_1\psi_2} \steps \td{M}{\psi_1\psi_2}$.
By $\coftype[\Psi']{M}{\bool}$,
$\inper[\Psi_2]{\td{M}{\psi_1\psi_2}}{\td{M}{\psi_1\psi_2}}{\bool}$.

\item All $\td{r_i}{\psi_1}$ are dimension names,
$\td{r}{\psi_1} \neq \td{r'}{\psi_1}$,
all $\td{r_i}{\psi_1\psi_2}$ are dimension names, and
$\td{r}{\psi_1\psi_2} \neq \td{r'}{\psi_1\psi_2}$.

Then $\isval{\td{\hcomsym}{\psi_1}}$,
$\isval{\td{\hcomsym}{\psi_1\psi_2}}$, and
by \cref{lem:td-judgments,lem:td-judgres} we know
$\coftype[\Psi_2]{\td{M}{\psi_1\psi_2}}{\bool}$,
$\ceqtmres[\Psi_2,y]{\td{r_i}{\psi_1\psi_2}=\e,\td{r_j}{\psi_1\psi_2}=\e'}
{\td{N^\e_i}{\psi_1\psi_2}}{\td{N^{\e'}_j}{\psi_1\psi_2}}{\bool}$
for all $i,j,\e,\e'$, and
$\ceqtmres[\Psi_2]{\td{r_i}{\psi_1\psi_2}=\e}
{\dsubst{\td{N^\e_i}{\psi_1\psi_2}}{\td{r}{\psi_1\psi_2}}{y}}
{\td{M}{\psi_1\psi_2}}{\bool}$
for all $i,\e$.
Then because $\td{r_i}{\psi_1\psi_2}$ are all dimension names and
$\td{r}{\psi_1\psi_2} \neq \td{r'}{\psi_1\psi_2}$, the cubical type system
having booleans directly implies
$\vinper[\Psi_2]{\td{\hcomsym}{\psi_1\psi_2}}{\td{\hcomsym}{\psi_1\psi_2}}{\bool}$.
\end{enumerate}

The second Kan condition requires that for any $\Psi'$, if
\begin{enumerate}
\item $\coftype[\Psi']{M}{\bool}$,
\item $\ceqtmres[\Psi',y]{r_i=\e,r_j=\e'}{N^\e_i}{N^{\e'}_j}{\bool}$
for all $i,j,\e,\e'$, and
\item $\ceqtmres[\Psi']{r_i=\e}{\dsubst{N^\e_i}{r}{y}}{M}{\bool}$
for all $i,\e$,
\end{enumerate}
then
$\ceqtm[\Psi']{\hcom{\etc{r_i}}{\bool}{r}{r}{M}{\etc{y.N^\e_i}}}{M}{\bool}$.
That is, if $r = r'$ then $\hcomsym$ is equal to its cap $M$.

By the first Kan condition, we already know that both sides are
$\coftype[\Psi']{-}{\bool}$. Thus by \cref{lem:coftype-ceqtm} it suffices to
show that for any $\msubsts{\Psi''}{\psi}{\Psi'}$,
$\inper[\Psi'']{\td{\hcomsym}{\psi}}{\td{M}{\psi}}{\bool}$.
There are two cases:

\begin{enumerate}
\item $\td{r_i}{\psi} = \e$ (where this is the smallest such $i$).

Then $\td{\hcomsym}{\psi} \steps
\dsubst{\td{N^\e_i}{\psi}}{\td{r}{\psi}}{y} =
\td{\dsubst{N^\e_i}{r}{y}}{\psi}$.
We know $\ceqtmres[\Psi']{r_i=\e}{\dsubst{N^\e_i}{r}{y}}{M}{\bool}$ and
$\td{r_i}{\psi}=\e$, so
$\ceqtm[\Psi'']{\td{\dsubst{N^\e_i}{r}{y}}{\psi}}
{\td{M}{\psi}}{\bool}$, and thus
$\inper[\Psi'']{\td{\dsubst{N^\e_i}{r}{y}}{\psi}}
{\td{M}{\psi}}{\bool}$.

\item All $\td{r_i}{\psi}$ are dimension names.

Then $\td{\hcomsym}{\psi} \steps \td{M}{\psi}$.
By $\coftype[\Psi']{M}{\bool}$,
$\inper[\Psi'']{\td{M}{\psi}}{\td{M}{\psi}}{\bool}$.
\end{enumerate}

The third Kan condition requires that for any $\Psi'$,
if $r_i = \e$ for some $i$,
\begin{enumerate}
\item $\coftype[\Psi']{M}{\bool}$,
\item $\ceqtmres[\Psi',y]{r_i=\e,r_j=\e'}{N^\e_i}{N^{\e'}_j}{\bool}$
for all $i,j,\e,\e'$, and
\item $\ceqtmres[\Psi']{r_i=\e}{\dsubst{N^\e_i}{r}{y}}{M}{\bool}$
for all $i,\e$,
\end{enumerate}
then
$\ceqtm[\Psi']{\hcomgeneric{\etc{r_i}}{\bool}}{\dsubst{N^\e_i}{r'}{y}}{\bool}$.

By the first Kan condition, we know that
$\coftype[\Psi']{\hcomgeneric{\etc{r_i}}{\bool}}{\bool}$.
By hypothesis, $\coftyperes[\Psi',y]{r_i=\e}{N^\e_i}{\bool}$, but $r_i = \e$ so
in fact $\coftype[\Psi',y]{N^\e_i}{\bool}$.
Thus \cref{lem:coftype-ceqtm} applies, and it suffices to show that for any
$\msubsts{\Psi''}{\psi}{\Psi'}$,
$\inper[\Psi'']{\td{\hcomsym}{\psi}}{\td{\dsubst{N^\e_i}{r'}{y}}{\psi}}{\bool}$.
Let $j$ be the smallest index such that $\td{r_j}{\psi} = \e'$. We know
$\ceqtmres[\Psi',y]{r_i=\e,r_j=\e'}{N^\e_i}{N^{\e'}_j}{\bool}$,
$\td{\dsubst{r_i}{r'}{y}}{\psi}=\e$, and $\td{\dsubst{r_j}{r'}{y}}{\psi}=\e'$,
so $\ceqtm[\Psi'']{\td{\dsubst{N^\e_i}{r'}{y}}{\psi}}
{\td{\dsubst{N^{\e'}_j}{r'}{y}}{\psi}}{\bool}$, and thus
$\inper[\Psi'']{\td{\dsubst{N^\e_i}{r'}{y}}{\psi}}
{\td{\dsubst{N^{\e'}_j}{r'}{y}}{\psi}}{\bool}$.

The fourth Kan condition requires that for any $\Psi'$, if
$\ceqtm[\Psi']{M}{N}{\bool}$, then
$\ceqtm[\Psi']{{\coe{x.\bool}{r}{r'}{M}}}{{\coe{x.\bool}{r}{r'}{N}}}{\bool}$.
But $\coe{x.\bool}{\td{r}{\psi}}{\td{r'}{\psi}}{\td{M}{\psi}} \steps
\td{M}{\psi}$ for all $\psi$,
so by \cref{lem:expansion},
$\coe{x.\bool}{r}{r'}{M} \eq M \eq N \eq \coe{x.\bool}{r}{r'}{N}$.

The fifth Kan condition requires that for any $\Psi'$, if
$\coftype[\Psi']{M}{\bool}$, then
$\ceqtm[\Psi']{{\coe{x.\bool}{r}{r}{M}}}{M}{\bool}$.
This again follows immediately by \cref{lem:expansion}.

\paragraph{Cubical}
For any $\Psi'$ and $\vinper[\Psi']{M}{N}{\bool}$, $\ceqtm[\Psi']{M}{N}{\bool}$.

There are three cases in which $\vinper[\Psi']{M}{N}{\bool}$.
For $\true$ and $\false$, this follows from the introduction rules already
proven. For $\vinper[\Psi']{\hcomgeneric{\etc{x_i}}{\bool}}
{\hcom{\etc{x_i}}{\bool}{r}{r'}{O}{\etc{y.P^\e_i}}}{\bool}$,
this follows by the first Kan condition of $\bool$, again already proven.

\paragraph{Elimination}
If $\ceqtm{M}{M'}{\bool}$,
$\eqtype{\oft{a}{\bool}}{A}{A'}$,
$\ceqtm{T}{T'}{\subst{A}{\true}{a}}$, and
$\ceqtm{F}{F'}{\subst{A}{\false}{a}}$, then
$\ceqtm{\ifb{a.A}{M}{T}{F}}{\ifb{a.A'}{M'}{T'}{F'}}{\subst{A}{M}{a}}$.

The elimination rule is complicated to prove, because it requires that
$\ifsym$ applied to a value $\hcomsym$ in $\bool$ has coherent aspects, even
though that $\hcomsym$'s aspects may be a different $\hcomsym$ or a different
term altogether. In \cref{lem:if-ceqtm} we show the elimination rule holds on
any elements of $\bool$ \emph{if} the elimination rule holds on those elements'
aspects. In \cref{lem:if-vinper} we use this to show that the elimination rule
holds on \emph{all} values in $\bool$.

\begin{lemma}\label{lem:if-ceqtm}
If
\begin{enumerate}
\item $\ceqtm{M}{M'}{\bool}$ such that for any $\msubsts{\Psi_1}{\psi_1}{\Psi}$
and $\msubsts{\Psi_2}{\psi_2}{\Psi_1}$, the elimination rule holds for any
pair of the aspects $M_2,M_{12},M_2',M_{12}'$ of $M,M'$,
\item $\eqtype{\oft{a}{\bool}}{A}{A'}$,
\item $\ceqtm{T}{T'}{\subst{A}{\true}{a}}$, and
\item $\ceqtm{F}{F'}{\subst{A}{\false}{a}}$,
\end{enumerate}
then $\ceqtm{\ifb{a.A}{M}{T}{F}}{\ifb{a.A'}{M'}{T'}{F'}}{\subst{A}{M}{a}}$.
\end{lemma}
\begin{proof}
We carefully work through the unary version of the proof, i.e., that
$\coftype{\ifb{a.A}{M}{T}{F}}{\subst{A}{M}{a}}$.
(The full proof follows by repeating the argument for $M',T',F'$.)
We show that for all $\msubsts{\Psi_1}{\psi_1}{\Psi}$ and
$\msubsts{\Psi_2}{\psi_2}{\Psi_1}$,
$\td{\ifsym}{\psi_1} \evals I_1$ and
$\inper[\Psi_2]{\td{I_1}{\psi_2}}{\td{\ifsym}{\psi_1\psi_2}}
{\td{\subst{A}{M}{a}}{\psi_1\psi_2}}$.

By $\coftype{M}{\bool}$, we know
$\td{M}{\psi_1}\evals M_1$,
$\td{M_1}{\psi_2}\evals M_2$,
$\td{M}{\psi_1\psi_2}\evals M_{12}$,
$\vinper[\Psi_1]{M_1}{M_1}{\bool}$, and
$\vinper[\Psi_2]{M_{12}}{M_2}{\bool}$.
Since $\ifsym$ evaluates its first argument,
$\ifb{a.\td{A}{\psi_1}}{\td{M}{\psi_1}}{\td{T}{\psi_1}}{\td{F}{\psi_1}}
\steps^* \ifb{a.\td{A}{\psi_1}}{M_1}{\td{T}{\psi_1}}{\td{F}{\psi_1}}$,
and by assumption,
$\coftype[\Psi_1]
{\ifb{a.\td{A}{\psi_1}}{M_1}{\td{T}{\psi_1}}{\td{F}{\psi_1}}}
{\subst{\td{A}{\psi_1}}{M_1}{a}}$.
This implies
\[
\inper[\Psi_2]
{\ifb{a.\td{A}{\psi_1\psi_2}}{\td{M_1}{\psi_2}}
  {\td{T}{\psi_1\psi_2}}{\td{F}{\psi_1\psi_2}} \steps^*
 \ifb{a.\td{A}{\psi_1\psi_2}}{M_2}{\td{T}{\psi_1\psi_2}}{\td{F}{\psi_1\psi_2}}}
{\td{I_1}{\psi_2}}
{\td{\subst{\td{A}{\psi_1}}{M_1}{a}}{\psi_2}}.
\]
On the other hand,
$\td{\ifsym}{\psi_1\psi_2} \steps^*
\ifb{a.\td{A}{\psi_1\psi_2}}{M_{12}}
  {\td{T}{\psi_1\psi_2}}{\td{F}{\psi_1\psi_2}}$
so
$\ceqtm[\Psi_2]
{\ifb{a.\td{A}{\psi_1\psi_2}}{M_{12}}
  {\td{T}{\psi_1\psi_2}}{\td{F}{\psi_1\psi_2}}}
{\ifb{a.\td{A}{\psi_1\psi_2}}{M_2}
  {\td{T}{\psi_1\psi_2}}{\td{F}{\psi_1\psi_2}}}
{\subst{\td{A}{\psi_1\psi_2}}{M_{12}}{a}}$
and these terms are
$\inper[\Psi_2]{-}{-}{\subst{\td{A}{\psi_1\psi_2}}{M_{12}}{a}}$ also.
Thus the result follows by transitivity once we show
$\eqper[\Psi_2]
{\subst{\td{A}{\psi_1\psi_2}}{\td{M_1}{\psi_2}}{a}}
{\subst{\td{A}{\psi_1\psi_2}}{M_{12}}{a}}$ and
$\eqper[\Psi_2]
{\subst{\td{A}{\psi_1\psi_2}}{M_{12}}{a}}
{\subst{\td{A}{\psi_1\psi_2}}{\td{M}{\psi_1\psi_2}}{a}}$.

The former holds because
$\ceqtm[\Psi_2]{\td{M_1}{\psi_2}}{M_{12}}{\bool}$,
which follows from \cref{lem:coftype-ceqtm} because both sides are
$\coftype[\Psi_1]{-}{\bool}$ (since $\bool$ is cubical) and
$\inper[\Psi']{\td{M_1}{\psi_2\psi}}{\td{M_{12}}{\psi}}{\bool}$ for all
$\msubsts{\Psi'}{\psi}{\Psi_2}$ (since each side is
$\inper[\Psi']{-}{\td{M}{\psi_1\psi_2\psi}}{\bool}$).
The latter holds similarly.
\end{proof}

We prove \cref{lem:if-ceqtm} first because we will appeal to it in the proof of
\cref{lem:if-vinper}, using the induction hypotheses to satisfy
the assumption about aspects. Note that in the statement of
\cref{lem:if-vinper}, $\cwftype{\subst{A}{V}{a}}$ because $\bool$ is cubical and
thus $\coftype{V}{\bool}$.

\begin{lemma}\label{lem:if-vinper}
If $\vinper{V}{V'}{\bool}$, then for any
$\eqtype{\oft{a}{\bool}}{A}{A'}$,
$\ceqtm{T}{T'}{\subst{A}{\true}{a}}$, and
$\ceqtm{F}{F'}{\subst{A}{\false}{a}}$,
$\ceqtm{\ifb{a.A}{V}{T}{F}}{\ifb{a.A'}{V'}{T'}{F'}}{\subst{A}{V}{a}}$.
\end{lemma}
\begin{proof}
By induction on $\vinper{V}{V'}{\bool}$.

\begin{enumerate}
\item $\vinper{\true}{\true}{\bool}$.

For all $\psi$, $\ifb{a.\td{A}{\psi}}{\true}{\td{T}{\psi}}{\td{F}{\psi}} \steps
\td{T}{\psi}$ and $\ifb{a.\td{A'}{\psi}}{\true}{\td{T'}{\psi}}{\td{F'}{\psi}}
\steps \td{T'}{\psi}$. Thus by \cref{lem:expansion} on both sides, the result
follows from $\ceqtm{T}{T'}{\subst{A}{\true}{a}}$, which we have assumed.

\item $\vinper{\false}{\false}{\bool}$.

This follows by \cref{lem:expansion} and $\ceqtm{F}{F'}{\subst{A}{\false}{a}}$.

\item
$\vinper{\hcomgeneric{\etc{x_i}}{\bool}}
{\hcom{\etc{x_i}}{\bool}{r}{r'}{O}{\etc{y.P^\e_i}}}{\bool}$,
such that $r\neq r'$,
$\ceqtm{M}{O}{\bool}$,
$\ceqtmres[\Psi,y]{x_i=\e,x_j=\e'}{N^\e_i}{N^{\e'}_j}{\bool}$
for all $i,j,\e,\e'$,
$\ceqtmres[\Psi,y]{x_i=\e}{N^\e_i}{P^\e_i}{\bool}$
for all $i,\e$, and
$\ceqtmres{x_i=\e}{\dsubst{N^\e_i}{r}{y}}{M}{\bool}$
for all $i,\e$.

We focus on the unary case, showing that for all
$\msubsts{\Psi_1}{\psi_1}{\Psi}$ and
$\msubsts{\Psi_2}{\psi_2}{\Psi_1}$,
$\td{\ifsym}{\psi_1} \evals I_1$ and
$\inper[\Psi_2]{\td{I_1}{\psi_2}}{\td{\ifsym}{\psi_1\psi_2}}
{\td{\subst{A}{V}{a}}{\psi_1\psi_2}}$.
We case-analyze $r$, $r'$, and $\etc{x_i}$ under $\psi_1$ and $\psi_1\psi_2$:

\begin{enumerate}
\item $\td{x_i}{\psi_1} = \e$ (where this is the smallest such $i$) and
$\td{x_j}{\psi_1\psi_2} = \e'$ (where this is the smallest such $j$).

\[\begin{aligned}
\td{\ifsym}{\psi_1} &\steps
\ifb{a.\td{A}{\psi_1}}{\td{\dsubst{N^\e_i}{r'}{y}}{\psi_1}}
{\td{T}{\psi_1}}{\td{F}{\psi_1}} \\
\td{\ifsym}{\psi_1\psi_2} &\steps
\ifb{a.\td{A}{\psi_1\psi_2}}{\td{\dsubst{N^{\e'}_j}{r'}{y}}{\psi_1\psi_2}}
{\td{T}{\psi_1\psi_2}}{\td{F}{\psi_1\psi_2}}
\end{aligned}\]
The induction hypothesis tells us this lemma holds for all aspects of
$\td{\dsubst{N^\e_i}{r'}{y}}{\psi_1}$; therefore \cref{lem:if-ceqtm} applies,
yielding
$\coftype[\Psi_1]{\ifb{a.\td{A}{\psi_1}}{\td{\dsubst{N^\e_i}{r'}{y}}{\psi_1}}
{\td{T}{\psi_1}}{\td{F}{\psi_1}}}
{\td{\subst{A}{\dsubst{N^\e_i}{r'}{y}}{a}}{\psi_1}}$, and thus
$\inper[\Psi_2]{\td{I_1}{\psi_2}}{\ifb{a.\td{A}{\psi_1\psi_2}}
{\td{\dsubst{N^\e_i}{r'}{y}}{\psi_1\psi_2}}
{\td{T}{\psi_1\psi_2}}{\td{F}{\psi_1\psi_2}}}
{\td{\subst{A}{\dsubst{N^\e_i}{r'}{y}}{a}}{\psi_1\psi_2}}$.
But by applying the induction hypothesis and \cref{lem:if-ceqtm} to
$\ceqtm[\Psi_2]
{\td{\dsubst{N^\e_i}{r'}{y}}{\psi_1\psi_2}}
{\td{\dsubst{N^{\e'}_j}{r'}{y}}{\psi_1\psi_2}}
{\bool}$ we have
\[\begin{aligned}
&\inper[\Psi_2]{\ifb{a.\td{A}{\psi_1\psi_2}}
  {\td{\dsubst{N^\e_i}{r'}{y}}{\psi_1\psi_2}}
  {\td{T}{\psi_1\psi_2}}{\td{F}{\psi_1\psi_2}}}
{}{\td{\subst{A}{\dsubst{N^\e_i}{r'}{y}}{a}}{\psi_1\psi_2}} \\
&{\ifb{a.\td{A}{\psi_1\psi_2}}
  {\td{\dsubst{N^{\e'}_j}{r'}{y}}{\psi_1\psi_2}}
  {\td{T}{\psi_1\psi_2}}{\td{F}{\psi_1\psi_2}}}.
\end{aligned}\]
It remains only to show that
$\eqper[\Psi_2]{\td{\subst{A}{\dsubst{N^\e_i}{r'}{y}}{a}}{\psi_1\psi_2}}
{\td{\subst{A}{V}{a}}{\psi_1\psi_2}}$,
which follows from
$\ceqtm[\Psi_2]{\td{\dsubst{N^\e_i}{r'}{y}}{\psi_1\psi_2}}
{\td{V}{\psi_1\psi_2}}{\bool}$ by the third Kan condition of $\bool$.

\item All $\td{x_i}{\psi_1}$ are dimension names,
$\td{r}{\psi_1} = \td{r'}{\psi_1}$, and
$\td{x_i}{\psi_1\psi_2} = \e$ (where this is the smallest such $i$).

This is similar to the previous case. We have
$\td{\ifsym}{\psi_1} \steps \ifbdots{a.\td{A}{\psi_1}}{\td{M}{\psi_1}}$
and $\td{\ifsym}{\psi_1\psi_2} \steps
\ifbdots{a.\td{A}{\psi_1\psi_2}}{\td{\dsubst{N^\e_i}{r'}{y}}{\psi_1\psi_2}}$.
By the induction hypothesis and \cref{lem:if-ceqtm} on
$\coftype{M}{\bool}$, we have
$\coftype[\Psi_1]{\ifbdots{a.\td{A}{\psi_1}}{\td{M}{\psi_1}}}
{\td{\subst{A}{M}{a}}{\psi_1}}$ and therefore
$\inper[\Psi_2]{\td{I_1}{\psi_2}}
{\ifbdots{a.\td{A}{\psi_1\psi_2}}{\td{M}{\psi_1\psi_2}}}
{\td{\subst{A}{M}{a}}{\psi_1\psi_2}}$.
On the other hand, by
$\ceqtm[\Psi_2]{\td{M}{\psi_1\psi_2}}
{\td{\dsubst{N^\e_i}{r}{y}}{\psi_1\psi_2}}{\bool}$ we have
$\inper[\Psi_2]
{\ifbdots{a.\td{A}{\psi_1\psi_2}}{\td{M}{\psi_1\psi_2}}}
{\ifbdots{a.\td{A}{\psi_1\psi_2}}{\td{\dsubst{N^\e_i}{r}{y}}{\psi_1\psi_2}}}
{\td{\subst{A}{M}{a}}{\psi_1\psi_2}}$.
The result follows from
$\ceqtm[\Psi_2]{\td{M}{\psi_1\psi_2}}{\td{V}{\psi_1\psi_2}}{\bool}$ by the
second Kan condition of $\bool$.

\item All $\td{x_i}{\psi_1}$ are dimension names,
$\td{r}{\psi_1} = \td{r'}{\psi_1}$, and
all $\td{x_i}{\psi_1\psi_2}$ are dimension names.

We have
$\td{\ifsym}{\psi_1} \steps \ifbdots{a.\td{A}{\psi_1}}{\td{M}{\psi_1}}$
and $\td{\ifsym}{\psi_1\psi_2} \steps
\ifbdots{a.\td{A}{\psi_1\psi_2}}{\td{M}{\psi_1\psi_2}}$.
By the induction hypothesis and \cref{lem:if-ceqtm} on
$\coftype{M}{\bool}$, we have
$\coftype[\Psi_1]{\ifbdots{a.\td{A}{\psi_1}}{\td{M}{\psi_1}}}
{\td{\subst{A}{M}{a}}{\psi_1}}$ and therefore
$\inper[\Psi_2]{\td{I_1}{\psi_2}}
{\ifbdots{a.\td{A}{\psi_1\psi_2}}{\td{M}{\psi_1\psi_2}}}
{\td{\subst{A}{M}{a}}{\psi_1\psi_2}}$.
The result follows from
$\ceqtm[\Psi_2]{\td{M}{\psi_1\psi_2}}{\td{V}{\psi_1\psi_2}}{\bool}$ by the
second Kan condition of $\bool$.

\item All $\td{x_i}{\psi_1}$ are dimension names,
$\td{r}{\psi_1} \neq \td{r'}{\psi_1}$, and
$\td{x_i}{\psi_1\psi_2} = \e$ (where this is the smallest such $i$).

\[\begin{aligned}
\td{\ifsym}{\psi_1} &\steps
\com{\etc{\td{x_i}{\psi_1}}}{z.\subst{\td{A}{\psi_1}}{H}{a}}
{\td{r}{\psi_1}}{\td{r'}{\psi_1}}
{\ifb{a.\td{A}{\psi_1}}{\td{M}{\psi_1}}{\td{T}{\psi_1}}{\td{F}{\psi_1}}}
{\etc{y.\ifb{a.\td{A}{\psi_1}}{\td{N^\e_i}{\psi_1}}
  {\td{T}{\psi_1}}{\td{F}{\psi_1}}}} \\
\td{\ifsym}{\psi_1\psi_2} &\steps
\ifb{a.\td{A}{\psi_1\psi_2}}{\td{\dsubst{N^\e_i}{r'}{y}}{\psi_1\psi_2}}
{\td{T}{\psi_1\psi_2}}{\td{F}{\psi_1\psi_2}}
\end{aligned}\]
where $H :=
\hcom{\etc{\td{x_i}{\psi_1}}}{\bool}{\td{r}{\psi_1}}{z}{\td{M}{\psi_1}}
{\etc{y.\td{N^\e_i}{\psi_1}}}$.
We start by showing
\[
\coftype[\Psi_1]{\comsym}
{\td{\subst{A}{\hcom{\etc{x_i}}{\bool}{r}{r'}{M}{\etc{y.N^\e_i}}}{a}}{\psi_1}}.
\]
Since
$\dsubst{\subst{\td{A}{\psi_1}}{H}{a}}{\td{r'}{\psi_1}}{z} =
\td{\subst{A}{\hcom{\etc{x_i}}{\bool}{r}{r'}{M}{\etc{y.N^\e_i}}}{a}}{\psi_1} =
\td{\subst{A}{V}{a}}{\psi_1}$
(by the definition of $H$, and $z\fresh A$),
this follows from \cref{thm:com} so long as:
\begin{enumerate}
\item $\coftype[\Psi_1]
{\ifbdots{a.\td{A}{\psi_1}}{\td{M}{\psi_1}}}
{\subst{\td{A}{\psi_1}}{\dsubst{H}{\td{r}{\psi_1}}{z}}{a}}$,
\item $\ceqtmres[\Psi_1,y]{\td{x_i}{\psi_1}=\e,\td{x_j}{\psi_1}=\e'}
{\ifbdots{a.\td{A}{\psi_1}}{\td{N^\e_i}{\psi_1}}}
{\ifbdots{a.\td{A}{\psi_1}}{\td{N^{\e'}_j}{\psi_1}}}
{\subst{\td{A}{\psi_1}}{\dsubst{H}{y}{z}}{a}}$
for all $i,j,\e,\e'$, and
\item $\ceqtmres[\Psi_1]{\td{x_i}{\psi_1}=\e}
{\ifbdots{a.\td{A}{\psi_1}}{\td{\dsubst{N^\e_i}{r}{y}}{\psi_1}}}
{\ifbdots{a.\td{A}{\psi_1}}{\td{M}{\psi_1}}}
{\subst{\td{A}{\psi_1}}{\dsubst{H}{\td{r}{\psi_1}}{z}}{a}}$
for all $i,\e$.
\end{enumerate}

All three follow from the induction hypothesis and \cref{lem:if-ceqtm} applied
to the appropriate equality judgments, using the Kan conditions of $\bool$ to
adjust the types: by the second Kan condition,
$\ceqtm[\Psi_1]{\dsubst{H}{\td{r}{\psi_1}}{z}}{\td{M}{\psi_1}}{\bool}$,
and by the third Kan condition of $\bool$,
$\ceqtmres[\Psi_1,y]{\td{x_i}{\psi_1}=\e}
{\dsubst{H}{y}{z}}{\td{N^\e_i}{\psi_1}}{\bool}$.

From this we conclude $\comsym \evals I_1$ and
$\inper[\Psi_2]{\td{I_1}{\psi_2}}{\td{\comsym}{\psi_2}}
{\td{\subst{A}{V}{a}}{\psi_1\psi_2}}$.
However, since $\td{x_i}{\psi_1\psi_2} = \e$, by \cref{thm:com},
$\ceqtm[\Psi_2]{\td{\comsym}{\psi_2}}
{\ifbdots{a.\td{A}{\psi_1\psi_2}}{\td{\dsubst{N^\e_i}{r'}{y}}{\psi_1\psi_2}}}
{\td{\subst{A}{V}{a}}{\psi_1\psi_2}}$.
But the right-hand side is what $\td{\ifsym}{\psi_1\psi_2}$ steps to, so
$\inper[\Psi_2]{\td{I_1}{\psi_2}}{\td{\ifsym}{\psi_1\psi_2}}
{\td{\subst{A}{V}{a}}{\psi_1\psi_2}}$ follows.

\item All $\td{x_i}{\psi_1}$ are dimension names,
$\td{r}{\psi_1} \neq \td{r'}{\psi_1}$,
all $\td{x_i}{\psi_1\psi_2}$ are dimension names, and
$\td{r}{\psi_1\psi_2} = \td{r'}{\psi_1\psi_2}$.

We have $\td{\ifsym}{\psi_1} \steps \comsym$ as in the previous case, and
$\td{\ifsym}{\psi_1\psi_2} \steps
\ifbdots{a.\td{A}{\psi_1\psi_2}}{\td{M}{\psi_1\psi_2}}$.
As before, $\coftype[\Psi_1]{\comsym}{\td{\subst{A}{V}{a}}{\psi_1}}$, so
$\comsym \evals I_1$ and
$\inper[\Psi_2]{\td{I_1}{\psi_2}}{\td{\comsym}{\psi_2}}
{\td{\subst{A}{V}{a}}{\psi_1\psi_2}}$.
Since $\td{r}{\psi_1\psi_2} = \td{r'}{\psi_1\psi_2}$, by \cref{thm:com},
$\ceqtm[\Psi_2]{\td{\comsym}{\psi_2}}
{\ifbdots{a.\td{A}{\psi_1\psi_2}}{\td{M}{\psi_1\psi_2}}}
{\td{\subst{A}{V}{a}}{\psi_1\psi_2}}$.
Thus
$\inper[\Psi_2]{\td{\comsym}{\psi_2}}
{\ifbdots{a.\td{A}{\psi_1\psi_2}}{\td{M}{\psi_1\psi_2}}}
{\td{\subst{A}{V}{a}}{\psi_1\psi_2}}$, and the result follows by transitivity.

\item All $\td{x_i}{\psi_1}$ are dimension names,
$\td{r}{\psi_1} \neq \td{r'}{\psi_1}$,
all $\td{x_i}{\psi_1\psi_2}$ are dimension names, and
$\td{r}{\psi_1\psi_2} \neq \td{r'}{\psi_1\psi_2}$.

We have $\td{\ifsym}{\psi_1} \steps \comsym$ as in the previous cases,
and $\td{\ifsym}{\psi_1\psi_2} \steps \td{\comsym}{\psi_2}$.
As before, $\coftype[\Psi_1]{\comsym}{\td{\subst{A}{V}{a}}{\psi_1}}$, so
$\comsym \evals I_1$ and
$\inper[\Psi_2]{\td{I_1}{\psi_2}}{\td{\comsym}{\psi_2}}
{\td{\subst{A}{V}{a}}{\psi_1\psi_2}}$, which is what we need to show.
\qedhere
\end{enumerate}
\end{enumerate}
\end{proof}

Since \cref{lem:if-vinper} holds for all $\vinper{V}{V'}{\bool}$,
\cref{lem:if-ceqtm} holds for \emph{any} $\ceqtm{M}{M'}{\bool}$, and thus
the elimination rule for booleans holds.

\paragraph{Computation}
If $\wftype{\oft{a}{\bool}}{A}$,
$\coftype{T}{\subst{A}{\true}{a}}$, and
$\coftype{F}{\subst{A}{\false}{a}}$, then
$\ceqtm{\ifb{a.A}{\true}{T}{F}}{T}{\subst{A}{\true}{a}}$ and
$\ceqtm{\ifb{a.A}{\false}{T}{F}}{F}{\subst{A}{\false}{a}}$.

For all $\psi$,
$\ifb{a.\td{A}{\psi}}{\true}{\td{T}{\psi}}{\td{F}{\psi}} \steps \td{T}{\psi}$, so
the former follows by \cref{lem:expansion} and
$\coftype{T}{\subst{A}{\true}{a}}$. The latter case is analogous.

\subsection{Circle}

Our definition of $\C$ is very similar to that of $\bool$, because we defined
$\bool$ as a higher inductive type (with no path constructors). A cubical type
system \emph{has the circle} if $\veqper{\C}{\C}$ for all $\Psi$, and
$\vinper[-]{-}{-}{\C}$ is the least relation such that:
\begin{enumerate}
\item $\vinper{\base}{\base}{\C}$,
\item $\vinper{\lp{x}}{\lp{x}}{\C}$, and
\item $\vinper{\hcomgeneric{\etc{x_i}}{\C}}{\hcom{\etc{x_i}}{\C}{r}{r'}{O}{\etc{y.P^\e_i}}}{\C}$ whenever $r\neq r'$,
\begin{enumerate}
\item $\ceqtm{M}{O}{\C}$,
\item $\ceqtmres[\Psi,y]{x_i=\e,x_j=\e'}{N^\e_i}{N^{\e'}_j}{\C}$
for all $i,j,\e,\e'$,
\item $\ceqtmres[\Psi,y]{x_i=\e}{N^\e_i}{P^\e_i}{\C}$
for all $i,\e$, and
\item $\ceqtmres{x_i=\e}{\dsubst{N^\e_i}{r}{y}}{M}{\C}$
for all $i,\e$.
\end{enumerate}
\end{enumerate}

We proceed by proving theorems about cubical type systems that have the circle.
We will omit the many proofs that are identical to those in the previous
subsection.

\paragraph{Pretype}
$\cpretype{\C}$.

For all $\psi_1,\psi_2$,
$\td{\C}{\psi_1}\evals\C$,
$\td{\C}{\psi_2}\evals\C$,
$\td{\C}{\psi_1\psi_2}\evals\C$,
and $\veqper[\Psi_2]{\C}{\C}$.

\paragraph{Introduction}
$\coftype{\base}{\C}$, $\coftype{\lp{r}}{\C}$, and $\ceqtm{\lp{\e}}{\base}{\C}$.

\begin{enumerate}
\item
For all $\msubsts{\Psi_1}{\psi_1}{\Psi}$
and $\msubsts{\Psi_2}{\psi_2}{\Psi_1}$,
$\td{\base}{\psi_1}\evals\base$,
$\td{\base}{\psi_2}\evals\base$,
$\td{\base}{\psi_1\psi_2}\evals\base$,
and $\vinper[\Psi_2]{\base}{\base}{\C}$.

\item
For all $\msubsts{\Psi_1}{\psi_1}{\Psi}$
and $\msubsts{\Psi_2}{\psi_2}{\Psi_1}$,
we case on $\td{r}{\psi_1}$ and $\td{r}{\psi_1\psi_2}$:
\begin{enumerate}
\item $\td{r}{\psi_1} = \e$. (Therefore $\td{r}{\psi_1\psi_2} = \e$ also.)

Then $\lp{\td{r}{\psi_1}}\evals\base$,
$\td{\base}{\psi_2}\evals\base$,
$\lp{\td{r}{\psi_1\psi_2}}\evals\base$,
and $\vinper[\Psi_2]{\base}{\base}{\C}$.

\item $\td{r}{\psi_1} = x$ and $\td{x}{\psi_2} = \e$.

Then $\lp{\td{r}{\psi_1}}\evals\lp{x}$,
$\lp{\td{x}{\psi_2}}\evals\base$,
$\lp{\td{r}{\psi_1\psi_2}}\evals\base$,
and $\vinper[\Psi_2]{\base}{\base}{\C}$.

\item $\td{r}{\psi_1} = x$ and $\td{x}{\psi_2} = x'$.

Then $\lp{\td{r}{\psi_1}}\evals\lp{x}$,
$\lp{\td{x}{\psi_2}}\evals\lp{x'}$,
$\lp{\td{r}{\psi_1\psi_2}}\evals\lp{x'}$,
and $\vinper[\Psi_2]{\lp{x'}}{\lp{x'}}{\C}$.
\end{enumerate}

\item By head expansion and the first introduction rule, since for all $\psi$,
$\lp{\td{\e}{\psi}}\steps\td{\base}{\psi}$.
\end{enumerate}

\paragraph{Kan}
$\cpretype{\C}$ is Kan.

This proof is identical to the proof that $\cpretype{\bool}$ is Kan, because the
relevant portions of the operational semantics and the definition of
$\vinper[-]{-}{-}{\C}$ are identical.

\paragraph{Cubical}
For any $\Psi'$ and $\vinper[\Psi']{M}{N}{\C}$, $\ceqtm[\Psi']{M}{N}{\C}$.

There are three cases in which $\vinper[\Psi']{M}{N}{\C}$.
For $\base$ and $\lp{x}$, this follows from the introduction rules already
proven. For $\vinper[\Psi']{\hcomgeneric{\etc{x_i}}{\C}}
{\hcom{\etc{x_i}}{\C}{r}{r'}{O}{\etc{y.P^\e_i}}}{\C}$,
this follows from the first Kan condition of $\C$, again already proven.

\paragraph{Elimination}
If $\ceqtm{M}{M'}{\C}$,
$\eqtype{\oft{a}{\C}}{A}{A'}$,
$\ceqtm{P}{P'}{\subst{A}{\base}{a}}$,
$\ceqtm[\Psi,x]{L}{L'}{\subst{A}{\lp{x}}{a}}$, and
$\ceqtm{\dsubst{L}{\e}{x}}{P}{\subst{A}{\base}{a}}$ for all $\e$, then
$\ceqtm{\Celim{a.A}{M}{P}{x.L}}{\Celim{a.A'}{M'}{P'}{x.L'}}{\subst{A}{M}{a}}$.

We prove the elimination rule using the same strategy as our proof of the
elimination rule for booleans; see \cref{lem:if-ceqtm,lem:if-vinper} for full
details.

\begin{lemma}\label{lem:Celim-ceqtm}
If
\begin{enumerate}
\item $\ceqtm{M}{M'}{\C}$ such that for any $\msubsts{\Psi_1}{\psi_1}{\Psi}$
and $\msubsts{\Psi_2}{\psi_2}{\Psi_1}$, the elimination rule holds for any
pair of the aspects $M_2,M_{12},M_2',M_{12}'$ of $M,M'$,
\item $\eqtype{\oft{a}{\C}}{A}{A'}$,
\item $\ceqtm{P}{P'}{\subst{A}{\base}{a}}$,
\item $\ceqtm[\Psi,x]{L}{L'}{\subst{A}{\lp{x}}{a}}$, and
\item $\ceqtm{\dsubst{L}{\e}{x}}{P}{\subst{A}{\base}{a}}$ for all $\e$,
\end{enumerate}
then
$\ceqtm{\Celim{a.A}{M}{P}{x.L}}{\Celim{a.A'}{M'}{P'}{x.L'}}{\subst{A}{M}{a}}$.
\end{lemma}
\begin{proof}
Same proof as \cref{lem:if-ceqtm}.
\end{proof}

\begin{lemma}\label{lem:Celim-vinper}
If $\vinper{V}{V'}{\C}$, then for any
$\eqtype{\oft{a}{\C}}{A}{A'}$,
$\ceqtm{P}{P'}{\subst{A}{\base}{a}}$, and
$\ceqtm[\Psi,x]{L}{L'}{\subst{A}{\lp{x}}{a}}$, such that
$\ceqtm{\dsubst{L}{\e}{x}}{P}{\subst{A}{\base}{a}}$ for all $\e$,
$\ceqtm{\Celim{a.A}{V}{P}{x.L}}{\Celim{a.A'}{V'}{P'}{x.L'}}{\subst{A}{V}{a}}$.
\end{lemma}
\begin{proof}
By induction on $\vinper{V}{V'}{\C}$.

\begin{enumerate}
\item $\vinper{\base}{\base}{\C}$.

For all $\psi$,
$\Celim{a.\td{A}{\psi}}{\base}{\td{P}{\psi}}{x.\td{L}{\psi}} \steps
\td{P}{\psi}$ and
$\Celim{a.\td{A'}{\psi}}{\base}{\td{P'}{\psi}}{x.\td{L'}{\psi}} \steps
\td{P'}{\psi}$.
Thus by \cref{lem:expansion} on both sides, the result follows from
$\ceqtm{P}{P'}{\subst{A}{\base}{a}}$, which we have assumed.

\item $\vinper{\lp{y}}{\lp{y}}{\C}$.

We focus on the unary case, showing that for all
$\msubsts{\Psi_1}{\psi_1}{\Psi}$ and
$\msubsts{\Psi_2}{\psi_2}{\Psi_1}$,
$\td{\Celimsym}{\psi_1} \evals E_1$ and
$\inper[\Psi_2]{\td{E_1}{\psi_2}}{\td{\Celimsym}{\psi_1\psi_2}}
{\td{\subst{A}{\lp{y}}{a}}{\psi_1\psi_2}}$.
We case-analyze $y$ under $\psi_1$ and $\psi_1\psi_2$:

\begin{enumerate}
\item $\td{y}{\psi_1} = \e$.

Then $\td{\Celimsym}{\psi_1} \steps^* \td{P}{\psi_1}$ and
$\td{\Celimsym}{\psi_1\psi_2} \steps^* \td{P}{\psi_1\psi_2}$.
By $\coftype{P}{\subst{A}{\base}{a}}$,
$\td{P}{\psi_1} \evals P_1$ and
$\inper[\Psi_2]{\td{P_1}{\psi_2}}{\td{P}{\psi_1\psi_2}}
{\td{\subst{A}{\base}{a}}{\psi_1\psi_2}}$. The result follows by
$\eqper[\Psi_2]
{\td{\subst{A}{\base}{a}}{\psi_1\psi_2}}
{\td{\subst{A}{\lp{y}}{a}}{\psi_1\psi_2}}$
because $\td{y}{\psi_1\psi_2} = \e$ and $\ceqtm[\Psi_2]{\base}{\lp{\e}}{\C}$.

\item $\td{y}{\psi_1} = y'$ and $\td{y}{\psi_1\psi_2} = \e$.

Then $\td{\Celimsym}{\psi_1} \steps^* \dsubst{\td{L}{\psi_1}}{y'}{x}$ and
$\td{\Celimsym}{\psi_1\psi_2} \steps^* \td{P}{\psi_1\psi_2}$.
By $\coftype[\Psi,x]{L}{\subst{A}{\lp{x}}{a}}$,
$\dsubst{\td{L}{\psi_1}}{y'}{x} \evals L_1$ and
$\inper[\Psi_2]{\td{L_1}{\psi_2}}
{\td{\dsubst{\td{L}{\psi_1}}{y'}{x}}{\psi_2} =
  \td{\dsubst{L}{\e}{x}}{\psi_1\psi_2}}
{\td{\subst{A}{\lp{\e}}{a}}{\psi_1\psi_2}}$,
since $x\fresh A$ and $\td{\dsubst{\td{x}{\psi_1}}{y'}{x}}{\psi_2} = \e$.
By $\ceqtm{\dsubst{L}{\e}{x}}{P}{\subst{A}{\base}{a}}$, we also have
$\inper[\Psi_2]{\td{\dsubst{L}{\e}{x}}{\psi_1\psi_2}}{\td{P}{\psi_1\psi_2}}
{\td{\subst{A}{\base}{a}}{\psi_1\psi_2}}$. The result follows by transitivity
and $\eqper[\Psi_2]
{\td{\subst{A}{\base}{a}}{\psi_1\psi_2}}
{\td{\subst{A}{\lp{\e}}{a}}{\psi_1\psi_2}}$.

\item $\td{y}{\psi_1} = y'$ and $\td{y}{\psi_1\psi_2} = y''$.

Then $\td{\Celimsym}{\psi_1} \steps^* \dsubst{\td{L}{\psi_1}}{y'}{x}$ and
$\td{\Celimsym}{\psi_1\psi_2} \steps^* \dsubst{\td{L}{\psi_1\psi_2}}{y''}{x}
= \td{\dsubst{\td{L}{\psi_1}}{y'}{x}}{\psi_2}$.
By $\coftype[\Psi,x]{L}{\subst{A}{\lp{x}}{a}}$,
$\dsubst{\td{L}{\psi_1}}{y'}{x} \evals L_1$ and
$\inper[\Psi_2]{\td{L_1}{\psi_2}}
{\td{\dsubst{\td{L}{\psi_1}}{y'}{x}}{\psi_2}}
{\td{\dsubst{\td{\subst{A}{\lp{x}}{a}}{\psi_1}}{y'}{x}}{\psi_2}}$.
The result follows by
$\td{\dsubst{\td{x}{\psi_1}}{y'}{x}}{\psi_2} = \td{y}{\psi_1\psi_2}$ and
$x\fresh A$.
\end{enumerate}

\item $\vinper{\hcomgeneric{\etc{x_i}}{\C}}{\hcom{\etc{x_i}}{\C}{r}{r'}{O}{\etc{y.P^\e_i}}}{\C}$, such that $r\neq r'$,
$\ceqtm{M}{O}{\bool}$,
$\ceqtmres[\Psi,y]{x_i=\e,x_j=\e'}{N^\e_i}{N^{\e'}_j}{\bool}$
for all $i,j,\e,\e'$,
$\ceqtmres[\Psi,y]{x_i=\e}{N^\e_i}{P^\e_i}{\bool}$
for all $i,\e$, and
$\ceqtmres{x_i=\e}{\dsubst{N^\e_i}{r}{y}}{M}{\bool}$
for all $i,\e$.

This case proceeds identically to the $\hcomsym$ case in \cref{lem:if-vinper},
and appeals to \cref{lem:Celim-ceqtm}.
\qedhere
\end{enumerate}
\end{proof}

Since \cref{lem:Celim-vinper} holds for all $\vinper{V}{V'}{\C}$,
\cref{lem:Celim-ceqtm} holds for \emph{any} $\ceqtm{M}{M'}{\C}$, and thus the
elimination rule for the circle holds.

\paragraph{Computation}
If $\wftype{\oft{a}{\C}}{A}$,
$\coftype{P}{\subst{A}{\base}{a}}$,
$\coftype[\Psi,x]{L}{\subst{A}{\lp{x}}{a}}$, and
$\ceqtm{\dsubst{L}{\e}{x}}{P}{\subst{A}{\base}{a}}$ for all $\e$, then
$\ceqtm{\Celim{a.A}{\base}{P}{x.L}}{P}{\subst{A}{\base}{a}}$ and
$\ceqtm{\Celim{a.A}{\lp{r}}{P}{x.L}}{\dsubst{L}{r}{x}}{\subst{A}{\lp{r}}{a}}$.

For all $\psi$, $\Celim{a.\td{A}{\psi}}{\base}{\td{P}{\psi}}{x.\td{L}{\psi}}
\steps \td{P}{\psi}$,
so the first computation rule follows by \cref{lem:expansion} and
$\coftype{P}{\subst{A}{\base}{a}}$.

For the second computation rule, we know that both sides are
$\coftype{-}{\subst{A}{\lp{r}}{a}}$, so by \cref{lem:coftype-ceqtm}, it suffices
to show that for any $\psitd$,
$\inper[\Psi']
{\Celim{a.\td{A}{\psi}}{\lp{\td{r}{\psi}}}{\td{P}{\psi}}{x.\td{L}{\psi}}}
{\td{\dsubst{L}{r}{x}}{\psi}}
{\td{\subst{A}{\lp{r}}{a}}{\psi}}$.
There are two cases:

\begin{enumerate}
\item $\td{r}{\psi} = \e$.

Then $\Celim{a.\td{A}{\psi}}{\lp{\td{r}{\psi}}}{\td{P}{\psi}}{x.\td{L}{\psi}}
\steps \Celim{a.\td{A}{\psi}}{\base}{\td{P}{\psi}}{x.\td{L}{\psi}}
\steps \td{P}{\psi}$
and we must show
$\inper[\Psi']{\td{P}{\psi}}{\td{\dsubst{L}{r}{x}}{\psi}}
{\td{\subst{A}{\lp{r}}{a}}{\psi}}$.
Since $\td{r}{\psi} = \e$, and $\eqper[\Psi']
{\subst{\td{A}{\psi}}{\lp{\e}}{a}}{\subst{\td{A}{\psi}}{\base}{a}}$,
this is equivalently
$\inper[\Psi']{\td{P}{\psi}}{\td{\dsubst{L}{\e}{x}}{\psi}}
{\td{\subst{A}{\base}{a}}{\psi}}$, which we know by
$\ceqtm{\dsubst{L}{\e}{x}}{P}{\subst{A}{\base}{a}}$.

\item $\td{r}{\psi} = w$.

Then $\Celim{a.\td{A}{\psi}}{\lp{\td{r}{\psi}}}{\td{P}{\psi}}{x.\td{L}{\psi}}
\steps \dsubst{\td{L}{\psi}}{w}{x} = \td{\dsubst{L}{r}{x}}{\psi}$,
and it suffices to show that
$\inper[\Psi']{\td{\dsubst{L}{r}{x}}{\psi}}{\td{\dsubst{L}{r}{x}}{\psi}}
{\td{\subst{A}{\lp{r}}{a}}{\psi}}$, which we know by
$\coftype[\Psi,x]{L}{\subst{A}{\lp{x}}{a}}$.
\end{enumerate}

\subsection{Dependent functions}

If a cubical type system has $\ceqtype{A}{A'}$ and
$\eqtype{\oft{a}{A}}{B}{B'}$, we say it
\emph{has their dependent function type} when for all $\psitd$,
$\veqper[\Psi']{\picl{a}{\td{A}{\psi}}{\td{B}{\psi}}}
{\picl{a}{\td{A'}{\psi}}{\td{B'}{\psi}}}$, and
$\vinper[\Psi']{-}{-}{\picl{a}{\td{A}{\psi}}{\td{B}{\psi}}}$
is the least relation such that
\[
\vinper[\Psi']{\lam{a}{M}}{\lam{a}{M'}}{\picl{a}{\td{A}{\psi}}{\td{B}{\psi}}}
\]
when $\eqtm[\Psi']{\oft{a}{\td{A}{\psi}}}{M}{M'}{\td{B}{\psi}}$.

In any cubical type system where
$\veqper[\Psi']{\picl{a}{\td{A}{\psi}}{\td{B}{\psi}}}
{\picl{a}{\td{A'}{\psi}}{\td{B'}{\psi}}}$, the PERs
$\vinper[\Psi']{-}{-}{\picl{a}{\td{A}{\psi}}{\td{B}{\psi}}}$ and
$\vinper[\Psi']{-}{-}{\picl{a}{\td{A'}{\psi}}{\td{B'}{\psi}}}$ must be equal.
This is true because
$\eqtm[\Psi']{\oft{a}{\td{A}{\psi}}}{M}{M'}{\td{B}{\psi}}$
if and only if
$\eqtm[\Psi']{\oft{a}{\td{A'}{\psi}}}{M}{M'}{\td{B'}{\psi}}$.

In the remainder of this subsection, we assume we are working with a cubical
type system that has $\ceqtype{A}{A'}$, $\eqtype{\oft aA}{B}{B'}$, and their
dependent function type.

\paragraph{Pretype}
$\ceqpretype{\picl{a}{A}{B}}{\picl{a}{A'}{B'}}$.

For any $\msubsts{\Psi_1}{\psi_1}{\Psi}$ and $\msubsts{\Psi_2}{\psi_2}{\Psi_1}$,
$\wfval[\Psi_1]{\picl{a}{\td{A}{\psi_1}}{\td{B}{\psi_1}}}$ and
$\wfval[\Psi_2]{\picl{a}{\td{A}{\psi_1\psi_2}}{\td{B}{\psi_1\psi_2}}}$.
Then $\veqper[\Psi_2]
{\picl{a}{\td{A}{\psi_1\psi_2}}{\td{B}{\psi_1\psi_2}}}
{\picl{a}{\td{A'}{\psi_1\psi_2}}{\td{B'}{\psi_1\psi_2}}}$, which follows by our
assumption that the cubical type system has this dependent function type.

\paragraph{Introduction}
If $\eqtm{\oft aA}{M}{M'}{B}$ then
$\ceqtm{\lam{a}{M}}{\lam{a}{M'}}{\picl{a}{A}{B}}$.

Each side has coherent aspects up to syntactic equality, since
$\wfval[\Psi']{\lam{a}{\td{M}{\psi}}}$ for all $\psitd$.
Thus it suffices to show
$\vinper[\Psi_2]
{\lam{a}{\td{M}{\psi_1\psi_2}}}
{\lam{a}{\td{M'}{\psi_1\psi_2}}}
{\picl{a}{\td{A}{\psi_1\psi_2}}{\td{B}{\psi_1\psi_2}}}$,
which is true because
$\eqtm[\Psi_2]{\oft{a}{\td{A}{\psi_1\psi_2}}}
{\td{M}{\psi_1\psi_2}}{\td{M'}{\psi_1\psi_2}}{\td{B}{\psi_1\psi_2}}$.

\paragraph{Elimination}
If $\ceqtm{M}{M'}{\picl{a}{A}{B}}$ and $\ceqtm{N}{N'}{A}$, then
$\ceqtm{\app{M}{N}}{\app{M'}{N'}}{\subst{B}{N}{a}}$.

For any $\msubsts{\Psi_1}{\psi_1}{\Psi}$ and $\msubsts{\Psi_2}{\psi_2}{\Psi_1}$,
by $\coftype{M}{\picl{a}{A}{B}}$ we know
$\td{M}{\psi_1}\evals \lam{a}{O_1}$ and
$\oftype[\Psi_1]{\oft{a}{\td{A}{\psi_1}}}{O_1}{\td{B}{\psi_1}}$. Thus
\[
\app{\td{M}{\psi_1}}{\td{N}{\psi_1}}
\steps^* \app{\lam{a}{O_1}}{\td{N}{\psi_1}}
\steps \subst{O_1}{\td{N}{\psi_1}}{a}
\]
and since $\coftype[\Psi_1]{\td{N}{\psi_1}}{\td{A}{\psi_1}}$, we know
$\coftype[\Psi_1]{\subst{O_1}{\td{N}{\psi_1}}{a}}
{\td{\subst{B}{N}{a}}{\psi_1}}$ and thus
$\subst{O_1}{\td{N}{\psi_1}}{a} \evals X_1$ and
$\inper[\Psi_2]{\td{X_1}{\psi_2}}
{\subst{\td{O_1}{\psi_2}}{\td{N}{\psi_1\psi_2}}{a}}
{\td{\subst{B}{N}{a}}{\psi_1\psi_2}}$.

We also know
$\vinper[\Psi_2]
{\td{M}{\psi_1\psi_2}\evals \lam{a}{O_{12}}}
{\lam{a}{\td{O_1}{\psi_2}}}
{\picl{a}{\td{A}{\psi_1\psi_2}}{\td{B}{\psi_1\psi_2}}}$ so
$\eqtm[\Psi_2]{\oft{a}{\td{A}{\psi_1\psi_2}}}
{\td{O_1}{\psi_2}}{O_{12}}{\td{B}{\psi_1\psi_2}}$. Thus
\[
\app{\td{M}{\psi_1\psi_2}}{\td{N}{\psi_1\psi_2}}
\steps^* \app{\lam{a}{O_{12}}}{\td{N}{\psi_1\psi_2}}
\steps \subst{O_{12}}{\td{N}{\psi_1\psi_2}}{a}
\]
and
$\ceqtm[\Psi_2]
{\subst{\td{O_1}{\psi_2}}{\td{N}{\psi_1\psi_2}}{a}}
{\subst{O_{12}}{\td{N}{\psi_1\psi_2}}{a}}
{\td{\subst{B}{N}{a}}{\psi_1\psi_2}}$.
This implies
$\inper[\Psi_2]
{\subst{\td{O_1}{\psi_2}}{\td{N}{\psi_1\psi_2}}{a}}
{\subst{O_{12}}{\td{N}{\psi_1\psi_2}}{a}}
{\td{\subst{B}{N}{a}}{\psi_1\psi_2}}$, so
$\inper[\Psi_2]
{\td{X_1}{\psi_2}}
{\subst{O_{12}}{\td{N}{\psi_1\psi_2}}{a}}
{\td{\subst{B}{N}{a}}{\psi_1\psi_2}}$
follows by transitivity.

An analogous argument shows $\app{M'}{N'}$ also has coherent aspects.
To see that the aspects of $\app{M}{N}$ and $\app{M'}{N'}$ are related to each
other, we use the fact that
$\td{M'}{\psi_1\psi_2} \evals \lam{a}{O'_{12}}$ and
$\eqtm[\Psi_2]{\oft{a}{\td{A}{\psi_1\psi_2}}}
{O_{12}}{O'_{12}}{\td{B}{\psi_1\psi_2}}$, so
$\app{\td{M'}{\psi_1\psi_2}}{\td{N'}{\psi_1\psi_2}} \steps^*
\subst{O'_{12}}{\td{N'}{\psi_1\psi_2}}{a}$ and
$\ceqtm[\Psi_2]
{\subst{O_{12}}{\td{N}{\psi_1\psi_2}}{a}}
{\subst{O'_{12}}{\td{N'}{\psi_1\psi_2}}{a}}
{\td{\subst{B}{N}{a}}{\psi_1\psi_2}}$.

\paragraph{Computation}
If $\oftype{\oft aA}{M}{B}$ and $\coftype{N}{A}$, then
$\ceqtm{\app{\lam{a}{M}}{N}}{\subst{M}{N}{a}}{\subst{B}{N}{a}}$.

That $\coftype{\subst{M}{N}{a}}{\subst{B}{N}{a}}$ follows from the definition of
$\oftype{\oft aA}{M}{B}$, and the result follows by \cref{lem:expansion}.

\paragraph{Eta}
If $\coftype{M}{\picl{a}{A}{B}}$ then
$\ceqtm{M}{\lam{a}{\app{M}{a}}}{\picl{a}{A}{B}}$.

We start by proving $\coftype{\lam{a}{\app{M}{a}}}{\picl{a}{A}{B}}$.
By the introduction rule for dependent functions, it suffices to show
$\oftype{\oft{a}{A}}{\app{M}{a}}{B}$, which holds when for any $\psitd$
and $\ceqtm[\Psi']{N}{N'}{\td{A}{\psi}}$,
$\ceqtm[\Psi']{\app{\td{M}{\psi}}{N}}{\app{\td{M}{\psi}}{N'}}
{\subst{\td{B}{\psi}}{N}{a}}$.
This follows from the elimination rule and
$\coftype[\Psi']{\td{M}{\psi}}{\picl{a}{\td{A}{\psi}}{\td{B}{\psi}}}$.

The eta rule will now follow from \cref{lem:coftype-ceqtm} once we show that for
any $\psitd$,
$\inper[\Psi']{\td{M}{\psi}}{\lam{a}{\app{\td{M}{\psi}}{a}}}
{\picl{a}{\td{A}{\psi}}{\td{B}{\psi}}}$.
Since $\td{M}{\psi} \evals \lam{a}{O}$, this follows when
$\eqtm[\Psi']{\oft{a}{\td{A}{\psi}}}{O}{\app{\td{M}{\psi}}{a}}{\td{B}{\psi}}$,
that is, when for any $\msubsts{\Psi''}{\psi'}{\Psi'}$ and
$\ceqtm[\Psi'']{N}{N'}{\td{A}{\psi\psi'}}$,
$\ceqtm[\Psi'']{\subst{\td{O}{\psi'}}{N}{a}}
{\app{\td{M}{\psi\psi'}}{N'}}{\subst{\td{B}{\psi\psi'}}{N}{a}}$.

Both sides have this type (because
$\oftype[\Psi']{\oft{a}{\td{A}{\psi}}}{O}{\td{B}{\psi}}$),
so by \cref{lem:coftype-ceqtm} it suffices to show
that for any $\msubsts{\Psi'''}{\psi''}{\Psi''}$,
$\inper[\Psi''']{\subst{\td{O}{\psi'\psi''}}{\td{N}{\psi''}}{a}}
{\app{\td{M}{\psi\psi'\psi''}}{\td{N'}{\psi''}}}
{\td{\subst{\td{B}{\psi\psi'}}{N}{a}}{\psi''}}$.
By $\coftype{M}{\picl{a}{A}{B}}$,
$\td{M}{\psi\psi'\psi''}\evals \lam{a}{O''}$ and
$\eqtm[\Psi''']{\oft{a}{\td{A}{\psi\psi'\psi''}}}{\td{O}{\psi'\psi''}}{O''}
{\td{B}{\psi\psi'\psi''}}$.
Thus $\app{\td{M}{\psi\psi'\psi''}}{\td{N'}{\psi''}} \steps^*
\app{\lam{a}{O''}}{\td{N'}{\psi''}} \steps
\subst{O''}{\td{N'}{\psi''}}{a}$, and
$\inper[\Psi''']
{\subst{\td{O}{\psi'\psi''}}{\td{N}{\psi''}}{a}}
{\subst{O''}{\td{N'}{\psi''}}{a}}
{\td{\subst{\td{B}{\psi\psi'}}{N}{a}}{\psi''}}$.

\paragraph{Kan}
$\ceqpretype{\picl{a}{A}{B}}{\picl{a}{A'}{B'}}$ are equally Kan.

For the first Kan condition, we must show that for any $\psitd$, if
\begin{enumerate}
\item $\ceqtm[\Psi']{M}{O}{\picl{a}{\td{A}{\psi}}{\td{B}{\psi}}}$,
\item $\ceqtmres[\Psi',y]{r_i=\e,r_j=\e'}{N^\e_i}{N^{\e'}_j}
{\picl{a}{\td{A}{\psi}}{\td{B}{\psi}}}$
for any $i,j,\e,\e'$,
\item $\ceqtmres[\Psi',y]{r_i=\e}{N^\e_i}{P^\e_i}
{\picl{a}{\td{A}{\psi}}{\td{B}{\psi}}}$
for any $i,\e$, and
\item $\ceqtmres[\Psi']{r_i=\e}{\dsubst{N^\e_i}{r}{y}}{M}
{\picl{a}{\td{A}{\psi}}{\td{B}{\psi}}}$
for any $i,\e$,
\end{enumerate}
then $\ceqtm[\Psi']{\hcomgeneric{\etc{r_i}}
  {\picl{a}{\td{A}{\psi}}{\td{B}{\psi}}}}
{\hcom{\etc{r_i}}{\picl{a}{\td{A'}{\psi}}{\td{B'}{\psi}}}
  {r}{r'}{O}{\etc{y.P^\e_i}}}
{\picl{a}{\td{A}{\psi}}{\td{B}{\psi}}}$.

By \cref{lem:expansion} on both sides, it suffices to show that
\[\begin{aligned}
& \lam{a}{\hcom{\etc{r_i}}{\td{B}{\psi}}{r}{r'}
  {\app{M}{a}}{\etc{y.\app{N^\e_i}{a}}}} \\
\ceqtmtab[\Psi']{}
{\lam{a}{\hcom{\etc{r_i}}{\td{B'}{\psi}}{r}{r'}
  {\app{O}{a}}{\etc{y.\app{P^\e_i}{a}}}}}
{\picl{a}{\td{A}{\psi}}{\td{B}{\psi}}}
\end{aligned}\]
By the introduction rule for dependent functions, it suffices to show
the bodies of these lambdas are
$\eqtm[\Psi']{\oft{a}{\td{A}{\psi}}}{-}{-}{\td{B}{\psi}}$.
That is, for any
$\msubsts{\Psi''}{\psi'}{\Psi'}$ and
$\ceqtm[\Psi'']{Q}{Q'}{\td{A}{\psi\psi'}}$,
\[\begin{aligned}
& \hcom{\etc{\td{r_i}{\psi'}}}
  {\subst{\td{B}{\psi\psi'}}{Q}{a}}
  {\td{r}{\psi'}}{\td{r'}{\psi'}}
{\app{\td{M}{\psi'}}{Q}}
{\etc{y.\app{\td{N^\e_i}{\psi'}}{Q}}} \\
\ceqtmtab[\Psi'']{}
{\hcom{\etc{\td{r_i}{\psi'}}}
  {\subst{\td{B'}{\psi\psi'}}{Q'}{a}}
  {\td{r}{\psi'}}{\td{r'}{\psi'}}
  {\app{\td{O}{\psi'}}{Q'}}
  {\etc{y.\app{\td{P^\e_i}{\psi'}}{Q'}}}}
{\subst{\td{B}{\psi\psi'}}{Q}{a}}
\end{aligned}\]
Since $\eqtype{\oft{a}{A}}{B}{B'}$, we know
$\ceqtype[\Psi'']
{\subst{\td{B}{\psi\psi'}}{Q}{a}}
{\subst{\td{B'}{\psi\psi'}}{Q'}{a}}$ so these types are equally Kan, and thus
the above $\hcomsym$s are equal so long as:

\begin{enumerate}
\item $\ceqtm[\Psi'']
{\app{\td{M}{\psi'}}{Q}}{\app{\td{O}{\psi'}}{Q'}}
{\subst{\td{B}{\psi\psi'}}{Q}{a}}$,

\item $\ceqtmres[\Psi'',y]{\td{r_i}{\psi'}=\e,\td{r_j}{\psi'}=\e'}
{\app{\td{N^\e_i}{\psi'}}{Q}}{\app{\td{N^{\e'}_j}{\psi'}}{Q}}
{\subst{\td{B}{\psi\psi'}}{Q}{a}}$
for any $i,j,\e,\e'$,

\item $\ceqtmres[\Psi'',y]{\td{r_i}{\psi'}=\e}
{\app{\td{N^\e_i}{\psi'}}{Q}}{\app{\td{P^\e_i}{\psi'}}{Q'}}
{\subst{\td{B}{\psi\psi'}}{Q}{a}}$
for any $i,\e$, and

\item $\ceqtmres[\Psi'']{\td{r_i}{\psi'}=\e}
{\app{\td{\dsubst{N^\e_i}{r}{y}}{\psi'}}{Q}}{\app{\td{M}{\psi'}}{Q}}
{\subst{\td{B}{\psi\psi'}}{Q}{a}}$
for any $i,\e$ (since $y\fresh Q$).
\end{enumerate}
These follow from our hypotheses and the elimination rule for dependent
functions; the context-restricted judgments follow from the fact that dimension
substitutions push into the subterms of $\app--$. For example,
$\ceqtmres[\Psi'',y]{\td{r_i}{\psi'}=\e}
{\app{\td{N^\e_i}{\psi'}}{Q}}{\app{\td{P^\e_i}{\psi'}}{Q'}}
{\subst{\td{B}{\psi\psi'}}{Q}{a}}$
means that for any $\msubsts{\Psi'''}{\psi''}{(\Psi'',y)}$ such that
$\td{r_i}{\psi'\psi''}=\e$,
$\ceqtm[\Psi''']
{\app{\td{N^\e_i}{\psi'\psi''}}{\td{Q}{\psi''}}}
{\app{\td{P^\e_i}{\psi'\psi''}}{\td{Q'}{\psi''}}}
{\td{\subst{\td{B}{\psi\psi'}}{Q}{a}}{\psi''}}$.
We know that
$\ceqtmres[\Psi',y]{r_i=\e}{N^\e_i}{P^\e_i}
{\picl{a}{\td{A}{\psi}}{\td{B}{\psi}}}$
and $\psi'\psi''$ satisfies $r_i=\e$, so
$\ceqtm[\Psi''']{\td{N^\e_i}{\psi'\psi''}}{\td{P^\e_i}{\psi'\psi''}}
{\picl{a}{\td{A}{\psi\psi'\psi''}}{\td{B}{\psi\psi'\psi''}}}$,
$\ceqtm[\Psi''']{\td{Q}{\psi''}}{\td{Q'}{\psi''}}{\td{A}{\psi\psi'\psi''}}$,
and the elimination rule applies.

The second Kan condition requires that for any $\psitd$, if
\begin{enumerate}
\item $\coftype[\Psi']{M}{\picl{a}{\td{A}{\psi}}{\td{B}{\psi}}}$,
\item $\ceqtmres[\Psi',y]{r_i=\e,r_j=\e'}{N^\e_i}{N^{\e'}_j}
{\picl{a}{\td{A}{\psi}}{\td{B}{\psi}}}$
for any $i,j,\e,\e'$, and
\item $\ceqtmres[\Psi']{r_i=\e}{\dsubst{N^\e_i}{r}{y}}{M}
{\picl{a}{\td{A}{\psi}}{\td{B}{\psi}}}$
for any $i,\e$,
\end{enumerate}
then $\ceqtm[\Psi']{\hcom{\etc{r_i}}{\picl{a}{\td{A}{\psi}}{\td{B}{\psi}}}
{r}{r}{M}{\etc{y.N^\e_i}}}
{M}{\picl{a}{\td{A}{\psi}}{\td{B}{\psi}}}$.

By \cref{lem:expansion}, it suffices to show
\[
\ceqtm[\Psi']
{\lam{a}{\hcom{\etc{r_i}}{\td{B}{\psi}}{r}{r}
  {\app{M}{a}}{\etc{y.\app{N^\e_i}{a}}}}}
{M}{\picl{a}{\td{A}{\psi}}{\td{B}{\psi}}}
\]
By the eta and introduction rules for dependent functions, it suffices to show
that for any $\msubsts{\Psi''}{\psi'}{\Psi'}$ and
$\ceqtm[\Psi'']{O}{O'}{\td{A}{\psi\psi'}}$,
\[
\ceqtm[\Psi'']
{\hcom{\etc{\td{r_i}{\psi'}}}
{\subst{\td{B}{\psi\psi'}}{O}{a}}{\td{r}{\psi'}}{\td{r}{\psi'}}
  {\app{\td{M}{\psi'}}{O}}{\etc{y.\app{\td{N^\e_i}{\psi'}}{O}}}}
{\app{\td{M}{\psi'}}{O'}}
{\subst{\td{B}{\psi}}{O}{a}}
\]
We apply the second Kan condition of
$\cwftype[\Psi'']{\subst{\td{B}{\psi\psi'}}{O}{a}}$, deriving its hypotheses
from the elimination rule for dependent functions, as in the previous case.
We conclude the above $\hcomsym$ is equal to $\app{\td{M}{\psi'}}{O}$, which is
equal to $\app{\td{M}{\psi'}}{O'}$ again by the elimination rule.

The third Kan condition asserts that for any $\psitd$,
if $r_i = \e$ for some $i$,
\begin{enumerate}
\item $\coftype[\Psi']{M}{\picl{a}{\td{A}{\psi}}{\td{B}{\psi}}}$,
\item $\ceqtmres[\Psi',y]{r_i=\e,r_j=\e'}{N^\e_i}{N^{\e'}_j}
{\picl{a}{\td{A}{\psi}}{\td{B}{\psi}}}$
for any $i,j,\e,\e'$, and
\item $\ceqtmres[\Psi']{r_i=\e}{\dsubst{N^\e_i}{r}{y}}{M}
{\picl{a}{\td{A}{\psi}}{\td{B}{\psi}}}$
for any $i,\e$,
\end{enumerate}
then $\ceqtm[\Psi']{\hcom{\etc{r_i}}{\picl{a}{\td{A}{\psi}}{\td{B}{\psi}}}
{r}{r'}{M}{\etc{y.N^\e_i}}}
{\dsubst{N^\e_i}{r'}{y}}{\picl{a}{\td{A}{\psi}}{\td{B}{\psi}}}$.

As in the previous case, by \cref{lem:expansion} and the eta and introduction
rules for dependent functions, it suffices to show that for any
$\msubsts{\Psi''}{\psi'}{\Psi'}$ and $\ceqtm[\Psi'']{O}{O'}{\td{A}{\psi\psi'}}$,
\[
\ceqtm[\Psi'']
{\hcom{\etc{\td{r_i}{\psi'}}}
{\subst{\td{B}{\psi\psi'}}{O}{a}}{\td{r}{\psi'}}{\td{r'}{\psi'}}
  {\app{\td{M}{\psi'}}{O}}{\etc{y.\app{\td{N^\e_i}{\psi'}}{O}}}}
{\app{\td{\dsubst{N^\e_i}{r'}{y}}{\psi'}}{O'}}
{\subst{\td{B}{\psi}}{O}{a}}
\]
By the third Kan condition of
$\cwftype[\Psi'']{\subst{\td{B}{\psi\psi'}}{O}{a}}$ and the elimination rule,
the above $\hcomsym$ is equal to $\app{\td{\dsubst{N^\e_i}{r'}{y}}{\psi'}}{O}$,
which is equal to $\app{\td{\dsubst{N^\e_i}{r'}{y}}{\psi'}}{O'}$ by the
elimination rule.

The fourth Kan condition asserts that for any
$\msubsts{(\Psi',x)}{\psi}{\Psi}$, if
\[
\ceqtm[\Psi']{M}{M'}
{\picl{a}{\dsubst{\td{A}{\psi}}{r}{x}}{\dsubst{\td{B}{\psi}}{r}{x}}}
\]
then
\[
\ceqtm[\Psi']
{\coe{x.\picl{a}{\td{A}{\psi}}{\td{B}{\psi}}}{r}{r'}{M}}
{\coe{x.\picl{a}{\td{A'}{\psi}}{\td{B'}{\psi}}}{r}{r'}{M'}}
{\picl{a}{\dsubst{\td{A}{\psi}}{r'}{x}}{\dsubst{\td{B}{\psi}}{r'}{x}}}.
\]

By \cref{lem:expansion} on both sides and the introduction rule for dependent
functions, it suffices to show that for any $\msubsts{\Psi''}{\psi'}{\Psi'}$ and
$\ceqtm[\Psi'']{N}{N'}{\dsubst{\td{A}{\psi\psi'}}{\td{r'}{\psi'}}{x}}$,
\[\begin{aligned}
& \coe{x.\subst{\td{B}{\psi\psi'}}{\coe{x.\td{A}{\psi\psi'}}{\td{r'}{\psi'}}{x}{N}}{a}}{\td{r}{\psi'}}{\td{r'}{\psi'}}
  {\app{\td{M}{\psi'}}{\coe{x.\td{A}{\psi\psi'}}{\td{r'}{\psi'}}{\td{r}{\psi'}}{N}}} \\
\ceqtmtab[\Psi'']
{}
{\coe{x.\subst{\td{B'}{\psi\psi'}}{\coe{x.\td{A'}{\psi\psi'}}{\td{r'}{\psi'}}{x}{N'}}{a}}{\td{r}{\psi'}}{\td{r'}{\psi'}}
  {\app{\td{M'}{\psi'}}{\coe{x.\td{A'}{\psi\psi'}}{\td{r'}{\psi'}}{\td{r}{\psi'}}{N'}}}}
{\subst{\dsubst{\td{B}{\psi\psi'}}{\td{r'}{\psi'}}{x}}{N}{a}}.
\end{aligned}\]
By the fourth Kan condition of $\ceqtype{A}{A'}$,
$\ceqtm[\Psi'',x]
{\coe{x.\td{A}{\psi\psi'}}{\td{r'}{\psi'}}{x}{N}}
{\coe{x.\td{A'}{\psi\psi'}}{\td{r'}{\psi'}}{x}{N'}}
{\td{A}{\psi\psi'}}$, so
$\ceqtype[\Psi'',x]
{\subst{\td{B}{\psi\psi'}}{\coe{x.\td{A}{\psi\psi'}}{\td{r'}{\psi'}}{x}{N}}{a}}
{\subst{\td{B'}{\psi\psi'}}{\coe{x.\td{A'}{\psi\psi'}}{\td{r'}{\psi'}}{x}{N'}}{a}}$.
By the fourth Kan condition of these types and of $\ceqtype{A}{A'}$, and the
elimination rule for dependent functions, the $\coesym$ of interest above are
$\ceqtm[\Psi'']{-}{-}
{\dsubst{\subst{\td{B}{\psi\psi'}}{\coe{x.\td{A}{\psi\psi'}}{\td{r'}{\psi'}}{x}{N}}{a}}{\td{r'}{\psi'}}{x}}$.
By the fifth Kan condition of $\ceqtype{A}{A'}$,
$\ceqtm[\Psi'']
{\coe{x.\td{A}{\psi\psi'}}{\td{r'}{\psi'}}{\td{r'}{\psi'}}{N}}
{N}{\dsubst{\td{A}{\psi\psi'}}{\td{r'}{\psi'}}{x}}$, so
$\ceqtype[\Psi'']
{\dsubst{\subst{\td{B}{\psi\psi'}}{\coe{x.\td{A}{\psi\psi'}}{\td{r'}{\psi'}}{x}{N}}{a}}{\td{r'}{\psi'}}{x}}
{\subst{\dsubst{\td{B}{\psi\psi'}}{\td{r'}{\psi'}}{x}}{N}{a}}$ and the result
follows.

The fifth Kan condition asserts that for any
$\msubsts{(\Psi',x)}{\psi}{\Psi}$, if
$\coftype[\Psi']{M}
{\picl{a}{\dsubst{\td{A}{\psi}}{r}{x}}{\dsubst{\td{B}{\psi}}{r}{x}}}$,
then
$\ceqtm[\Psi']
{\coe{x.\picl{a}{\td{A}{\psi}}{\td{B}{\psi}}}{r}{r}{M}}{M}
{\picl{a}{\dsubst{\td{A}{\psi}}{r}{x}}{\dsubst{\td{B}{\psi}}{r}{x}}}$.

By \cref{lem:expansion} and the eta and introduction rules for dependent
functions, it suffices to show that for any $\msubsts{\Psi''}{\psi'}{\Psi'}$ and
$\ceqtm[\Psi'']{N}{N'}{\dsubst{\td{A}{\psi\psi'}}{\td{r}{\psi'}}{x}}$,
\[
\ceqtm[\Psi'']
{\coe{x.\subst{\td{B}{\psi\psi'}}{\coe{x.\td{A}{\psi\psi'}}{\td{r}{\psi'}}{x}{N}}{a}}{\td{r}{\psi'}}{\td{r}{\psi'}}
  {\app{\td{M}{\psi'}}{\coe{x.\td{A}{\psi\psi'}}{\td{r}{\psi'}}{\td{r}{\psi'}}{N}}}}
{\app{\td{M}{\psi'}}{N'}}
{\subst{\dsubst{\td{B}{\psi\psi'}}{\td{r}{\psi'}}{x}}{N}{a}}.
\]
As above, by the fourth Kan condition of $\ceqtype{A}{A'}$,
$\coftype[\Psi'',x]
{\coe{x.\td{A}{\psi\psi'}}{\td{r}{\psi'}}{x}{N}}
{\td{A}{\psi\psi'}}$, so
$\cwftype[\Psi'',x]
{\subst{\td{B}{\psi\psi'}}{\coe{x.\td{A}{\psi\psi'}}{\td{r}{\psi'}}{x}{N}}{a}}$.
By the fifth Kan condition of $\ceqtype{A}{A'}$ and transitivity of $\eq$,
$\ceqtm[\Psi'']
{\coe{x.\td{A}{\psi\psi'}}{\td{r}{\psi'}}{\td{r}{\psi'}}{N}}{N'}
{\dsubst{\td{A}{\psi\psi'}}{\td{r}{\psi'}}{x}}$.
Then by the fifth Kan condition of
$\cwftype[\Psi'',x]
{\subst{\td{B}{\psi\psi'}}{\coe{x.\td{A}{\psi\psi'}}{\td{r}{\psi'}}{x}{N}}{a}}$
and the elimination rule for dependent functions, the $\coesym$ on the left-hand
side above is
$\ceqtm[\Psi'']{-}
{\app{\td{M}{\psi'}}{N'}}
{\dsubst{\subst{\td{B}{\psi\psi'}}{\coe{x.\td{A}{\psi\psi'}}{\td{r}{\psi'}}{x}{N}}{a}}{\td{r}{\psi'}}{x}}$,
and this type is again equal to
$\subst{\dsubst{\td{B}{\psi\psi'}}{\td{r}{\psi'}}{x}}{N}{a}$ so the result
follows.

\paragraph{Cubical}
For any $\psitd$ and
$\vinper[\Psi']{M}{M'}{\picl{a}{\td{A}{\psi}}{\td{B}{\psi}}}$, we have
$\ceqtm[\Psi']{M}{M'}{\picl{a}{\td{A}{\psi}}{\td{B}{\psi}}}$.

Then $M = \lam{a}{N}$, $M' = \lam{a}{N'}$, and
$\eqtm[\Psi']{\oft{a}{\td{A}{\psi}}}{N}{N'}{\td{B}{\psi}}$,
and the result follows from the introduction rule for functions.

\subsection{Dependent pairs}

If a cubical type system has $\ceqtype{A}{A'}$ and
$\eqtype{\oft{a}{A}}{B}{B'}$, we say it
\emph{has their dependent pair type} when for all $\psitd$,
$\veqper[\Psi']{\sigmacl{a}{\td{A}{\psi}}{\td{B}{\psi}}}
{\sigmacl{a}{\td{A'}{\psi}}{\td{B'}{\psi}}}$, and
$\vinper[\Psi']{-}{-}{\sigmacl{a}{\td{A}{\psi}}{\td{B}{\psi}}}$
is the least relation such that
\[
\vinper[\Psi']{\pair{M}{N}}{\pair{M'}{N'}}
{\sigmacl{a}{\td{A}{\psi}}{\td{B}{\psi}}}
\]
when $\ceqtm[\Psi']{M}{M'}{\td{A}{\psi}}$ and
$\ceqtm[\Psi']{N}{N'}{\subst{\td{B}{\psi}}{M}{a}}$.

This relation is symmetric because
$\ceqtype[\Psi']{\subst{\td{B}{\psi}}{M}{a}}{\subst{\td{B}{\psi}}{M'}{a}}$, so
$\ceqtm[\Psi']{N'}{N}{\subst{\td{B}{\psi}}{M'}{a}}$.
The PERs
$\vinper[\Psi']{-}{-}{\sigmacl{a}{\td{A}{\psi}}{\td{B}{\psi}}}$ and
$\vinper[\Psi']{-}{-}{\sigmacl{a}{\td{A'}{\psi}}{\td{B'}{\psi}}}$ are equal
because $\ceqtm[\Psi']{M}{M'}{\td{A}{\psi}}$ if and only if
$\ceqtm[\Psi']{M}{M'}{\td{A'}{\psi}}$, and
$\ceqtm[\Psi']{N}{N'}{\subst{\td{B}{\psi}}{M}{a}}$ if and only if
$\ceqtm[\Psi']{N}{N'}{\subst{\td{B'}{\psi}}{M}{a}}$.

In the remainder of this subsection, we assume we are working with a cubical
type system that has $\ceqtype{A}{A'}$, $\eqtype{\oft aA}{B}{B'}$, and their
dependent pair type.

\paragraph{Pretype}
$\ceqpretype{\sigmacl{a}{A}{B}}{\sigmacl{a}{A'}{B'}}$.

For any $\msubsts{\Psi_1}{\psi_1}{\Psi}$ and $\msubsts{\Psi_2}{\psi_2}{\Psi_1}$,
$\veqper[\Psi_2]
{\sigmacl{a}{\td{A}{\psi_1\psi_2}}{\td{B}{\psi_1\psi_2}}}
{\sigmacl{a}{\td{A'}{\psi_1\psi_2}}{\td{B'}{\psi_1\psi_2}}}$ because the cubical
type system has this dependent pair type.

\paragraph{Introduction}
If $\ceqtm{M}{M'}{A}$ and $\ceqtm{N}{N'}{\subst{B}{M}{a}}$, then
$\ceqtm{\pair{M}{N}}{\pair{M'}{N'}}{\sigmacl{a}{A}{B}}$.

Each side has coherent aspects up to syntactic equality, since
$\wfval[\Psi']{\pair{\td{M}{\psi}}{\td{N}{\psi}}}$ for all $\psitd$.
Thus it suffices to show
$\vinper[\Psi_2]
{\pair{\td{M}{\psi_1\psi_2}}{\td{N}{\psi_1\psi_2}}}
{\pair{\td{M'}{\psi_1\psi_2}}{\td{N'}{\psi_1\psi_2}}}
{\sigmacl{a}{\td{A}{\psi_1\psi_2}}{\td{B}{\psi_1\psi_2}}}$,
which is true because
$\ceqtm[\Psi_2]{\td{M}{\psi_1\psi_2}}{\td{M'}{\psi_1\psi_2}}
{\td{A}{\psi_1\psi_2}}$ and
$\ceqtm[\Psi_2]{\td{N}{\psi_1\psi_2}}{\td{N'}{\psi_1\psi_2}}
{\subst{\td{B}{\psi_1\psi_2}}{\td{M}{\psi_1\psi_2}}{a}}$.

\paragraph{Elimination}
If $\ceqtm{M}{M'}{\sigmacl{a}{A}{B}}$, then
$\ceqtm{\fst{M}}{\fst{M'}}{A}$ and
$\ceqtm{\snd{M}}{\snd{M'}}{\subst{B}{\fst{M}}{a}}$.

For any $\msubsts{\Psi_1}{\psi_1}{\Psi}$ and $\msubsts{\Psi_2}{\psi_2}{\Psi_1}$,
we know $\td{M}{\psi_1}\evals \pair{O_1}{P_1}$ where
$\coftype[\Psi_1]{O_1}{\td{A}{\psi_1}}$ and
$\coftype[\Psi_1]{P_1}{\subst{\td{B}{\psi_1}}{O_1}{a}}$.
Thus $\fst{\td{M}{\psi_1}}\steps^* \fst{\pair{O_1}{P_1}} \steps O_1$
where $O_1 \evals X_1$ and
$\inper[\Psi_2]{\td{X_1}{\psi_2}}{\td{O_1}{\psi_2}}{\td{A}{\psi_1\psi_2}}$.
We also know
$\vinper[\Psi_2]
{\td{M}{\psi_1\psi_2}\evals \pair{O_{12}}{P_{12}}}
{\pair{\td{O_1}{\psi_2}}{\td{P_1}{\psi_2}}}
{\sigmacl{a}{\td{A}{\psi_1\psi_2}}{\td{B}{\psi_1\psi_2}}}$, so
$\ceqtm[\Psi_2]{O_{12}}{\td{O_1}{\psi_2}}{\td{A}{\psi_1\psi_2}}$ and
$\ceqtm[\Psi_2]{P_{12}}{\td{P_1}{\psi_2}}
{\subst{\td{B}{\psi_1\psi_2}}{O_{12}}{a}}$.
Thus $\fst{\td{M}{\psi_1\psi_2}}\steps^* \fst{\pair{O_{12}}{P_{12}}}
\steps O_{12}$.
We want to show
$\inper[\Psi_2]{\td{X_1}{\psi_2}}{O_{12}}{\td{A}{\psi_1\psi_2}}$,
which follows from
$\inper[\Psi_2]{O_{12}}{\td{O_1}{\psi_2}}{\td{A}{\psi_1\psi_2}}$ and
transitivity.

An analogous argument shows $\fst{M'}$ also has coherent aspects. To see that
the aspects of $\fst{M}$ and $\fst{M'}$ are related to each other,
we use the fact that $\td{M'}{\psi_1\psi_2} \evals \pair{O_{12}'}{P_{12}'}$ and
$\ceqtm[\Psi_2]{O_{12}}{O_{12}'}{\td{A}{\psi_1}}$, so
$\fst{\td{M'}{\psi_1\psi_2}} \steps^* O_{12}'$ and
$\inper[\Psi_2]{O_{12}}{O_{12}'}{\td{A}{\psi_1\psi_2}}$.

For the second elimination rule,
$\snd{\td{M}{\psi_1}}\steps^* \snd{\pair{O_1}{P_1}} \steps P_1$
where $P_1 \evals Y_1$ and
$\inper[\Psi_2]{\td{Y_1}{\psi_2}}{\td{P_1}{\psi_2}}
{\subst{\td{B}{\psi_1\psi_2}}{\td{O_1}{\psi_2}}{a}}$; and
$\snd{\td{M}{\psi_1\psi_2}}\steps^* \snd{\pair{O_{12}}{P_{12}}}
\steps P_{12}$.
We want to show
$\inper[\Psi_2]{\td{Y_1}{\psi_2}}{P_{12}}
{\subst{\td{B}{\psi_1\psi_2}}{\fst{\td{M}{\psi_1\psi_2}}}{a}}$,
which follows from
$\inper[\Psi_2]{P_{12}}{\td{P_1}{\psi_2}}
{\subst{\td{B}{\psi_1\psi_2}}{O_{12}}{a}}$ and transitivity
once we show
$\eqper[\Psi_2]
{\subst{\td{B}{\psi_1\psi_2}}{O_{12}}{a}}
{\subst{\td{B}{\psi_1\psi_2}}{\td{O_1}{\psi_2}}{a}}$ and
$\eqper[\Psi_2]
{\subst{\td{B}{\psi_1\psi_2}}{O_{12}}{a}}
{\subst{\td{B}{\psi_1\psi_2}}{\fst{\td{M}{\psi_1\psi_2}}}{a}}$.

The former holds by
$\ceqtm[\Psi_2]{O_{12}}{\td{O_1}{\psi_2}}{\td{A}{\psi_1\psi_2}}$,
and the latter holds by
$\ceqtm[\Psi_2]{O_{12}}{\fst{\td{M}{\psi_1\psi_2}}}{\td{A}{\psi_1\psi_2}}$,
which follows from \cref{lem:coftype-ceqtm} because both sides are
$\coftype[\Psi_2]{-}{\td{A}{\psi_1\psi_2}}$ (using the first elimination rule)
and for any $\msubsts{\Psi'}{\psi}{\Psi_2}$,
$\inper[\Psi']{\td{O_{12}}{\psi}}
{\fst{\td{M}{\psi_1\psi_2\psi}}}{\td{A}{\psi_1\psi_2\psi}}$
(since $\fst{\td{M}{\psi_1\psi_2\psi}} \steps^* O$ where
$\ceqtm[\Psi']{\td{O_{12}}{\psi}}{O}{\td{A}{\psi_1\psi_2\psi}}$).
The aspects of $\snd{M'}$ are coherent and related to those of $\snd{M}$ by
the same argument as for $\fst{-}$.

\paragraph{Computation}
If $\coftype{M}{A}$ and $\coftype{N}{\subst{B}{M}{a}}$, then
$\ceqtm{\fst{\pair{M}{N}}}{M}{A}$ and
$\ceqtm{\snd{\pair{M}{N}}}{N}{\subst{B}{M}{a}}$.

Both follow from \cref{lem:expansion}.

\paragraph{Eta}
If $\coftype{M}{\sigmacl{a}{A}{B}}$, then
$\ceqtm{M}{\pair{\fst{M}}{\snd{M}}}{\sigmacl{a}{A}{B}}$.

The elimination and introduction rules for dependent pairs imply that
$\coftype{\pair{\fst{M}}{\snd{M}}}{\sigmacl{a}{A}{B}}$.
Thus by \cref{lem:coftype-ceqtm} it suffices to show that for any $\psitd$,
$\inper[\Psi']{\td{M}{\psi}}{\pair{\fst{\td{M}{\psi}}}{\snd{\td{M}{\psi}}}}
{\sigmacl{a}{\td{A}{\psi}}{\td{B}{\psi}}}$.
Since $\td{M}{\psi}\evals\pair{O}{P}$, this follows when
$\ceqtm[\Psi']{O}{\fst{\td{M}{\psi}}}{\td{A}{\psi}}$ and
$\ceqtm[\Psi']{P}{\snd{\td{M}{\psi}}}{\subst{\td{B}{\psi}}{O}{a}}$.

To show the former, we apply \cref{lem:coftype-ceqtm} and show that for any
$\msubsts{\Psi''}{\psi'}{\Psi'}$,
$\inper[\Psi'']{\td{O}{\psi'}}{\fst{\td{M}{\psi\psi'}}}{\td{A}{\psi\psi'}}$.
By $\coftype{M}{\sigmacl{a}{A}{B}}$,
$\td{M}{\psi\psi'}\evals\pair{O'}{P'}$ where
$\ceqtm[\Psi'']{O'}{\td{O}{\psi'}}{\td{A}{\psi\psi'}}$ and
$\ceqtm[\Psi'']{P'}{\td{P}{\psi'}}{\subst{\td{B}{\psi\psi'}}{O'}{a}}$.
Thus $\fst{\td{M}{\psi\psi'}} \steps^* \fst{\pair{O'}{P'}} \steps O'$ and
$\inper[\Psi'']{O'}{\td{O}{\psi'}}{\td{A}{\psi\psi'}}$.

To show the latter, we apply \cref{lem:coftype-ceqtm} and show that for any
$\msubsts{\Psi''}{\psi'}{\Psi'}$,
$\inper[\Psi'']{\td{P}{\psi'}}{\snd{\td{M}{\psi\psi'}}}
{\subst{\td{B}{\psi\psi'}}{\td{O}{\psi'}}{a}}$.
But $\snd{\td{M}{\psi\psi'}} \steps^* \snd{\pair{O'}{P'}} \steps P'$ and
$\inper[\Psi'']{P'}{\td{P}{\psi'}}{\subst{\td{B}{\psi\psi'}}{O'}{a}}$.

\paragraph{Kan}
$\ceqpretype{\sigmacl{a}{A}{B}}{\sigmacl{a}{A'}{B'}}$ are equally Kan.

For the first Kan condition, we must show that for any $\psitd$, if
\begin{enumerate}
\item $\ceqtm[\Psi']{M}{O}{\sigmacl{a}{\td{A}{\psi}}{\td{B}{\psi}}}$,
\item $\ceqtmres[\Psi',y]{r_i=\e,r_j=\e'}{N^\e_i}{N^{\e'}_j}
{\sigmacl{a}{\td{A}{\psi}}{\td{B}{\psi}}}$
for any $i,j,\e,\e'$,
\item $\ceqtmres[\Psi',y]{r_i=\e}{N^\e_i}{P^\e_i}
{\sigmacl{a}{\td{A}{\psi}}{\td{B}{\psi}}}$
for any $i,\e$, and
\item $\ceqtmres[\Psi']{r_i=\e}{\dsubst{N^\e_i}{r}{y}}{M}
{\sigmacl{a}{\td{A}{\psi}}{\td{B}{\psi}}}$
for any $i,\e$,
\end{enumerate}
then $\ceqtm[\Psi']{\hcomgeneric{\etc{r_i}}
  {\sigmacl{a}{\td{A}{\psi}}{\td{B}{\psi}}}}
{\hcom{\etc{r_i}}{\sigmacl{a}{\td{A'}{\psi}}{\td{B'}{\psi}}}
  {r}{r'}{O}{\etc{y.P^\e_i}}}
{\picl{a}{\td{A}{\psi}}{\td{B}{\psi}}}$.

By \cref{lem:expansion} on both sides, it suffices to show that
\begin{small}\[\begin{aligned}
&\pair{\hcom{\etc{r_i}}{\td{A}{\psi}}{r}{r'}{\fst{M}}{\etc{y.\fst{N^\e_i}}}}
      {\com{\etc{r_i}}{z.\subst{\td{B}{\psi}}{F}{a}}
        {r}{r'}{\snd{M}}{\etc{y.\snd{N^\e_i}}}} \\
\ceqtmtab[\Psi']{}
{\pair{\hcom{\etc{r_i}}{\td{A'}{\psi}}{r}{r'}{\fst{O}}{\etc{y.\fst{P^\e_i}}}}
      {\com{\etc{r_i}}{z.\subst{\td{B'}{\psi}}{F'}{a}}
        {r}{r'}{\snd{O}}{\etc{y.\snd{P^\e_i}}}}}
{\sigmacl{a}{\td{A}{\psi}}{\td{B}{\psi}}}
\end{aligned}\]\end{small}
where
\[\begin{aligned}
F &= \hcom{\etc{r_i}}{\td{A}{\psi}}{r}{z}{\fst{M}}{\etc{y.\fst{N^\e_i}}} \\
F' &= \hcom{\etc{r_i}}{\td{A'}{\psi}}{r}{z}{\fst{O}}{\etc{y.\fst{P^\e_i}}} \\
\end{aligned}\]
By the introduction rule for dependent pairs, it suffices to show the first
components are $\ceqtm[\Psi']{-}{-}{\td{A}{\psi}}$ and the second components are
$\ceqtm[\Psi']{-}{-}{\subst{\td{B}{\psi}}
{\hcom{\etc{r_i}}{\td{A}{\psi}}{r}{r'}{\fst{M}}{\etc{y.\fst{N^\e_i}}}}{a}}$.
The first components follow from the first Kan condition of
$\ceqtype[\Psi']{\td{A}{\psi}}{\td{A'}{\psi}}$ and the first elimination rule
for dependent pairs (again, deriving the context-restricted judgments by pushing
dimension substitutions into $\fst{-}$).
For the same reason, $\ceqtm[\Psi',z]{F}{F'}{\td{A}{\psi}}$.

The second components follow from \cref{thm:com} and
$\ceqtype[\Psi',z]
{\subst{\td{B}{\psi}}{F}{a}}
{\subst{\td{B'}{\psi}}{F'}{a}}$: it suffices to show
\begin{enumerate}
\item $\ceqtm[\Psi']{\snd{M}}{\snd{O}}
{\subst{\td{B}{\psi}}{\dsubst{F}{r}{z}}{a}}$,
\item $\ceqtmres[\Psi',y]{r_i=\e,r_j=\e'}{\snd{N^\e_i}}{\snd{N^{\e'}_j}}
{\subst{\td{B}{\psi}}{\dsubst{F}{y}{z}}{a}}$
for any $i,j,\e,\e'$,
\item $\ceqtmres[\Psi',y]{r_i=\e}{\snd{N^\e_i}}{\snd{P^\e_i}}
{\subst{\td{B}{\psi}}{\dsubst{F}{y}{z}}{a}}$
for any $i,\e$, and
\item $\ceqtmres[\Psi']{r_i=\e}{\snd{\dsubst{N^\e_i}{r}{y}}}{\snd{M}}
{\subst{\td{B}{\psi}}{\dsubst{F}{r}{z}}{a}}$
for any $i,\e$.
\end{enumerate}
because it implies
$\ceqtm[\Psi']{\comsym}{\comsym}{\subst{\td{B}{\psi}}{\dsubst{F}{r'}{z}}{a}}$,
and $\dsubst{F}{r'}{z} =
\hcom{\etc{r_i}}{\td{A}{\psi}}{r}{r'}{\fst{M}}{\etc{y.\fst{N^\e_i}}}$.
These all follow from the second elimination rule for dependent pairs,
$\ceqtm[\Psi']{\dsubst{F}{r}{z}}{\fst{M}}{\td{A}{\psi}}$
(by the second Kan condition of $\cwftype[\Psi']{\td{A}{\psi}}$),
and $\ceqtmres[\Psi',y]{r_i=\e}{\dsubst{F}{y}{z}}{\fst{N^\e_i}}{\td{A}{\psi}}$
(by the third Kan condition of $\cwftype[\Psi']{\td{A}{\psi}}$ and the
definition of context restriction).

The second Kan condition requires that for any $\psitd$, if
\begin{enumerate}
\item $\coftype[\Psi']{M}{\sigmacl{a}{\td{A}{\psi}}{\td{B}{\psi}}}$,
\item $\ceqtmres[\Psi',y]{r_i=\e,r_j=\e'}{N^\e_i}{N^{\e'}_j}
{\sigmacl{a}{\td{A}{\psi}}{\td{B}{\psi}}}$
for any $i,j,\e,\e'$, and
\item $\ceqtmres[\Psi']{r_i=\e}{\dsubst{N^\e_i}{r}{y}}{M}
{\sigmacl{a}{\td{A}{\psi}}{\td{B}{\psi}}}$
for any $i,\e$,
\end{enumerate}
then $\ceqtm[\Psi']{\hcom{\etc{r_i}}{\sigmacl{a}{\td{A}{\psi}}{\td{B}{\psi}}}
{r}{r}{M}{\etc{y.N^\e_i}}}
{M}{\sigmacl{a}{\td{A}{\psi}}{\td{B}{\psi}}}$.

By \cref{lem:expansion}, it suffices to show
\begin{small}\[
\ceqtm[\Psi']
{\pair{\hcom{\etc{r_i}}{\td{A}{\psi}}{r}{r}{\fst{M}}{\etc{y.\fst{N^\e_i}}}}
      {\com{\etc{r_i}}{z.\subst{\td{B}{\psi}}{F}{a}}
        {r}{r}{\snd{M}}{\etc{y.\snd{N^\e_i}}}}}
{M}{\sigmacl{a}{\td{A}{\psi}}{\td{B}{\psi}}}
\]\end{small}%
where $F = \hcom{\etc{r_i}}{\td{A}{\psi}}{r}{z}{\fst{M}}{\etc{y.\fst{N^\e_i}}}$.
By the eta and introduction rules for dependent pairs, it suffices to show
that
\[
\ceqtm[\Psi']{\fst{M}}
{\hcom{\etc{r_i}}{\td{A}{\psi}}{r}{r}{\fst{M}}{\etc{y.\fst{N^\e_i}}}}
{\td{A}{\psi}}
\]
which follows from the second Kan condition of $\cwftype[\Psi']{\td{A}{\psi}}$,
and
\[
\ceqtm[\Psi']
{\snd{M}}
{\com{\etc{r_i}}{z.\subst{\td{B}{\psi}}{F}{a}}
  {r}{r}{\snd{M}}{\etc{y.\snd{N^\e_i}}}}
{\subst{\td{B}{\psi}}{\fst{M}}{a}}
\]
which follows from \cref{thm:com}, deriving its hypotheses from the second
elimination rule for dependent pairs as in the previous case, and by
$\ceqtm[\Psi']{\dsubst{F}{r}{z}}{\fst{M}}{\td{A}{\psi}}$
(by the second Kan condition of $\cwftype[\Psi']{\td{A}{\psi}}$).

The third Kan condition asserts that for any $\psitd$,
if $r_i = \e$ for some $i$,
\begin{enumerate}
\item $\coftype[\Psi']{M}{\sigmacl{a}{\td{A}{\psi}}{\td{B}{\psi}}}$,
\item $\ceqtmres[\Psi',y]{r_i=\e,r_j=\e'}{N^\e_i}{N^{\e'}_j}
{\sigmacl{a}{\td{A}{\psi}}{\td{B}{\psi}}}$
for any $i,j,\e,\e'$, and
\item $\ceqtmres[\Psi']{r_i=\e}{\dsubst{N^\e_i}{r}{y}}{M}
{\sigmacl{a}{\td{A}{\psi}}{\td{B}{\psi}}}$
for any $i,\e$,
\end{enumerate}
then $\ceqtm[\Psi']{\hcom{\etc{r_i}}{\sigmacl{a}{\td{A}{\psi}}{\td{B}{\psi}}}
{r}{r'}{M}{\etc{y.N^\e_i}}}
{\dsubst{N^\e_i}{r'}{y}}{\sigmacl{a}{\td{A}{\psi}}{\td{B}{\psi}}}$.

As in the previous case, by \cref{lem:expansion}, it suffices to show
\begin{small}\[
\ceqtm[\Psi']
{\pair{\hcom{\etc{r_i}}{\td{A}{\psi}}{r}{r}{\fst{M}}{\etc{y.\fst{N^\e_i}}}}
      {\com{\etc{r_i}}{z.\subst{\td{B}{\psi}}{F}{a}}
        {r}{r}{\snd{M}}{\etc{y.\snd{N^\e_i}}}}}
{N^\e_i}{\sigmacl{a}{\td{A}{\psi}}{\td{B}{\psi}}}
\]\end{small}%
where $F = \hcom{\etc{r_i}}{\td{A}{\psi}}{r}{z}{\fst{M}}{\etc{y.\fst{N^\e_i}}}$.
The first projection of that pair is equal to $\fst{N^\e_i}$ by the third Kan
condition of $\cwftype[\Psi']{\td{A}{\psi}}$, and the second projection is equal
to $\snd{N^\e_i}$ by \cref{thm:com} and the fact that
$\ceqtm[\Psi']{\dsubst{F}{r'}{z}}{\fst{N^\e_i}}{\td{A}{\psi}}$
(by the third Kan condition of $\cwftype[\Psi']{\td{A}{\psi}}$ and $r_i=\e$).
The result follows by the eta and introduction rules for dependent pairs.

The fourth Kan condition asserts that for any
$\msubsts{(\Psi',x)}{\psi}{\Psi}$, if
\[
\ceqtm[\Psi']{M}{M'}
{\sigmacl{a}{\dsubst{\td{A}{\psi}}{r}{x}}{\dsubst{\td{B}{\psi}}{r}{x}}}
\]
then
\[
\ceqtm[\Psi']
{\coe{x.\sigmacl{a}{\td{A}{\psi}}{\td{B}{\psi}}}{r}{r'}{M}}
{\coe{x.\sigmacl{a}{\td{A'}{\psi}}{\td{B'}{\psi}}}{r}{r'}{M'}}
{\sigmacl{a}{\dsubst{\td{A}{\psi}}{r'}{x}}{\dsubst{\td{B}{\psi}}{r'}{x}}}.
\]
By \cref{lem:expansion} on both sides and the introduction rule for dependent
pairs, it suffices to show
\[
\ceqtm[\Psi']
{\coe{x.\td{A}{\psi}}{r}{r'}{\fst{M}}}
{\coe{x.\td{A'}{\psi}}{r}{r'}{\fst{M'}}}
{\td{A}{\psi}}
\]
which holds by the fourth Kan condition of
$\ceqtype[\Psi',x]{\td{A}{\psi}}{\td{A'}{\psi}}$, and
\[
\ceqtm[\Psi']
{\coe{x.\subst{\td{B}{\psi}}{\coe{x.\td{A}{\psi}}{r}{x}{\fst{M}}}{a}}{r}{r'}{\snd{M}}}
{\coe{x.\subst{\td{B'}{\psi}}{\coe{x.\td{A'}{\psi}}{r}{x}{\fst{M'}}}{a}}{r}{r'}{\snd{M'}}}
{\subst{\td{B}{\psi}}{\coe{x.\td{A}{\psi}}{r}{r'}{\fst{M}}}{a}}
\]
which holds by the fourth Kan condition of
$\ceqtype[\Psi',x]
{\subst{\td{B}{\psi}}{\coe{x.\td{A}{\psi}}{r}{x}{\fst{M}}}{a}}
{\subst{\td{B'}{\psi}}{\coe{x.\td{A}{\psi}}{r}{x}{\fst{M'}}}{a}}$
(since $\ceqtm[\Psi',x]
{\coe{x.\td{A}{\psi}}{r}{x}{\fst{M}}}
{\coe{x.\td{A'}{\psi}}{r}{x}{\fst{M'}}}
{\td{A}{\psi}}$) and $x\fresh \td{B}{\psi}$.

The fifth Kan condition asserts that for any
$\msubsts{(\Psi',x)}{\psi}{\Psi}$, if
$\coftype[\Psi']{M}
{\sigmacl{a}{\dsubst{\td{A}{\psi}}{r}{x}}{\dsubst{\td{B}{\psi}}{r}{x}}}$,
then
$\ceqtm[\Psi']
{\coe{x.\picl{a}{\td{A}{\psi}}{\td{B}{\psi}}}{r}{r}{M}}{M}
{\sigmacl{a}{\dsubst{\td{A}{\psi}}{r}{x}}{\dsubst{\td{B}{\psi}}{r}{x}}}$.

By \cref{lem:expansion} on both sides and the introduction rule for dependent
pairs, it suffices to show
\[
\ceqtm[\Psi']{\fst{M}}
{\coe{x.\td{A}{\psi}}{r}{r}{\fst{M}}}
{\td{A}{\psi}}
\]
which holds by the fifth Kan condition of
$\ceqtype[\Psi',x]{\td{A}{\psi}}{\td{A'}{\psi}}$, and
\[
\ceqtm[\Psi']{\snd{M}}
{\coe{x.\subst{\td{B}{\psi}}{\coe{x.\td{A}{\psi}}{r}{x}{\fst{M}}}{a}}{r}{r}{\snd{M}}}
{\subst{\td{B}{\psi}}{\fst{M}}{a}}
\]
which holds by the fifth Kan condition of
$\cwftype[\Psi',x]
{\subst{\td{B}{\psi}}{\coe{x.\td{A}{\psi}}{r}{x}{\fst{M}}}{a}}$
and by the fact that
$\ceqtype[\Psi']
{\subst{\td{B}{\psi}}{\coe{x.\td{A}{\psi}}{r}{r}{\fst{M}}}{a}}
{\subst{\td{B}{\psi}}{\fst{M}}{a}}$.

\paragraph{Cubical}
For any $\psitd$ and
$\vinper[\Psi']{M}{M'}{\sigmacl{a}{\td{A}{\psi}}{\td{B}{\psi}}}$, we have
$\ceqtm[\Psi']{M}{M'}{\sigmacl{a}{\td{A}{\psi}}{\td{B}{\psi}}}$.

Then $M = \pair{O}{P}$, $M' = \pair{O'}{P'}$,
$\ceqtm[\Psi']{O}{O'}{\td{A}{\psi}}$ and
$\ceqtm[\Psi']{P}{P'}{\subst{\td{B}{\psi}}{O}{a}}$,
and the result follows from the introduction rule for dependent pairs.

\subsection{Identification types}

If a cubical type system has $\ceqtype[\Psi,x]{A}{A'}$,
$\ceqtm{P_0}{P_0'}{\dsubst{A}{0}{x}}$, and
$\ceqtm{P_1}{P_1'}{\dsubst{A}{1}{x}}$, we say it
\emph{has their identification type} when for all $\psitd$,
$\veqper[\Psi']{\Id{x.\td{A}{\psi}}{\td{P_0}{\psi}}{\td{P_1}{\psi}}}
{\Id{x.\td{A'}{\psi}}{\td{P_0'}{\psi}}{\td{P_1'}{\psi}}}$, and
$\vinper[\Psi']{-}{-}{\Id{x.\td{A}{\psi}}{\td{P_0}{\psi}}{\td{P_1}{\psi}}}$
is the least relation such that
\[
\vinper[\Psi']{\dlam{x}{M}}{\dlam{x}{M'}}
{\Id{x.\td{A}{\psi}}{\td{P_0}{\psi}}{\td{P_1}{\psi}}}
\]
when $\ceqtm[\Psi',x]{M}{M'}{\td{A}{\psi}}$ and
$\ceqtm[\Psi']{\dsubst{M}{\e}{x}}{\td{P_\e}{\psi}}
{\dsubst{\td{A}{\psi}}{\e}{x}}$ for all $\e$.

The PERs
$\vinper[\Psi']{-}{-}{\Id{x.\td{A}{\psi}}{\td{P_0}{\psi}}{\td{P_1}{\psi}}}$ and
$\vinper[\Psi']{-}{-}{\Id{x.\td{A'}{\psi}}{\td{P_0'}{\psi}}{\td{P_1'}{\psi}}}$
are equal because $\ceqtm[\Psi',x]{M}{M'}{\td{A}{\psi}}$ if and only if
$\ceqtm[\Psi',x]{M}{M'}{\td{A'}{\psi}}$, and
$\ceqtm[\Psi']{\dsubst{M}{\e}{x}}{\td{P_\e}{\psi}}
{\dsubst{\td{A}{\psi}}{\e}{x}}$
if and only if
$\ceqtm[\Psi']{\dsubst{M}{\e}{x}}{\td{P_\e'}{\psi}}
{\dsubst{\td{A'}{\psi}}{\e}{x}}$.

In the remainder of this subsection, we assume we are working with a cubical
type system that has $\ceqtype[\Psi,x]{A}{A'}$,
$\ceqtm{P_0}{P_0'}{\dsubst{A}{0}{x}}$,
$\ceqtm{P_1}{P_1'}{\dsubst{A}{1}{x}}$, and their
identification type.

\paragraph{Pretype}
$\ceqpretype{\Id{x.A}{P_0}{P_1}}{\Id{x.A'}{P_0'}{P_1'}}$.

For any $\msubsts{\Psi_1}{\psi_1}{\Psi}$ and $\msubsts{\Psi_2}{\psi_2}{\Psi_1}$,
\[
\veqper[\Psi_2]
{\Id{x.\td{A}{\psi_1\psi_2}}{\td{P_0}{\psi_1\psi_2}}{\td{P_1}{\psi_1\psi_2}}}
{\Id{x.\td{A'}{\psi_1\psi_2}}{\td{P_0'}{\psi_1\psi_2}}{\td{P_1'}{\psi_1\psi_2}}}
\]
because the cubical type system has this identification type.

\paragraph{Introduction}
If $\ceqtm[\Psi,x]{M}{M'}{A}$ and
$\ceqtm{\dsubst{M}{\e}{x}}{P_\e}{\dsubst{A}{\e}{x}}$ for all $\e$, then
$\ceqtm{\dlam{x}{M}}{\dlam{x}{M'}}{\Id{x.A}{P_0}{P_1}}$.

Each side has coherent aspects up to syntactic equality, since
$\wfval[\Psi']{\dlam{x}{\td{M}{\psi}}}$ for all $\psitd$. Thus it suffices to
show
$\vinper[\Psi_2]
{\dlam{x}{\td{M}{\psi_1\psi_2}}}
{\dlam{x}{\td{M'}{\psi_1\psi_2}}}
{\Id{x.\td{A}{\psi_1\psi_2}}{\td{P_0}{\psi_1\psi_2}}{\td{P_1}{\psi_1\psi_2}}}$,
which is true because
$\ceqtm[\Psi_2,x]{\td{M}{\psi_1\psi_2}}{\td{M'}{\psi_1\psi_2}}
{\td{A}{\psi_1\psi_2}}$ and
$\ceqtm[\Psi_2]{\dsubst{\td{M}{\psi_1\psi_2}}{\e}{x}}{\td{P_\e}{\psi_1\psi_2}}
{\dsubst{\td{A}{\psi_1\psi_2}}{\e}{x}}$ for all $\e$.

\paragraph{Elimination}
If $\ceqtm{M}{M'}{\Id{x.A}{P_0}{P_1}}$ then
$\ceqtm{\dapp{M}{r}}{\dapp{M'}{r}}{\dsubst{A}{r}{x}}$ and
$\ceqtm{\dapp{M}{\e}}{P_\e}{\dsubst{A}{\e}{x}}$.

For any $\msubsts{\Psi_1}{\psi_1}{\Psi}$ and $\msubsts{\Psi_2}{\psi_2}{\Psi_1}$,
we know $\td{M}{\psi_1}\evals \dlam{x}{N_1}$ where
$\coftype[\Psi_1,x]{N_1}{\td{A}{\psi_1}}$.
Thus $\dapp{(\td{M}{\psi_1})}{\td{r}{\psi_1}} \steps^*
\dapp{(\dlam{x}{N_1})}{\td{r}{\psi_1}} \steps \dsubst{N_1}{\td{r}{\psi_1}}{x}$
where $\dsubst{N_1}{\td{r}{\psi_1}}{x} \evals X_1$ and
$\inper[\Psi_2]{\td{X_1}{\psi_2}}{\td{\dsubst{N_1}{\td{r}{\psi_1}}{x}}{\psi_2}}
{\td{\dsubst{A}{r}{x}}{\psi_1\psi_2}}$.
We also know
$\vinper[\Psi_2]
{\td{M}{\psi_1\psi_2}\evals \dlam{x}{N_{12}}}
{\dlam{x}{\td{N_1}{\psi_2}}}
{\Id{x.\td{A}{\psi_1\psi_2}}{\td{P_0}{\psi_1\psi_2}}{\td{P_1}{\psi_1\psi_2}}}$,
so $\ceqtm[\Psi_2,x]{N_{12}}{\td{N_1}{\psi_2}}{\td{A}{\psi_1\psi_2}}$.
Thus $\dapp{(\td{M}{\psi_1\psi_2})}{\td{r}{\psi_1\psi_2}}\steps^*
\dapp{(\dlam{x}{N_{12}})}{\td{r}{\psi_1\psi_2}} \steps
\dsubst{N_{12}}{\td{r}{\psi_1\psi_2}}{x}$.
We want to show
$\inper[\Psi_2]{\td{X_1}{\psi_2}}{\dsubst{N_{12}}{\td{r}{\psi_1\psi_2}}{x}}
{\td{\dsubst{A}{r}{x}}{\psi_1\psi_2}}$,
which follows by
$\inper[\Psi_2]
{\td{\dsubst{N_1}{\td{r}{\psi_1}}{x}}{\psi_2}}
{\dsubst{N_{12}}{\td{r}{\psi_1\psi_2}}{x}}
{\td{\dsubst{A}{r}{x}}{\psi_1\psi_2}}$ and transitivity.

An analogous argument shows $\dapp{M'}{r}$ also has coherent aspects. To see
that the aspects of $\dapp{M}{r}$ and $\dapp{M'}{r}$ are related to each other,
we use the fact that $\td{M'}{\psi_1\psi_2} \evals \dlam{x}{N_{12}'}$ and
$\ceqtm[\Psi_2,x]{N_{12}}{N_{12}'}{\td{A}{\psi_1\psi_2}}$, so
$\dapp{(\td{M'}{\psi_1\psi_2})}{\td{r}{\psi_1\psi_2}} \steps^*
\dsubst{N_{12}'}{\td{r}{\psi_1\psi_2}}{x}$ and
$\inper[\Psi_2]
{\dsubst{N_{12}}{\td{r}{\psi_1\psi_2}}{x}}
{\dsubst{N_{12}'}{\td{r}{\psi_1\psi_2}}{x}}
{\td{\dsubst{A}{r}{x}}{\psi_1\psi_2}}$.

For the second elimination rule, by \cref{lem:coftype-ceqtm} it suffices to
show that for any $\psitd$,
$\inper[\Psi']
{\dapp{(\td{M}{\psi})}{\e}}
{\td{P_\e}{\psi}}
{\td{\dsubst{A}{\e}{x}}{\psi}}$.
But $\td{M}{\psi}\evals \dlam{x}{N}$ where
$\coftype[\Psi',x]{N}{\td{A}{\psi}}$ and
$\ceqtm[\Psi']{\dsubst{N}{\e}{x}}{P_\e}{\dsubst{\td{A}{\psi}}{\e}{x}}$.
Thus $\dapp{(\td{M}{\psi})}{\e} \steps^*
\dapp{(\dlam{x}{N})}{\e} \steps \dsubst{N}{\e}{x}$ and
$\inper[\Psi']
{\dsubst{N}{\e}{x}}
{\td{P_\e}{\psi}}
{\td{\dsubst{A}{\e}{x}}{\psi}}$.

\paragraph{Computation}
If $\coftype[\Psi,x]{M}{A}$ then
$\ceqtm{\dapp{(\dlam{x}{M})}{r}}{\dsubst{M}{r}{x}}{\dsubst{A}{r}{x}}$.

Follows from \cref{lem:expansion}.

\paragraph{Eta}
If $\coftype{M}{\Id{x.A}{P_0}{P_1}}$ then
$\ceqtm{M}{\dlam{x}{(\dapp{M}{x})}}{\Id{x.A}{P_0}{P_1}}$.

By the elimination and introduction rules for identification types,
$\coftype{\dlam{x}{(\dapp{M}{x})}}{\Id{x.A}{P_0}{P_1}}$.
By \cref{lem:coftype-ceqtm} it suffices to show that for any $\psitd$,
$\inper[\Psi']{\td{M}{\psi}}{\dlam{x}{(\dapp{\td{M}{\psi}}{x})}}
{\Id{x.\td{A}{\psi}}{\td{P_0}{\psi}}{\td{P_1}{\psi}}}$.
We know $\td{M}{\psi}\evals\dlam{x}{N}$ such that
$\coftype[\Psi',x]{N}{\td{A}{\psi}}$ and
$\ceqtm[\Psi']{\dsubst{N}{\e}{x}}{\td{P_\e}{\psi}}
{\dsubst{\td{A}{\psi}}{\e}{x}}$; we must show
$\ceqtm[\Psi',x]{N}{\dapp{\td{M}{\psi}}{x}}{\td{A}{\psi}}$.
Again by \cref{lem:coftype-ceqtm}, we show that for any
$\msubsts{\Psi''}{\psi'}{(\Psi',x)}$,
$\inper[\Psi'']{\td{N}{\psi'}}{\dapp{\td{M}{\psi\psi'}}{\td{x}{\psi'}}}
{\td{A}{\psi\psi'}}$. But
$\vinper[\Psi'']{\td{M}{\psi\psi'} \evals \dlam{x}{N'}}{\dlam{x}{\td{N}{\psi'}}}
{\Id{x.\td{A}{\psi\psi'}}{\td{P_0}{\psi\psi'}}{\td{P_1}{\psi\psi'}}}$, so
$\ceqtm[\Psi'',x]{N'}{\td{N}{\psi'}}{\td{A}{\psi\psi'}}$ and
$\dapp{\td{M}{\psi\psi'}}{\td{x}{\psi'}} \steps^*
\dapp{(\dlam{x}{N'})}{\td{x}{\psi'}} \steps \dsubst{N'}{\td{x}{\psi'}}{x}$.
The result follows by
$\ceqtm[\Psi'',x]{\dsubst{N'}{\td{x}{\psi'}}{x}}{\td{N}{\psi'}}
{\td{A}{\psi\psi'}}$.

\paragraph{Kan}
$\ceqpretype{\Id{x.A}{P_0}{P_1}}{\Id{x.A'}{P_0'}{P_1'}}$ are equally Kan.

For the first Kan condition, we must show that for any $\psitd$, if
\begin{enumerate}
\item $\ceqtm[\Psi']{M}{O}
{\Id{x.\td{A}{\psi}}{\td{P_0}{\psi}}{\td{P_1}{\psi}}}$,
\item $\ceqtmres[\Psi',y]{r_i=\e,r_j=\e'}{N^\e_i}{N^{\e'}_j}
{\Id{x.\td{A}{\psi}}{\td{P_0}{\psi}}{\td{P_1}{\psi}}}$
for any $i,j,\e,\e'$,
\item $\ceqtmres[\Psi',y]{r_i=\e}{N^\e_i}{P^\e_i}
{\Id{x.\td{A}{\psi}}{\td{P_0}{\psi}}{\td{P_1}{\psi}}}$
for any $i,\e$, and
\item $\ceqtmres[\Psi']{r_i=\e}{\dsubst{N^\e_i}{r}{y}}{M}
{\Id{x.\td{A}{\psi}}{\td{P_0}{\psi}}{\td{P_1}{\psi}}}$
for any $i,\e$, then
\end{enumerate}
$\ceqtm[\Psi']{\hcomgeneric{\etc{r_i}}
{\Id{x.\td{A}{\psi}}{\td{P_0}{\psi}}{\td{P_1}{\psi}}}}%
{\hcom{\etc{r_i}}{\Id{x.\td{A'}{\psi}}{\td{P_0'}{\psi}}{\td{P_1'}{\psi}}}
{r}{r'}{O}{\etc{y.P^\e_i}}}
{\Id{x.\td{A}{\psi}}{\td{P_0}{\psi}}{\td{P_1}{\psi}}}$.

By \cref{lem:expansion} on both sides and the introduction rule for
identification types, it suffices to show
\[\begin{aligned}
&{\hcom{\etc{r_i},x}{\td{A}{\psi}}{r}{r'}
  {\dapp{M}{x}}
  {\etc{y.\dapp{N^\e_i}{x}},\_.\td{P_0}{\psi},\_.\td{P_1}{\psi}}} \\
\ceqtmtab[\Psi',x]{}
{\hcom{\etc{r_i},x}{\td{A'}{\psi}}{r}{r'}
  {\dapp{O}{x}}
  {\etc{y.\dapp{P^\e_i}{x}},\_.\td{P_0'}{\psi},\_.\td{P_1'}{\psi}}}
{\td{A}{\psi}}
\end{aligned}\]
and $\ceqtm[\Psi']{\dsubst{\hcomsym}{\e}{x}}{\td{P_\e}{\psi}}
{\dsubst{\td{A}{\psi}}{\e}{x}}$.
These follow from the first and third Kan conditions of
$\ceqtype[\Psi']{\td{A}{\psi}}{\td{A'}{\psi}}$; most of the Kan conditions'
hypotheses follow from the first elimination rule for identifications, but
because we add a new pair of tubes, there are four additional hypotheses:
\begin{enumerate}
\item $\ceqtmres[\Psi',x,y]{x=\e}{\td{P_\e}{\psi}}{\td{P_\e'}{\psi}}
{\td{A}{\psi}}$ for any $\e$.

Since $x\fresh \td{P_\e}{\psi}$, this follows by
$\ceqtm[\Psi',y]{\td{P_\e}{\psi}}{\td{P_\e'}{\psi}}
{\dsubst{\td{A}{\psi}}{\e}{x}}$.

\item $\ceqtmres[\Psi',x,y]{r_i=\e,x=\e'}{\dapp{N^\e_i}{x}}{\td{P_{\e'}}{\psi}}
{\td{A}{\psi}}$ for any $i,\e,\e'$.

It suffices to show that for any $\msubsts{\Psi''}{\psi'}{(\Psi',x,y)}$ such
that $\td{x}{\psi'}=\e'$ and $\td{r_i}{\psi'}=\e$,
$\ceqtm[\Psi'']{\dapp{(\td{N^\e_i}{\psi'})}{\e'}}{\td{P_{\e'}}{\psi\psi'}}
{\td{A}{\psi\psi'}}$, which follows from the second elimination rule for
identifications.

\item $\ceqtmres[\Psi',x,y]{x=\e,x=\e'}{\td{P_{\e}}{\psi}}{\td{P_{\e'}}{\psi}}
{\td{A}{\psi}}$ for any $\e,\e'$.

If $\e=\e'$ this follows by assumption; otherwise it holds vacuously.

\item $\ceqtmres[\Psi',x]{x=\e}{\dapp{M}{x}}
{\dsubst{\td{P_\e}{\psi}}{r'}{y}}{\td{A}{\psi}}$ for any $\e$.

It suffices to show that for any $\msubsts{\Psi''}{\psi'}{(\Psi',x)}$ such that
$\td{x}{\psi'}=\e$,
$\ceqtm[\Psi'']{\dapp{M}{\e}}{\td{P_\e}{\psi}}{\td{A}{\psi}}$ (since $x,y\fresh
P_\e$), which holds by the second elimination rule for identifications.
\end{enumerate}

For the second Kan condition, we must show that for any $\psitd$, if
\begin{enumerate}
\item $\coftype[\Psi']{M}
{\Id{x.\td{A}{\psi}}{\td{P_0}{\psi}}{\td{P_1}{\psi}}}$,
\item $\ceqtmres[\Psi',y]{r_i=\e,r_j=\e'}{N^\e_i}{N^{\e'}_j}
{\Id{x.\td{A}{\psi}}{\td{P_0}{\psi}}{\td{P_1}{\psi}}}$
for any $i,j,\e,\e'$, and
\item $\ceqtmres[\Psi']{r_i=\e}{\dsubst{N^\e_i}{r}{y}}{M}
{\Id{x.\td{A}{\psi}}{\td{P_0}{\psi}}{\td{P_1}{\psi}}}$
for any $i,\e$, then
\end{enumerate}
$\ceqtm[\Psi']
{\hcom{\etc{r_i}}{\Id{x.\td{A}{\psi}}{\td{P_0}{\psi}}{\td{P_1}{\psi}}}{r}{r}{M}{\etc{y.N^\e_i}}
}{M}
{\Id{x.\td{A}{\psi}}{\td{P_0}{\psi}}{\td{P_1}{\psi}}}$.

By \cref{lem:expansion} and the eta and introduction rules for identifications,
it suffices to show
\[
\ceqtm[\Psi',x]
{{\hcom{\etc{r_i},x}{\td{A}{\psi}}{r}{r}
  {\dapp{M}{x}}
  {\etc{y.\dapp{N^\e_i}{x}},\_.\td{P_0}{\psi},\_.\td{P_1}{\psi}}}}
{\dapp{M}{x}}
{\td{A}{\psi}}
\]
which holds by the second Kan condition of $\cwftype[\Psi']{\td{A}{\psi}}$,
deriving its hypotheses as before.

For the third Kan condition, we must show that for any $\psitd$, if $r_i=\e$,
\begin{enumerate}
\item $\coftype[\Psi']{M}
{\Id{x.\td{A}{\psi}}{\td{P_0}{\psi}}{\td{P_1}{\psi}}}$,
\item $\ceqtmres[\Psi',y]{r_i=\e,r_j=\e'}{N^\e_i}{N^{\e'}_j}
{\Id{x.\td{A}{\psi}}{\td{P_0}{\psi}}{\td{P_1}{\psi}}}$
for any $i,j,\e,\e'$, and
\item $\ceqtmres[\Psi']{r_i=\e}{\dsubst{N^\e_i}{r}{y}}{M}
{\Id{x.\td{A}{\psi}}{\td{P_0}{\psi}}{\td{P_1}{\psi}}}$
for any $i,\e$, then
\end{enumerate}
$\ceqtm[\Psi']
{\hcomgeneric{\etc{r_i}}{\Id{x.\td{A}{\psi}}{\td{P_0}{\psi}}{\td{P_1}{\psi}}}
}{\dsubst{N^\e_i}{r'}{y}}
{\Id{x.\td{A}{\psi}}{\td{P_0}{\psi}}{\td{P_1}{\psi}}}$.

By \cref{lem:expansion} and the eta and introduction rules for identifications,
it suffices to show
\[
\ceqtm[\Psi',x]
{{\hcom{\etc{r_i},x}{\td{A}{\psi}}{r}{r'}
  {\dapp{M}{x}}
  {\etc{y.\dapp{N^\e_i}{x}},\_.\td{P_0}{\psi},\_.\td{P_1}{\psi}}}}
{\dapp{\dsubst{N^\e_i}{r'}{y}}{x}}
{\td{A}{\psi}}
\]
which follows from the third Kan condition of $\cwftype[\Psi']{\td{A}{\psi}}$,
deriving its hypotheses as before.

For the fourth Kan condition, we must show that for any
$\msubsts{(\Psi',y)}{\psi}{\Psi}$, if
\[
\ceqtm[\Psi']{M}{M'}
{\Id{x.\dsubst{\td{A}{\psi}}{r}{y}}{\dsubst{\td{P_0}{\psi}}{r}{y}}{\dsubst{\td{P_1}{\psi}}{r}{y}}}
\]
then
\[
\ceqtm[\Psi']
{\coe{y.\Id{x.\td{A}{\psi}}{\td{P_0}{\psi}}{\td{P_1}{\psi}}}{r}{r'}{M}}
{\coe{y.\Id{x.\td{A'}{\psi}}{\td{P_0'}{\psi}}{\td{P_1'}{\psi}}}{r}{r'}{M'}}
{\Id{x.\dsubst{\td{A}{\psi}}{r'}{y}}{\dsubst{\td{P_0}{\psi}}{r'}{y}}{\dsubst{\td{P_1}{\psi}}{r'}{y}}}.
\]

By \cref{lem:expansion} on both sides and the introduction rule for
identification types, it suffices to show
\[
\ceqtm[\Psi',x]
{\com{x}{y.\td{A}{\psi}}{r}{r'}{\dapp{M}{x}}{y.\td{P_0}{\psi},y.\td{P_1}{\psi}}}
{\com{x}{y.\td{A'}{\psi}}{r}{r'}{\dapp{M'}{x}}{y.\td{P_0'}{\psi},y.\td{P_1'}{\psi}}}
{\dsubst{\td{A}{\psi}}{r'}{y}}
\]
and for all $\e$,
\[
\ceqtm[\Psi']
{\com{\e}{y.\td{A}{\psi}}{r}{r'}{\dapp{M}{\e}}{y.\td{P_0}{\psi},y.\td{P_1}{\psi}}}
{\dsubst{\td{P_\e}{\psi}}{r'}{y}}
{\dsubst{\dsubst{\td{A}{\psi}}{r'}{y}}{\e}{x}}.
\]
These both follow from \cref{thm:com} once we have shown:
\begin{enumerate}
\item $\ceqtm[\Psi',x]{\dapp{M}{x}}{\dapp{M'}{x}}{\dsubst{\td{A}{\psi}}{r}{y}}$.

Follows from the first elimination rule for identifications.

\item $\ceqtmres[\Psi',x,y]{x=\e,x=\e'}{\td{P_\e}{\psi}}{\td{P_{\e'}}{\psi}}
{\td{A}{\psi}}$ for any $\e,\e'$.

Follows from $\coftype[\Psi',y]{\td{P_\e}{\psi}}{\dsubst{\td{A}{\psi}}{\e}{x}}$
if $\e=\e'$ and vacuously if $\e\neq\e'$.

\item $\ceqtmres[\Psi',x,y]{x=\e}{\td{P_\e}{\psi}}{\td{P_\e'}{\psi}}
{\td{A}{\psi}}$ for any $\e$.

Follows from $\ceqtm[\Psi',y]{\td{P_\e}{\psi}}{\td{P_\e'}{\psi}}
{\dsubst{\td{A}{\psi}}{\e}{x}}$.

\item $\ceqtmres[\Psi',x]{x=\e}{\dsubst{\td{P_\e}{\psi}}{r}{y}}{\dapp{M}{x}}
{\dsubst{\td{A}{\psi}}{r}{y}}$ for any $\e$.

Follows from the second elimination rule for identifications.
\end{enumerate}

For the fifth Kan condition, we must show that for any
$\msubsts{(\Psi',y)}{\psi}{\Psi}$, if
\[
\coftype[\Psi']{M}
{\Id{x.\dsubst{\td{A}{\psi}}{r}{y}}{\dsubst{\td{P_0}{\psi}}{r}{y}}{\dsubst{\td{P_1}{\psi}}{r}{y}}}
\]
then
\[
\ceqtm[\Psi']
{\coe{y.\Id{x.\td{A}{\psi}}{\td{P_0}{\psi}}{\td{P_1}{\psi}}}{r}{r}{M}}
{M}
{\Id{x.\dsubst{\td{A}{\psi}}{r}{y}}{\dsubst{\td{P_0}{\psi}}{r}{y}}{\dsubst{\td{P_1}{\psi}}{r}{y}}}.
\]

By \cref{lem:expansion} and the eta rule for identification types, it suffices
to show
\[
\ceqtm[\Psi',x]
{\com{x}{y.\td{A}{\psi}}{r}{r}{\dapp{M}{x}}{y.\td{P_0}{\psi},y.\td{P_1}{\psi}}}
{\dapp{M}{x}}
{\dsubst{\td{A}{\psi}}{r}{y}}
\]
which follows from \cref{thm:com}, establishing its hypotheses as before.

\paragraph{Cubical}
For any $\psitd$ and $\vinper[\Psi']{M}{M'}
{\Id{x.\td{A}{\psi}}{\td{P_0}{\psi}}{\td{P_1}{\psi}}}$, we have
$\ceqtm[\Psi']{M}{M'}{\Id{x.\td{A}{\psi}}{\td{P_0}{\psi}}{\td{P_1}{\psi}}}$.

Then $M = \dlam{x}{N}$, $M' = \dlam{x}{N'}$,
$\ceqtm[\Psi',x]{N}{N'}{\td{A}{\psi}}$,
$\ceqtm[\Psi']{\dsubst{N}{\e}{x}}{\td{P_\e}{\psi}}
{\dsubst{\td{A}{\psi}}{\e}{x}}$ for all $\e$, and the result follows from the
introduction rule for identification types.

\subsection{Not}

If a cubical type system has booleans, we say it \emph{has the $\notb{x}$ type}
when $\veqper[\Psi,x]{\notb{x}}{\notb{x}}$ for all $\Psi$,
and $\vinper[\Psi,x]{-}{-}{\notb{x}}$ is the least relation such that:
\[ \vinper[\Psi,x]{\notel{x}{M}}{\notel{x}{M'}}{\notb{x}} \]
when $\ceqtm[\Psi,x]{M}{M'}{\bool}$.

This type is somewhat unusual because it exists primarily to be coerced along
($\coe{x.\notb{x}}{r}{r'}{M}$), rather than to be introduced or eliminated in
the manner of dependent function, pair, and identification types. Accordingly,
the bulk of this section is dedicated to proving that
$\cpretype[\Psi,x]{\notb{x}}$ is Kan.

In the remainder of this section, we assume we are working with a cubical type
system that has booleans and the $\notb{x}$ type. We will need the following
lemmas, which hold in any cubical type system with booleans:

\begin{lemma}
If $\coftype{M}{\bool}$ then $\ceqtm{\notf{\notf{M}}}{M}{\bool}$.
\end{lemma}
\begin{proof}
Recalling that $\notf{M}$ is notation for $\ifb{\_.\bool}{M}{\false}{\true}$,
we conclude from the introduction and elimination rules for booleans that
$\coftype{\notf{\notf{M}}}{\bool}$. By \cref{lem:coftype-ceqtm}, it suffices to
show that for any $\psitd$,
$\inper[\Psi']{\notf{\notf{\td{M}{\psi}}}}{\td{M}{\psi}}{\bool}$. We case on
$\inper[\Psi']{\td{M}{\psi}}{\td{M}{\psi}}{\bool}$:
\begin{enumerate}
\item $\td{M}{\psi} \evals \true$.

Then $\notf{\notf{\td{M}{\psi}}} \steps^*
\notf{\notf{\true}} \steps \notf{\false} \steps \true$, and
$\vinper[\Psi']{\true}{\true}{\bool}$.

\item $\td{M}{\psi} \evals \false$.

Then $\notf{\notf{\td{M}{\psi}}} \steps^*
\notf{\notf{\false}} \steps \notf{\true} \steps \false$, and
$\vinper[\Psi']{\false}{\false}{\bool}$.

\item $\td{M}{\psi} \evals \hcom{\etc{x_i}}{\bool}{r}{r'}{O}{\etc{y.N^\e_i}}$
where $r\neq r'$,
$\coftype[\Psi']{O}{\bool}$,
$\ceqtmres[\Psi',y]{x_i=\e,x_j=\e'}{N^\e_i}{N^{\e'}_j}{\bool}$ for all
$i,j,\e,\e'$, and
$\ceqtmres[\Psi']{x_i=\e}{\dsubst{N^\e_i}{r}{y}}{O}{\bool}$ for all $i,\e$.

Then, expanding the definition of $\comsym$,
\[\begin{aligned}
\notf{\notf{\td{M}{\psi}}} &\steps^*
\notf{\hcom{\etc{x_i}}{\bool}{r}{r'}
  {\coe{\_.\bool}{r}{r'}{\notf{O}}}
  {\etc{y.\coe{\_.\bool}{y}{r'}{\notf{N^\e_i}}}}}
\\ &\steps\hphantom{{}^*}
\hcom{\etc{x_i}}{\bool}{r}{r'}
  {\coe{\_.\bool}{r}{r'}{\notf{\coe{\_.\bool}{r}{r'}{\notf{O}}}}}
  {\dots}
\end{aligned}\]
which, by the elimination rule and first and fourth Kan conditions of $\bool$,
is
\[
\ceqtm[\Psi']{-}
{\hcom{\etc{x_i}}{\bool}{r}{r'}{\notf{\notf{O}}}{\notf{\notf{N^\e_i}}}}{\bool}.
\]
The result follows by the inductive hypothesis and the first Kan condition of
$\bool$.
\qedhere
\end{enumerate}
\end{proof}

\begin{lemma}\label{lem:inper-not-swap}
If $\coftype{M}{\bool}$, $\coftype{N}{\bool}$, and for all $\psitd$,
$\inper[\Psi']{\notf{\td{M}{\psi}}}{\td{N}{\psi}}{\bool}$, then
$\ceqtm{M}{\notf{N}}{\bool}$ (and in particular, $\inper{M}{\notf{N}}{\bool}$).
\end{lemma}
\begin{proof}
By \cref{lem:coftype-ceqtm}, $\ceqtm{\notf{M}}{N}{\bool}$. Then
$\ceqtm{\notf{\notf{M}}}{\notf{N}}{\bool}$, so $\ceqtm{M}{\notf{N}}{\bool}$.
\end{proof}

\paragraph{Pretype}
$\cpretype[\Psi,x]{\notb{x}}$ and $\ceqpretype{\notb{\e}}{\bool}$.

For the first part, there are three cases to consider.
For any $\msubsts{\Psi_1}{\psi_1}{\Psi}$ and $\msubsts{\Psi_2}{\psi_2}{\Psi_1}$,
\begin{enumerate}
\item If $\td{x}{\psi_1} = \e$ then
$\notb{\td{x}{\psi_1}}\evals\bool$,
$\td{\bool}{\psi_2}\evals\bool$, and
$\notb{\td{x}{\psi_1\psi_2}}\evals\bool$;

\item If $\td{x}{\psi_1} = x'$ and $\td{x'}{\psi_2} = \e$ then
$\notb{\td{x}{\psi_1}}\evals\notb{x'}$,
$\notb{\td{x'}{\psi_2}}\evals\bool$, and
$\notb{\td{x}{\psi_1\psi_2}}\evals\bool$; and

\item If $\td{x}{\psi_1} = x'$ and $\td{x'}{\psi_2} = x''$ then
$\notb{\td{x}{\psi_1}}\evals\notb{x'}$,
$\notb{\td{x'}{\psi_2}}\evals\notb{x''}$, and
$\notb{\td{x}{\psi_1\psi_2}}\evals\notb{x''}$.
\end{enumerate}
But $\veqper[\Psi_2]{\bool}{\bool}$ and
$\veqper[\Psi_2',x'']{\notb{x''}}{\notb{x''}}$ where $\Psi_2 = (\Psi_2',x'')$.

For the second part, $\notb{\td{\e}{\psi}} \evals \bool$ and
$\veqper[\Psi']{\bool}{\bool}$ for any $\psitd$.

\paragraph{Introduction}
If $\ceqtm{M}{N}{\bool}$ then $\ceqtm{\notel{r}{M}}{\notel{r}{N}}{\notb{r}}$.

Let $\msubsts{\Psi_1}{\psi_1}{\Psi}$ and $\msubsts{\Psi_2}{\psi_2}{\Psi_1}$.
\begin{enumerate}
\item If $\td{r}{\psi_1} = 0$ then
$\notel{0}{\td{M}{\psi_1}} \steps \notf{\td{M}{\psi_1}}$ and
$\notel{0}{\td{M}{\psi_1\psi_2}} \steps \notf{\td{M}{\psi_1\psi_2}}$. By
$\ceqtm[\Psi_1]{\notf{\td{M}{\psi_1}}}{\notf{\td{N}{\psi_1}}}{\bool}$ we know
$\notf{\td{M}{\psi_1}} \evals M_1'$, $\notf{\td{N}{\psi_1}} \evals N_1'$, and
$\inperfour[\Psi_2]
{\td{M_1'}{\psi_2}}{\td{N_1'}{\psi_2}}
{\notf{\td{M}{\psi_1\psi_2}}}{\notf{\td{N}{\psi_1\psi_2}}}{\bool}$, which is
what we wanted to show.

\item If $\td{r}{\psi_1} = 1$ then
$\notel{1}{\td{M}{\psi_1}} \steps \td{M}{\psi_1}$ and
$\notel{1}{\td{M}{\psi_1\psi_2}} \steps \td{M}{\psi_1\psi_2}$. Our assumption
directly implies $\td{M}{\psi_1} \evals M_1$, $\td{N}{\psi_1} \evals N_1$, and
$\inperfour[\Psi_2]
{\td{M_1}{\psi_2}}{\td{N_1}{\psi_2}}
{\td{M}{\psi_1\psi_2}}{\td{N}{\psi_1\psi_2}}{\bool}$.

\item If $\td{r}{\psi_1} = x$ and $\td{x}{\psi_2} = 0$ then
$\isval{\notel{x}{\td{M}{\psi_1}}}$ and
$\notel{0}{\td{M}{\psi_1\psi_2}} \steps \notf{\td{M}{\psi_1\psi_2}}$.
Then by
$\ceqtm[\Psi_2]{\notf{\td{M}{\psi_1\psi_2}}}{\notf{\td{N}{\psi_1\psi_2}}}{\bool}$
we conclude
$\inperfour[\Psi_2]
{\notf{\td{M}{\psi_1\psi_2}}}{\notf{\td{N}{\psi_1\psi_2}}}
{\notf{\td{M}{\psi_1\psi_2}}}{\notf{\td{N}{\psi_1\psi_2}}}{\bool}$.

\item If $\td{r}{\psi_1} = x$ and $\td{x}{\psi_2} = 1$ then
$\isval{\notel{x}{\td{M}{\psi_1}}}$ and
$\notel{1}{\td{M}{\psi_1\psi_2}} \steps \td{M}{\psi_1\psi_2}$.
By our assumption,
$\inperfour[\Psi_2]
{\td{M}{\psi_1\psi_2}}{\td{N}{\psi_1\psi_2}}
{\td{M}{\psi_1\psi_2}}{\td{N}{\psi_1\psi_2}}{\bool}$.

\item If $\td{r}{\psi_1} = x$ and $\td{x}{\psi_2} = x'$ then
$\isval{\notel{x}{\td{M}{\psi_1}}}$ and
$\isval{\notel{x'}{\td{M}{\psi_1\psi_2}}}$.
By
$\ceqtm[\Psi_2]{\td{M}{\psi_1\psi_2}}{\td{N}{\psi_1\psi_2}}{\bool}$,
$\vinperfour[\Psi_2]
{\notel{x'}{\td{M}{\psi_1\psi_2}}}{\notel{x'}{\td{N}{\psi_1\psi_2}}}
{\notel{x'}{\td{M}{\psi_1\psi_2}}}{\notel{x'}{\td{N}{\psi_1\psi_2}}}
{\notb{x'}}$.
\end{enumerate}

\paragraph{Computation}
If $\coftype{M}{\bool}$, then
$\ceqtm{\coe{x.\notb{x}}{\e}{\e}{M}}{M}{\bool}$ and
$\ceqtm{\coe{x.\notb{x}}{\e}{\eb}{M}}{\notf{M}}{\bool}$.

Follows from \cref{lem:expansion}.

\paragraph{Kan}
$\ceqpretype{\notb{\e}}{\bool}$ are equally Kan, and
$\cpretype[\Psi,x]{\notb{x}}$ is Kan.

The first follows directly from \cref{lem:expansion} and the fact that
$\cpretype{\bool}$ is Kan, because $\hcomsym$ and $\coesym$ first evaluate their
type argument and $\notb{\e} \steps \bool$.

For the second,
the operational semantics for $\hcomsym$ at $\notb{x}$ involve $\coesym$, so we
start by proving the fourth Kan condition, which asserts that for
any $\msubsts{(\Psi',x')}{\psi}{(\Psi,x)}$, if
$\ceqtm[\Psi']{M}{N}{\notb{\dsubst{\td{x}{\psi}}{r}{x'}}}$, then
$\ceqtm[\Psi']
{\coe{x'.\notb{\td{x}{\psi}}}{r}{r'}{M}}
{\coe{x'.\notb{\td{x}{\psi}}}{r}{r'}{N}}
{\notb{\dsubst{\td{x}{\psi}}{r'}{x'}}}$.

If $\td{x}{\psi} = \e$ then by \cref{lem:expansion,lem:ceqpretype-ceqtm} it
suffices to show
$\ceqtm[\Psi']{\coe{x'.\bool}{r}{r'}{M}}{\coe{x'.\bool}{r}{r'}{N}}{\bool}$,
which is the fourth Kan condition of $\bool$.
If $\td{x}{\psi} = y \neq x'$ then by \cref{lem:expansion} it suffices to show
$\ceqtm[\Psi']{M}{N}{\notb{y}}$ when $\ceqtm[\Psi']{M}{N}{\notb{y}}$.
Otherwise, $\td{x}{\psi} = x'$, and we must show that if
$\ceqtm[\Psi']{M}{N}{\notb{r}}$ then
$\ceqtm[\Psi']
{\coe{x'.\notb{x'}}{r}{r'}{M}}
{\coe{x'.\notb{x'}}{r}{r'}{N}}
{\notb{r'}}$.
Establishing this requires a large case split; we focus on the unary version
because the binary one follows easily. Let
$\msubsts{\Psi_1}{\psi_1}{\Psi'}$ and $\msubsts{\Psi_2}{\psi_2}{\Psi_1}$.

\begin{enumerate}
\item
If $\td{r}{\psi_1} = \e$ and $\td{r'}{\psi_1} = \e$
then $\coe{x.\notb{x}}{\e}{\e}{\td{M}{\psi_1}} \steps
\td{M}{\psi_1} \evals M_1$,
and $\coe{x.\notb{x}}{\e}{\e}{\td{M}{\psi_1\psi_2}} \steps
\td{M}{\psi_1\psi_2}$ where
$\inper[\Psi_2]{\td{M_1}{\psi_2}}{\td{M}{\psi_1\psi_2}}{\bool}$.

\item
If $\td{r}{\psi_1} = \e$ and $\td{r'}{\psi_1} = \eb$ then
$\coe{x.\notb{x}}{\e}{\eb}{\td{M}{\psi_1}} \steps
\notf{\td{M}{\psi_1}} \evals X_1$,
and $\coe{x.\notb{x}}{\e}{\eb}{\td{M}{\psi_1\psi_2}} \steps
\notf{\td{M}{\psi_1\psi_2}}$.
By \cref{lem:ceqpretype-ceqtm} we know $\coftype[\Psi_1]{\td{M}{\psi_1}}{\bool}$
so $\coftype[\Psi_1]{\notf{\td{M}{\psi_1}}}{\bool}$ and therefore
$\inper[\Psi_2]{\td{X_1}{\psi_2}}{\notf{\td{M}{\psi_1\psi_2}}}{\bool}$.

\item
If $\td{r}{\psi_1} = 1$ and $\td{r'}{\psi_1} = x$ then
$\coe{x.\notb{x}}{1}{x}{\td{M}{\psi_1}} \steps \notel{x}{\td{M}{\psi_1}}$.

\begin{enumerate}
\item
If $\td{x}{\psi_2} = 1$ then
$\notel{1}{\td{M}{\psi_1\psi_2}} \steps \td{M}{\psi_1\psi_2}$ and
$\coe{x.\notb{x}}{1}{1}{\td{M}{\psi_1\psi_2}} \steps \td{M}{\psi_1\psi_2}$,
where $\inper[\Psi_2]{\td{M}{\psi_1\psi_2}}{\td{M}{\psi_1\psi_2}}{\bool}$.

\item
If $\td{x}{\psi_2} = 0$ then
$\notel{0}{\td{M}{\psi_1\psi_2}} \steps \notf{\td{M}{\psi_1\psi_2}}$,
$\coe{x.\notb{x}}{1}{0}{\td{M}{\psi_1\psi_2}} \steps
\notf{\td{M}{\psi_1\psi_2}}$, and the result follows from
$\coftype[\Psi_2]{\td{M}{\psi_1\psi_2}}{\bool}$.

\item
If $\td{x}{\psi_2} = x'$ then
$\isval{\notel{x'}{\td{M}{\psi_1\psi_2}}}$ and
$\coe{x.\notb{x}}{1}{x'}{\td{M}{\psi_1\psi_2}} \steps
\notel{x'}{\td{M}{\psi_1\psi_2}}$, and so
$\vinper[\Psi_2]
{\notel{x'}{\td{M}{\psi_1\psi_2}}}
{\notel{x'}{\td{M}{\psi_1\psi_2}}}
{\notb{x'}}$
because $\coftype[\Psi_2]{\td{M}{\psi_1\psi_2}}{\bool}$.
\end{enumerate}

\item
If $\td{r}{\psi_1} = 0$ and $\td{r'}{\psi_1} = x$ then
$\coe{x.\notb{x}}{0}{x}{\td{M}{\psi_1}} \steps
\notel{x}{\notf{\td{M}{\psi_1}}}$.

\begin{enumerate}
\item
If $\td{x}{\psi_2} = 0$ then
$\notel{0}{\notf{\td{M}{\psi_1\psi_2}}} \steps
\notf{\notf{\td{M}{\psi_1\psi_2}}}$ and
$\coe{x.\notb{x}}{0}{0}{\td{M}{\psi_1\psi_2}} \steps \td{M}{\psi_1\psi_2}$.
By $\coftype[\Psi_2]{\td{M}{\psi_1\psi_2}}{\bool}$ we have
$\ceqtm[\Psi_2]
{\notf{\notf{\td{M}{\psi_1\psi_2}}}}
{\td{M}{\psi_1\psi_2}}{\bool}$ and in particular
$\inper[\Psi_2]
{\notf{\notf{\td{M}{\psi_1\psi_2}}}}
{\td{M}{\psi_1\psi_2}}{\bool}$.

\item
If $\td{x}{\psi_2} = 1$ then
$\notel{1}{\notf{\td{M}{\psi_1\psi_2}}} \steps
\notf{\td{M}{\psi_1\psi_2}}$,
$\coe{x.\notb{x}}{0}{1}{\td{M}{\psi_1\psi_2}} \steps
\notf{\td{M}{\psi_1\psi_2}}$, and the result follows from
$\coftype[\Psi_2]{\notf{\td{M}{\psi_1\psi_2}}}{\bool}$.

\item
If $\td{x}{\psi_2} = x'$ then
$\isval{\notel{x'}{\notf{\td{M}{\psi_1\psi_2}}}}$,
$\coe{x.\notb{x}}{0}{x'}{\td{M}{\psi_1\psi_2}} \steps
\notel{x'}{\notf{\td{M}{\psi_1\psi_2}}}$, and the result follows from
$\coftype[\Psi_2]{\notf{\td{M}{\psi_1\psi_2}}}{\bool}$.
\end{enumerate}

\item
If $\td{r}{\psi_1} = x$ and $\td{r'}{\psi_1} = 1$ then
$\coftype[\Psi_1]{\td{M}{\psi_1}}{\notb{x}}$ and so
$\td{M}{\psi_1} \evals \notel{x}{N}$ where $\coftype[\Psi_1]{N}{\bool}$.
Therefore
$\coe{x.\notb{x}}{x}{1}{\td{M}{\psi_1}} \steps^*
\coe{x.\notb{x}}{x}{1}{\notel{x}{N}} \steps^* N \evals N_0$.

\begin{enumerate}
\item
If $\td{x}{\psi_2} = 1$ then
$\coe{x.\notb{x}}{1}{1}{\td{M}{\psi_1\psi_2}} \steps \td{M}{\psi_1\psi_2}$.
By $\coftype[\Psi_1]{\td{M}{\psi_1}}{\notb{x}}$ we know that
$\inper[\Psi_2]{\notel{1}{\td{N}{\psi_2}} \steps
\td{N}{\psi_2}}{\td{M}{\psi_1\psi_2}}{\bool}$ and by
$\coftype[\Psi_1]{N}{\bool}$ we know
$\inper[\Psi_2]{\td{N}{\psi_2}}{\td{N_0}{\psi_2}}{\bool}$.
We conclude
$\inper[\Psi_2]{\td{M}{\psi_1\psi_2}}{\td{N_0}{\psi_2}}{\bool}$ as desired.

\item
If $\td{x}{\psi_2} = 0$ then
$\coe{x.\notb{x}}{0}{1}{\td{M}{\psi_1\psi_2}} \steps
\notf{\td{M}{\psi_1\psi_2}}$, and we must show
$\inper[\Psi_2]{\notf{\td{M}{\psi_1\psi_2}}}{\td{N_0}{\psi_2}}{\bool}$.
By $\coftype[\Psi_1]{N}{\bool}$, it suffices to show
$\inper[\Psi_2]{\notf{\td{M}{\psi_1\psi_2}}}{\td{N}{\psi_2}}{\bool}$.
By $\coftype[\Psi_1]{\td{M}{\psi_1}}{\notb{x}}$
we know that for any $\msubsts{\Psi'}{\psi}{\Psi_2}$,
$\inper[\Psi']{\notel{0}{\td{N}{\psi_2\psi}} \steps
\notf{\td{N}{\psi_2\psi}}}{\td{M}{\psi_1\psi_2\psi}}{\bool}$,
and the result follows by \cref{lem:inper-not-swap}.

\item
If $\td{x}{\psi_2} = x'$ then
$\vinper[\Psi_2]
{\td{M}{\psi_1\psi_2} \evals \notel{x'}{N'}}
{\notel{x'}{\td{N}{\psi_2}}}{\notb{x'}}$ where
$\ceqtm[\Psi_2]{N'}{\td{N}{\psi_2}}{\bool}$, and
$\coe{x.\notb{x}}{x'}{1}{\td{M}{\psi_1\psi_2}} \steps^*
\coe{x.\notb{x}}{x'}{1}{\notel{x'}{N'}} \steps^* N'$. Then
$\inper[\Psi_2]{\td{N}{\psi_2}}{\td{N_0}{\psi_2}}{\bool}$ and
$\inper[\Psi_2]{N'}{\td{N}{\psi_2}}{\bool}$ so we have
$\inper[\Psi_2]{N'}{\td{N_0}{\psi_2}}{\bool}$.
\end{enumerate}

\item
If $\td{r}{\psi_1} = x$ and $\td{r'}{\psi_1} = 0$ then
$\coftype[\Psi_1]{\td{M}{\psi_1}}{\notb{x}}$ and so
$\td{M}{\psi_1} \evals \notel{x}{N}$ where $\coftype[\Psi_1]{N}{\bool}$.
Therefore
$\coe{x.\notb{x}}{x}{0}{\td{M}{\psi_1}} \steps^*
\coe{x.\notb{x}}{x}{0}{\notel{x}{N}} \steps^* \notf{N} \evals X_1$.

\begin{enumerate}
\item
If $\td{x}{\psi_2} = 0$ then
$\coe{x.\notb{x}}{0}{0}{\td{M}{\psi_1\psi_2}} \steps \td{M}{\psi_1\psi_2}$.
By $\coftype[\Psi_1]{\td{M}{\psi_1}}{\notb{x}}$ we know that
$\inper[\Psi_2]{\notel{0}{\td{N}{\psi_2}} \steps \notf{\td{N}{\psi_2}}}
{\td{M}{\psi_1\psi_2}}{\bool}$, and by
$\coftype[\Psi_1]{\notf{N}}{\bool}$ we know
$\inper[\Psi_2]{\td{X_1}{\psi_2}}{\notf{\td{N}{\psi_2}}}{\bool}$. Thus
$\inper[\Psi_2]{\td{M}{\psi_1\psi_2}}{\td{X_1}{\psi_2}}{\bool}$.

\item
If $\td{x}{\psi_2} = 1$ then
$\coe{x.\notb{x}}{1}{0}{\td{M}{\psi_1\psi_2}} \steps
\notf{\td{M}{\psi_1\psi_2}}$.
By $\coftype[\Psi_1]{\td{M}{\psi_1}}{\notb{x}}$ we know that for any
$\msubsts{\Psi'}{\psi}{\Psi_2}$,
$\inper[\Psi']{\notel{1}{\td{N}{\psi_2\psi}} \steps \td{N}{\psi_2\psi}}
{\td{M}{\psi_1\psi_2\psi}}{\bool}$, so by \cref{lem:coftype-ceqtm},
$\ceqtm[\Psi_2]{\td{N}{\psi_2}}{\td{M}{\psi_1\psi_2}}{\bool}$, and thus
$\inper[\Psi_2]{\notf{\td{N}{\psi_2}}}{\notf{\td{M}{\psi_1\psi_2}}}{\bool}$.
By $\coftype[\Psi_1]{\notf{N}}{\bool}$, we know
$\inper[\Psi_2]{\td{X_1}{\psi_2}}{\notf{\td{N}{\psi_2}}}{\bool}$. Therefore
$\inper[\Psi_2]{\td{X_1}{\psi_2}}{\notf{\td{M}{\psi_1\psi_2}}}{\bool}$.

\item
If $\td{x}{\psi_2} = x'$ then
$\vinper[\Psi_2]
{\td{M}{\psi_1\psi_2} \evals \notel{x'}{N'}}
{\notel{x'}{\td{N}{\psi_2}}}{\notb{x'}}$ where
$\ceqtm[\Psi_2]{N'}{\td{N}{\psi_2}}{\bool}$, and
$\coe{x.\notb{x}}{x'}{0}{\td{M}{\psi_1\psi_2}} \steps^*
\coe{x.\notb{x}}{x'}{0}{\notel{x'}{N'}} \steps^* \notf{N'}$. Then
$\inper[\Psi_2]{\notf{N'}}{\notf{\td{N}{\psi_2}}}{\bool}$, and by
$\coftype[\Psi_1]{\notf{N}}{\bool}$,
$\inper[\Psi_2]{\td{X_1}{\psi_2}}{\notf{\td{N}{\psi_2}}}{\bool}$, so
$\inper[\Psi_2]{\notf{N'}}{\td{X_1}{\psi_2}}{\bool}$.
\end{enumerate}

\item
If $\td{r}{\psi_1} = x$ and $\td{r'}{\psi_1} = y$ then
$\coftype[\Psi_1]{\td{M}{\psi_1}}{\notb{x}}$ and so
$\td{M}{\psi_1} \evals \notel{x}{N}$ where $\coftype[\Psi_1]{N}{\bool}$.
Therefore
$\coe{x.\notb{x}}{x}{y}{\td{M}{\psi_1}} \steps^*
\coe{x.\notb{x}}{x}{y}{\notel{x}{N}} \steps^* \notel{y}{N}$.

\begin{enumerate}
\item
If $\td{x}{\psi_2} = \e$ and $\td{y}{\psi_2} = \e$ then
$\coe{x.\notb{x}}{\e}{\e}{\td{M}{\psi_1\psi_2}} \steps
\td{M}{\psi_1\psi_2}$. By
$\coftype[\Psi_1]{\td{M}{\psi_1}}{\notb{x}}$, we have
$\inper[\Psi_2]{\notel{\e}{\td{N}{\psi_2}}}{\td{M}{\psi_1\psi_2}}{\bool}$ as
desired.

\item
If $\td{x}{\psi_2} = 0$ and $\td{y}{\psi_2} = 1$ then
$\coe{x.\notb{x}}{0}{1}{\td{M}{\psi_1\psi_2}} \steps
\notf{\td{M}{\psi_1\psi_2}}$, and
$\notel{1}{\td{N}{\psi_2}} \steps \td{N}{\psi_2}$. By
$\coftype[\Psi_1]{\td{M}{\psi_1}}{\notb{x}}$, for any
$\msubsts{\Psi'}{\psi}{\Psi_2}$,
$\inper[\Psi']{\notel{0}{\td{N}{\psi_2\psi}} \steps \notf{\td{N}{\psi_2\psi}}}
{\td{M}{\psi_1\psi_2\psi}}{\bool}$, so by \cref{lem:inper-not-swap} we have
$\inper[\Psi_2]{\td{N}{\psi_2}}{\notf{\td{M}{\psi_1\psi_2}}}{\bool}$ as desired.

\item
If $\td{x}{\psi_2} = 1$ and $\td{y}{\psi_2} = 0$ then
$\coe{x.\notb{x}}{1}{0}{\td{M}{\psi_1\psi_2}} \steps
\notf{\td{M}{\psi_1\psi_2}}$, and
$\notel{0}{\td{N}{\psi_2}} \steps \notf{\td{N}{\psi_2}}$. By
$\coftype[\Psi_1]{\td{M}{\psi_1}}{\notb{x}}$, for any
$\msubsts{\Psi'}{\psi}{\Psi_2}$,
$\inper[\Psi']{\notel{1}{\td{N}{\psi_2\psi}} \steps \td{N}{\psi_2\psi}}
{\td{M}{\psi_1\psi_2\psi}}{\bool}$, so by \cref{lem:coftype-ceqtm}
$\ceqtm[\Psi_2]{\td{N}{\psi_2}}{\td{M}{\psi_1\psi_2}}{\bool}$ and so
$\inper[\Psi_2]{\notf{\td{N}{\psi_2}}}{\notf{\td{M}{\psi_1\psi_2}}}{\bool}$.

\item
If $\td{x}{\psi_2} = 1$ and $\td{y}{\psi_2} = y'$ then
$\coe{x.\notb{x}}{1}{y'}{\td{M}{\psi_1\psi_2}} \steps
\notel{y'}{\td{M}{\psi_1\psi_2}}$, and $\isval{\notel{y'}{\td{N}{\psi_2}}}$. By
$\coftype[\Psi_1]{\td{M}{\psi_1}}{\notb{x}}$, for any
$\msubsts{\Psi'}{\psi}{\Psi_2}$,
$\inper[\Psi']{\notel{1}{\td{N}{\psi_2\psi}} \steps \td{N}{\psi_2\psi}}
{\td{M}{\psi_1\psi_2\psi}}{\bool}$, so by \cref{lem:coftype-ceqtm},
$\ceqtm[\Psi_2]{\td{N}{\psi_2}}{\td{M}{\psi_1\psi_2}}{\bool}$ and so
$\vinper[\Psi_2]{\notel{y'}{\td{M}{\psi_1\psi_2}}}{\notel{y'}{\td{N}{\psi_2}}}{\notb{y'}}$.

\item
If $\td{x}{\psi_2} = 0$ and $\td{y}{\psi_2} = y'$ then
$\coe{x.\notb{x}}{0}{y'}{\td{M}{\psi_1\psi_2}} \steps
\notel{y'}{\notf{\td{M}{\psi_1\psi_2}}}$, and
$\isval{\notel{y'}{\td{N}{\psi_2}}}$. By
$\coftype[\Psi_1]{\td{M}{\psi_1}}{\notb{x}}$, for any
$\msubsts{\Psi'}{\psi}{\Psi_2}$,
$\inper[\Psi']{\notel{0}{\td{N}{\psi_2\psi}} \steps \notf{\td{N}{\psi_2\psi}}}
{\td{M}{\psi_1\psi_2\psi}}{\bool}$, so by \cref{lem:inper-not-swap},
$\ceqtm[\Psi_2]{\td{N}{\psi_2}}{\notf{\td{M}{\psi_1\psi_2}}}{\bool}$ and so
$\vinper[\Psi_2]
{\notel{y'}{\notf{\td{M}{\psi_1\psi_2}}}}
{\notel{y'}{\td{N}{\psi_2}}}{\notb{y'}}$.

\item
If $\td{x}{\psi_2} = x'$ and $\td{y}{\psi_2} = 1$ then
$\vinper[\Psi_2]
{\td{M}{\psi_1\psi_2} \evals \notel{x'}{N'}}
{\notel{x'}{\td{N}{\psi_2}}}{\notb{x'}}$ where
$\ceqtm[\Psi_2]{N'}{\td{N}{\psi_2}}{\bool}$. Moreover,
$\coe{x.\notb{x}}{x'}{1}{\td{M}{\psi_1\psi_2}} \steps^*
\coe{x.\notb{x}}{x'}{1}{\notel{x'}{N'}} \steps^* N'$, and
$\notel{1}{\td{N}{\psi_2}} \steps \td{N}{\psi_2}$. Then
$\inper[\Psi_2]{N'}{\td{N}{\psi_2}}{\bool}$ as desired.

\item
If $\td{x}{\psi_2} = x'$ and $\td{y}{\psi_2} = 0$ then
$\vinper[\Psi_2]
{\td{M}{\psi_1\psi_2} \evals \notel{x'}{N'}}
{\notel{x'}{\td{N}{\psi_2}}}{\notb{x'}}$ where
$\ceqtm[\Psi_2]{N'}{\td{N}{\psi_2}}{\bool}$. Moreover,
$\coe{x.\notb{x}}{x'}{0}{\td{M}{\psi_1\psi_2}} \steps^*
\coe{x.\notb{x}}{x'}{0}{\notel{x'}{N'}} \steps^* \notf{N'}$, and
$\notel{0}{\td{N}{\psi_2}} \steps \notf{\td{N}{\psi_2}}$. Then
$\inper[\Psi_2]{\notf{N'}}{\notf{\td{N}{\psi_2}}}{\bool}$ as desired.

\item
If $\td{x}{\psi_2} = x'$ and $\td{y}{\psi_2} = y'$ then
$\vinper[\Psi_2]
{\td{M}{\psi_1\psi_2} \evals \notel{x'}{N'}}
{\notel{x'}{\td{N}{\psi_2}}}{\notb{x'}}$ where
$\ceqtm[\Psi_2]{N'}{\td{N}{\psi_2}}{\bool}$. Moreover,
$\coe{x.\notb{x}}{x'}{y'}{\td{M}{\psi_1\psi_2}} \steps^*
\coe{x.\notb{x}}{x'}{y'}{\notel{x'}{N'}} \steps \notel{y'}{N'}$,
$\isval{\notel{y'}{\td{N}{\psi_2}}}$, and
$\vinper[\Psi_2]{\notel{y'}{N'}}{\notel{y'}{\td{N}{\psi_2}}}{\bool}$ as desired.
\end{enumerate}
\end{enumerate}

The fifth Kan condition asserts that for any
$\msubsts{(\Psi',x')}{\psi}{(\Psi,x)}$, if
$\coftype[\Psi']{M}{\notb{\dsubst{\td{x}{\psi}}{r}{x'}}}$, then
$\ceqtm[\Psi']{\coe{x'.\notb{\td{x}{\psi}}}{r}{r}{M}}{M}
{\notb{\dsubst{\td{x}{\psi}}{r}{x'}}}$.
By \cref{lem:coftype-ceqtm}, it suffices to show that for any
$\msubsts{\Psi''}{\psi'}{\Psi'}$,
$\inper[\Psi'']
{\coe{x'.\notb{\td{x}{\psi\psi'}}}{\td{r}{\psi'}}{\td{r}{\psi'}}{\td{M}{\psi'}}}
{\td{M}{\psi'}}{\notb{\td{\dsubst{\td{x}{\psi}}{r}{x'}}{\psi'}}}$; we already
know
$\inper[\Psi'']{\td{M}{\psi'}}{\td{M}{\psi'}}
{\notb{\td{\dsubst{\td{x}{\psi}}{r}{x'}}{\psi'}}}$.

\begin{enumerate}
\item $\td{x}{\psi\psi'} = \e$.
Then $\coe{x'.\notb{\e}}{\td{r}{\psi'}}{\td{r}{\psi'}}{\td{M}{\psi'}} \steps
\td{M}{\psi'}$.

\item $\td{x}{\psi\psi'} = y \neq x'$.
Then $\coe{x'.\notb{y}}{\td{r}{\psi'}}{\td{r}{\psi'}}{\td{M}{\psi'}} \steps
\td{M}{\psi'}$.

\item
$\td{x}{\psi\psi'} = x'$ and $\td{r}{\psi'} = \e$.
Then $\coe{x'.\notb{x'}}{\e}{\e}{\td{M}{\psi'}} \steps \td{M}{\psi'}$.

\item
$\td{x}{\psi\psi'} = x'$ and $\td{r}{\psi'} = y$.
Then
$\coe{x'.\notb{x'}}{y}{y}{\td{M}{\psi'}} \steps^*
 \coe{x'.\notb{x'}}{y}{y}{\notel{y}{N}} \steps \notel{y}{N}$,
where $\td{M}{\psi'} \evals \notel{y}{N}$ and $\coftype[\Psi'']{N}{\bool}$;
thus $\vinper[\Psi'']{\notel{y}{N}}{\notel{y}{N}}{\notb{y}}$, which is what we
wanted to show.
\end{enumerate}

The proofs of the first three Kan conditions rely on the fourth Kan condition,
as well as one additional lemma:

\begin{lemma}\label{lem:ceqtm-notel-coe}
If $\coftype{M}{\notb{r}}$, then
$\ceqtm{\notel{r}{\coe{x.\notb{x}}{r}{1}{M}}}{M}{\notb{r}}$.
\end{lemma}
\begin{proof}
The introduction rule and fourth Kan condition of $\notb{x}$ imply
$\coftype{\notel{r}{\coe{x.\notb{x}}{r}{1}{M}}}{\notb{r}}$. Therefore by
\cref{lem:coftype-ceqtm} it suffices to show that for any $\psitd$,
$\inper[\Psi']
{\notel{\td{r}{\psi}}{\coe{x.\notb{x}}{\td{r}{\psi}}{1}{\td{M}{\psi}}}}
{\td{M}{\psi}}
{\notb{\td{r}{\psi}}}$.

\begin{enumerate}
\item If $\td{r}{\psi} = 0$ then
$\notel{0}{\coe{x.\notb{x}}{0}{1}{\td{M}{\psi}}} \steps
\notf{\coe{x.\notb{x}}{0}{1}{\td{M}{\psi}}} \steps
\notf{\notf{\td{M}{\psi}}}$.
By \cref{lem:ceqpretype-ceqtm},
$\coftype[\Psi']{\td{M}{\psi}}{\bool}$, so
$\ceqtm[\Psi']{\notf{\notf{\td{M}{\psi}}}}{\td{M}{\psi}}{\bool}$ and therefore
$\inper[\Psi']{\notf{\notf{\td{M}{\psi}}}}{\td{M}{\psi}}{\bool}$.

\item If $\td{r}{\psi} = 1$ then
$\notel{1}{\coe{x.\notb{x}}{1}{1}{\td{M}{\psi}}} \steps
\coe{x.\notb{x}}{1}{1}{\td{M}{\psi}} \steps \td{M}{\psi}$, and
$\inper[\Psi']{\td{M}{\psi}}{\td{M}{\psi}}{\bool}$.

\item If $\td{r}{\psi} = x$ then
$\isval{\notel{x}{\coe{x.\notb{x}}{x}{1}{\td{M}{\psi}}}}$. We know
$\inper[\Psi']{\td{M}{\psi}}{\td{M}{\psi}}{\notb{x}}$, so
$\td{M}{\psi} \evals \notel{x}{N}$ where
$\coftype[\Psi']{N}{\bool}$. To show
$\vinper[\Psi']
{\notel{x}{\coe{x.\notb{x}}{x}{1}{\td{M}{\psi}}}}
{\notel{x}{N}}{\notb{x}}$, we must show
$\ceqtm[\Psi']{\coe{x.\notb{x}}{x}{1}{\td{M}{\psi}}}{N}{\bool}$.
Again, by \cref{lem:coftype-ceqtm} it suffices to show that for any
$\msubsts{\Psi''}{\psi'}{\Psi'}$,
$\inper[\Psi'']
{\coe{x.\notb{x}}{\td{x}{\psi'}}{1}{\td{M}{\psi\psi'}}}
{\td{N}{\psi'}}{\bool}$.
\begin{enumerate}
\item If $\td{x}{\psi'} = 0$ then
$\coe{x.\notb{x}}{0}{1}{\td{M}{\psi\psi'}} \steps
\notf{\td{M}{\psi\psi'}}$. Because
$\coftype[\Psi']{\td{M}{\psi}}{\notb{x}}$ and
$\td{M}{\psi} \evals \notel{x}{N}$,
we know that for any $\msubsts{\Psi'''}{\psi''}{\Psi''}$,
$\inper[\Psi''']
{\notel{0}{\td{N}{\psi'\psi''}} \steps \notf{\td{N}{\psi'\psi''}}}
{\td{M}{\psi\psi'\psi''}}{\bool}$.
Therefore by \cref{lem:inper-not-swap} we have
$\inper[\Psi'']{\notf{\td{M}{\psi\psi'}}}{\td{N}{\psi'}}{\bool}$, which is what
we needed.

\item If $\td{x}{\psi'} = 1$ then
$\coe{x.\notb{x}}{1}{1}{\td{M}{\psi\psi'}} \steps \td{M}{\psi\psi'}$. By
$\coftype[\Psi']{\td{M}{\psi}}{\notb{x}}$ and
$\td{M}{\psi} \evals \notel{x}{N}$, we have
$\inper[\Psi'']
{\notel{1}{\td{N}{\psi'}} \steps \td{N}{\psi'}}
{\td{M}{\psi\psi'}}{\bool}$, which is what we needed.

\item If $\td{x}{\psi'} = x'$ then by
$\coftype[\Psi']{\td{M}{\psi}}{\notb{x}}$ and
$\td{M}{\psi} \evals \notel{x}{N}$, we have
$\inper[\Psi'']{\notel{x'}{\td{N}{\psi'}}}{\td{M}{\psi\psi'}}{\notb{x'}}$,
so $\td{M}{\psi\psi'} \evals \notel{x'}{N'}$ and
$\ceqtm[\Psi'']{\td{N}{\psi'}}{N'}{\bool}$. Therefore
$\coe{x.\notb{x}}{x'}{1}{\td{M}{\psi\psi'}} \steps^*
\coe{x.\notb{x}}{x'}{1}{\notel{x'}{N'}} \steps
\notel{1}{N'} \steps N'$ and we must show
$\inper[\Psi'']{N'}{\td{N}{\psi'}}{\notb{x'}}$, which follows from
$\ceqtm[\Psi'']{\td{N}{\psi'}}{N'}{\bool}$.
\qedhere
\end{enumerate}
\end{enumerate}
\end{proof}
In particular, if $\coftype{M}{\bool}$, then
$\ceqtm{\notf{\coe{x.\notb{x}}{0}{1}{M}}}{M}{\bool}$.

The first Kan condition asserts that, for any
$\msubsts{\Psi'}{\psi}{(\Psi,x)}$, if
\begin{enumerate}
\item $\ceqtm[\Psi']{M}{O}{\notb{\td{x}{\psi}}}$,
\item $\ceqtmres[\Psi',y]{r_i=\e,r_j=\e'}{N^\e_i}{N^{\e'}_j}{\notb{\td{x}{\psi}}}$
for any $i,j,\e,\e'$,
\item $\ceqtmres[\Psi',y]{r_i=\e}{N^\e_i}{P^\e_i}{\notb{\td{x}{\psi}}}$
for any $i,\e$, and
\item $\ceqtmres[\Psi']{r_i=\e}{\dsubst{N^\e_i}{r}{y}}{M}{\notb{\td{x}{\psi}}}$
for any $i,\e$, then
\end{enumerate}
then $\ceqtm[\Psi']{\hcomgeneric{\etc{r_i}}{\notb{\td{x}{\psi}}}}%
{\hcom{\etc{r_i}}{\notb{\td{x}{\psi}}}{r}{r'}{O}{\etc{y.P^\e_i}}}{\notb{\td{x}{\psi}}}$.

Let $\msubsts{\Psi_1}{\psi_1}{\Psi'}$ and
$\msubsts{\Psi_2}{\psi_2}{\Psi_1}$. We will again focus on the unary case.

\begin{enumerate}
\item If $\td{x}{\psi\psi_1} = \e$ then
$\coftype[\Psi_1]
{\hcom{\etc{\td{r_i}{\psi_1}}}{\notb{\e}}{\td{r}{\psi_1}}{\td{r'}{\psi_1}}
{\td{M}{\psi_1}}{\etc{y.\td{N^\e_i}{\psi_1}}}}
{\notb{\e}}$
by the first Kan condition of $\cwftype[\Psi_1]{\notb{\e}}$.
Therefore $\td{\hcomsym}{\psi_1} \evals X_1$ and
$\inper[\Psi_2]{\td{X_1}{\psi_2}}{\td{\hcomsym}{\psi_1\psi_2}}{\bool}$.

\item If $\td{x}{\psi\psi_1} = x'$ then
\[
\td{\hcomsym}{\psi_1} \evals
\notel{x'}{
  \hcom{\etc{\td{r_i}{\psi_1}}}{\bool}{\td{r}{\psi_1}}{\td{r'}{\psi_1}}
  {\coe{x.\notb{x}}{x'}{1}{\td{M}{\psi_1}}}
  {\etc{y.\coe{x.\notb{x}}{x'}{1}{\td{N^\e_i}{\psi_1}}}}}.
\]
Let $H$ be the argument of the above $\notel{x'}{-}$; we must show that
$\inper[\Psi_2]
{\td{\hcomsym}{\psi_1\psi_2}}
{\notel{\td{x'}{\psi_2}}{\td{H}{\psi_2}}}
{\notb{\td{x'}{\psi_2}}}$.
By the fourth Kan condition of $\cpretype[\Psi_1,x]{\notb{x}}$ and the first Kan
condition of $\cwftype[\Psi_1]{\bool}$,
$\coftype[\Psi_1]{H}{\bool}$.

\begin{enumerate}
\item If $\td{x'}{\psi_2} = 0$ then each side of
$\inper[\Psi_2]
{\td{\hcomsym}{\psi_1\psi_2}}
{\notel{\td{x'}{\psi_2}}{\td{H}{\psi_2}}}
{\notb{\td{x'}{\psi_2}}}$ steps once to:
\[
\inper[\Psi_2]
{\hcom{\etc{\td{r_i}{\psi_1\psi_2}}}{\bool}{\td{r}{\psi_1\psi_2}}{\td{r'}{\psi_1\psi_2}}
{\td{M}{\psi_1\psi_2}}{\etc{y.\td{N^\e_i}{\psi_1\psi_2}}}}
{\notf{\td{H}{\psi_2}}}{\bool}.
\]

If some $\td{r_i}{\psi_1\psi_2} = \e$ then by the elimination rule and third Kan
condition of $\cwftype[\Psi_2]{\bool}$, it suffices to show
$\inper[\Psi_2]
{\td{\dsubst{N^\e_i}{r'}{y}}{\psi_1\psi_2}}
{\notf{\coe{x.\notb{x}}{0}{1}{\td{\dsubst{N^\e_i}{r'}{y}}{\psi_1\psi_2}}}}
{\bool}$,
which follows from \cref{lem:ceqtm-notel-coe}.

If $\td{r}{\psi_1\psi_2}=\td{r'}{\psi_1\psi_2}$,
then by the elimination rule and second Kan condition of
$\cwftype[\Psi_2]{\bool}$, it suffices to show
$\inper[\Psi_2]
{\td{M}{\psi_1\psi_2}}
{\notf{\coe{x.\notb{x}}{0}{1}{\td{M}{\psi_1\psi_2}}}}
{\bool}$,
which follows from \cref{lem:ceqtm-notel-coe}.

If all $\td{r_i}{\psi_1\psi_2}$ are dimension names and
$\td{r}{\psi_1\psi_2}\neq\td{r'}{\psi_1\psi_2}$, then
the left-hand side is a value and the right-hand side steps twice to
\[
\hcom{\etc{\td{r_i}{\psi_1\psi_2}}}{\bool}{\td{r}{\psi_1\psi_2}}{\td{r'}{\psi_1\psi_2}}
{\coe{\_.\bool}{\td{r}{\psi_1\psi_2}}{\td{r'}{\psi_1\psi_2}}
  {\notf{\coe{x.\notb{x}}{0}{1}{\td{M}{\psi_1\psi_2}}}}}
{\dots}
\]
which $\eq$ the left-hand side by \cref{lem:ceqtm-notel-coe}, the first Kan
condition of $\cwftype[\Psi_2]{\bool}$, and \cref{lem:expansion}.

\item If $\td{x'}{\psi_2} = 1$ then each side steps once to
\[
\inper[\Psi_2]
{\hcom{\etc{\td{r_i}{\psi_1\psi_2}}}{\bool}{\td{r}{\psi_1\psi_2}}{\td{r'}{\psi_1\psi_2}}
{\td{M}{\psi_1\psi_2}}{\etc{y.\td{N^\e_i}{\psi_1\psi_2}}}}
{\td{H}{\psi_2}}{\bool}
\]
which follows by the first Kan condition of $\cwftype[\Psi_2]{\bool}$ and
$\ceqtm[\Psi_2]{\td{M}{\psi_1\psi_2}}{\coe{x.\notb{x}}{1}{1}{\td{M}{\psi_1\psi_2}}}
{\bool}$ (by the computation rule for $\notb{x}$).

\item If $\td{x'}{\psi_2} = x''$ then each side steps once to
\[
\vinper[\Psi_2]{\notel{x''}{\td{H}{\psi_2}}}{\notel{x''}{\td{H}{\psi_2}}}{\notb{x''}}
\]
which follows from $\coftype[\Psi_2]{\td{H}{\psi_2}}{\bool}$.
\end{enumerate}
\end{enumerate}

The second Kan condition asserts that, for any
$\msubsts{\Psi'}{\psi}{(\Psi,x)}$, if
\begin{enumerate}
\item $\coftype[\Psi']{M}{\notb{\td{x}{\psi}}}$,
\item $\ceqtmres[\Psi',y]{r_i=\e,r_j=\e'}{N^\e_i}{N^{\e'}_j}{\notb{\td{x}{\psi}}}$
for any $i,j,\e,\e'$, and
\item $\ceqtmres[\Psi']{r_i=\e}{\dsubst{N^\e_i}{r}{y}}{M}{\notb{\td{x}{\psi}}}$
for any $i,\e$,
\end{enumerate}
then $\ceqtm[\Psi']
{\hcom{\etc{r_i}}{\notb{\td{x}{\psi}}}{r}{r}{M}{\etc{y.N^\e_i}}}
{M}{\notb{\td{x}{\psi}}}$.

By \cref{lem:coftype-ceqtm}, it suffices to show that for all
$\msubsts{\Psi''}{\psi'}{\Psi'}$,
\[
\inper[\Psi'']
{\hcom{\etc{\td{r_i}{\psi'}}}{\notb{\td{x}{\psi\psi'}}}{\td{r}{\psi'}}{\td{r}{\psi'}}
{\td{M}{\psi'}}{\etc{y.\td{N^\e_i}{\psi'}}}}
{\td{M}{\psi'}}
{\notb{\td{x}{\psi\psi'}}}.
\]
\begin{enumerate}
\item If $\td{x}{\psi\psi'} = \e$ then 
$\ceqtm[\Psi'']
{\hcom{\etc{\td{r_i}{\psi'}}}{\notb{\e}}{\td{r}{\psi'}}{\td{r}{\psi'}}
{\td{M}{\psi'}}{\etc{y.\td{N^\e_i}{\psi'}}}}
{\td{M}{\psi'}}
{\notb{\e}}$
by the second Kan condition of $\cpretype[\Psi'']{\notb{\e}}$.

\item If $\td{x}{\psi\psi'} = x'$ then the left side steps to
\[
\notel{x'}{
  \hcom{\etc{\td{r_i}{\psi'}}}{\bool}{\td{r}{\psi'}}{\td{r}{\psi'}}
  {\coe{x.\notb{x}}{x'}{1}{\td{M}{\psi'}}}
  {\etc{y.\coe{x.\notb{x}}{x'}{1}{\td{N^\e_i}{\psi'}}}}}
\]
The second Kan condition of $\bool$ (along with the introduction rule of
$\notb{x'}$ and fourth Kan condition of $\cpretype[\Psi'',x]{\notb{x}}$) implies
this is
$\ceqtm[\Psi'']{-}
{\notel{x'}{\coe{x.\notb{x}}{x'}{1}{\td{M}{\psi'}}}}
{\notb{x'}}$
which by \cref{lem:ceqtm-notel-coe} is
$\ceqtm[\Psi'']{-}{\td{M}{\psi'}}{\notb{x'}}$.
\end{enumerate}

The third Kan condition asserts that, for any $\msubsts{\Psi'}{\psi}{(\Psi,x)}$,
if $r_i = \e$ for some $i$,
\begin{enumerate}
\item $\coftype[\Psi']{M}{\notb{\td{x}{\psi}}}$,
\item $\ceqtmres[\Psi',y]{r_i=\e,r_j=\e'}{N^\e_i}{N^{\e'}_j}{\notb{\td{x}{\psi}}}$
for any $i,j,\e,\e'$, and
\item $\ceqtmres[\Psi']{r_i=\e}{\dsubst{N^\e_i}{r}{y}}{M}{\notb{\td{x}{\psi}}}$
for any $i,\e$,
\end{enumerate}
then $\ceqtm[\Psi']
{\hcom{\etc{r_i}}{\notb{\td{x}{\psi}}}{r}{r'}{M}{\etc{y.N^\e_i}}}
{\dsubst{N^\e_i}{r'}{y}}{\notb{\td{x}{\psi}}}$.

By \cref{lem:coftype-ceqtm}, it suffices to show that for all
$\msubsts{\Psi''}{\psi'}{\Psi'}$,
\[
\inper[\Psi'']
{\hcom{\etc{\td{r_i}{\psi'}}}{\notb{\td{x}{\psi\psi'}}}{\td{r}{\psi'}}{\td{r'}{\psi'}}
{\td{M}{\psi'}}{\etc{y.\td{N^\e_i}{\psi'}}}}
{\dsubst{\td{N^\e_i}{\psi'}}{\td{r'}{\psi'}}{y}}
{\notb{\td{x}{\psi\psi'}}}.
\]
\begin{enumerate}
\item If $\td{x}{\psi\psi'} = \e'$ then by the third Kan condition of
$\cpretype[\Psi'']{\notb{\e'}}$ (since $\td{r_i}{\psi'} = \e$),
$\ceqtm[\Psi_1]{\td{\hcomsym}{\psi'}}
{\td{\dsubst{N^\e}{r'}{y}}{\psi'}}
{\notb{\e'}}$.

\item If $\td{x}{\psi\psi'} = x'$ then
\[
\td{\hcomsym}{\psi'} \steps
\notel{x'}{
  \hcom{\etc{\td{r_i}{\psi'}}}{\bool}{\td{r}{\psi'}}{\td{r'}{\psi'}}
  {\coe{x.\notb{x}}{x'}{1}{\td{M}{\psi'}}}
  {\etc{y.\coe{x.\notb{x}}{x'}{1}{\td{N^\e_i}{\psi'}}}}}
\]
By the introduction rule for $\notb{x'}$ and the third Kan condition of $\bool$,
this term is
$\ceqtm[\Psi'']{-}
{\notel{x'}{\coe{x.\notb{x}}{x'}{1}{\td{\dsubst{N^\e_i}{r'}{y}}{\psi'}}}}
{\notb{x'}}$, which by \cref{lem:ceqtm-notel-coe} is
$\ceqtm[\Psi'']{-}{\td{\dsubst{N^\e_i}{r'}{y}}{\psi'}}{\notb{x'}}$.
\end{enumerate}

\paragraph{Cubical}
For any $\msubsts{\Psi'}{\psi}{(\Psi,x)}$ and
$\vinper[\Psi']{M}{N}{A_0}$ (where $\notb{\td{x}{\psi}}\evals A_0$) then
$\ceqtm[\Psi']{M}{N}{\notb{\td{x}{\psi}}}$.

If $\td{x}{\psi} = \e$ then $A_0 = \bool$ and
$\ceqtm[\Psi']{M}{N}{\bool}$ (because $\cpretype[\Psi']{\bool}$ is cubical), so
$\ceqtm[\Psi']{M}{N}{\notb{\e}}$.
If $\td{x}{\psi} = x'$ then
$M = \notel{x'}{O}$,
$N = \notel{x'}{P}$,
$\ceqtm[\Psi']{O}{P}{\bool}$, and the result follows by the introduction rule
for $\notb{x'}$.

\newpage
\section{Proof theory}
\label{sec:proof-theory}

As mentioned in the introduction, there is wide latitude in the choice
of proof theories for computational type theory.  Here we consider
some rules inspired by the formal cubical type theories given
by~\citet{cohen2016cubical} and \citet{licata2014cubical} so as to make clear
that our computational semantics are a valid interpretation of those rules.
However, these semantics may be used to justify concepts, such as
strict types, that are not currently considered in the formal setting.
We emphasize that there is no strong reason to limit consideration to
inductively defined proof theories.  The role of a proof theory is to
provide access to the truth, in particular to support mechanization.
But there are methods of accessing the truth, such as decision
procedures for arithmetic, that do not fit into the conventional setup
for proof theory.

For the sake of concision and clarity, we state the following rules in
\emph{local form}, extending them to \emph{global form} by \emph{uniformity},
also called \emph{naturality}. (This format was suggested
by~\citet{martin1984intuitionistic}, itself inspired by Gentzen's original
concept of natural deduction.)
All the rules for dependent function and pair types should include the
hypotheses $\cwftype{A}$ and $\wftype{\oft aA}{B}$; the rules for identification
types should include the hypothesis $\cwftype{A}$.

Recall that $\Psi$ and $\Xi$ are unordered sets, and the equations in $\Xi$ are
also unordered. $\J$ stands for any type equality or element equality judgment.
Rather than stating every rule with an arbitrary context restriction $\Xi$, we
introduce context restrictions in the hypotheses of the $\hcomsym$ typing rules,
and discharge them with a special set of ``restriction rules.''

While the theorems in \cref{sec:types} are stated only for closed terms, the
corresponding generalizations to open-term sequents follow by the definition of
the open judgments, the fact that the introduction and elimination rules respect
equality (proven in \cref{sec:types}), and the fact that all substitutions
commute with term formers.

\paragraph{Structural rules}

\[
\infer
  {\oftype{\oft{a}{A}}{a}{A}}
  {\cwftype{A}}
\qquad
\infer
  {\judgctx{\oft aA}{\J}}
  {\judg{\J} & \cwftype{A}}
\qquad
\infer
  {\judg[\Psi']{\td{\J}{\psi}}}
  {\judg{\J} & \psitd}
\]

\[
\infer
  {\ceqtype{A'}{A}}
  {\ceqtype{A}{A'}}
\qquad
\infer
  {\ceqtype{A}{A''}}
  {\ceqtype{A}{A'} & \ceqtype{A'}{A''}}
\]

\[
\infer
  {\ceqtm{M}{M'}{A}}
  {\ceqtm{M'}{M}{A}}
\qquad
\infer
  {\ceqtm{M}{M''}{A}}
  {\ceqtm{M}{M'}{A} & \ceqtm{M'}{M''}{A}}
\]

\[
\infer
  {\ceqtm{M}{M'}{A'}}
  {\ceqtm{M}{M'}{A} & \ceqtype{A}{A'}}
\]

\[
\infer
  {\ceqtype{\subst{B}{N}{a}}{\subst{B'}{N'}{a}}}
  {\eqtype{\oft{a}{A}}{B}{B'} & \ceqtm{N}{N'}{A}}
\qquad
\infer
  {\ceqtm{\subst{M}{N}{a}}{\subst{M'}{N'}{a}}{\subst{B}{N}{a}}}
  {\eqtm{\oft{a}{A}}{M}{M'}{B} & \ceqtm{N}{N'}{A}}
\]

\paragraph{Dependent function types}

\[
\infer
  {\ceqtype{\picl{a}{A}{B}}{\picl{a}{A'}{B'}}}
  {\ceqtype{A}{A'} &
   \eqtype{\oft aA}{B}{B'}}
\]

\[
\infer
  {\ceqtm{\lam{a}{M}}{\lam{a}{M'}}{\picl{a}{A}{B}}}
  {\eqtm{\oft aA}{M}{M'}{B}}
\]

\[
\infer
  {\ceqtm{\app{M}{N}}{\app{M'}{N'}}{\subst{B}{N}{a}}}
  {\ceqtm{M}{M'}{\picl{a}{A}{B}} &
   \ceqtm{N}{N'}{A}}
\]

\[
\infer
  {\ceqtm{\app{\lam{a}{M}}{N}}{\subst{M}{N}{a}}{\subst{B}{N}{a}}}
  {\oftype{\oft aA}{M}{B} &
   \coftype{N}{A}}
\]

\[
\infer
  {\ceqtm{M}{\lam{a}{\app{M}{a}}}{\picl{a}{A}{B}}}
  {\coftype{M}{\picl{a}{A}{B}}}
\]

\paragraph{Dependent pair types}

\[
\infer
  {\ceqtype{\sigmacl{a}{A}{B}}{\sigmacl{a}{A'}{B'}}}
  {\ceqtype{A}{A'} &
   \eqtype{\oft aA}{B}{B'}}
\]

\[
\infer
  {\ceqtm{\pair{M}{N}}{\pair{M'}{N'}}{\sigmacl{a}{A}{B}}}
  {\ceqtm{M}{M'}{A} &
   \ceqtm{N}{N'}{\subst{B}{M}{a}}}
\]

\[
\infer
  {\ceqtm{\fst{P}}{\fst{P'}}{A}}
  {\ceqtm{P}{P'}{\sigmacl{a}{A}{B}}}
\qquad
\infer
  {\ceqtm{\snd{P}}{\snd{P'}}{\subst{B}{\fst{P}}{a}}}
  {\ceqtm{P}{P'}{\sigmacl{a}{A}{B}}}
\]

\[
\infer
  {\ceqtm{\fst{\pair{M}{N}}}{M}{A}}
  {\coftype{M}{A} &
   \coftype{N}{\subst{B}{M}{a}}}
\qquad
\infer
  {\ceqtm{\snd{\pair{M}{N}}}{N}{\subst{B}{M}{a}}}
  {\coftype{M}{A} &
   \coftype{N}{\subst{B}{M}{a}}}
\]

\[
\infer
  {\ceqtm{P}{\pair{\fst{P}}{\snd{P}}}{\sigmacl{a}{A}{B}}}
  {\coftype{P}{\sigmacl{a}{A}{B}}}
\]

\paragraph{Identification types}

\[
\infer
  {\ceqtype{\Id{x.A}{P_0}{P_1}}{\Id{x.A'}{P_0'}{P_1'}}}
  {\ceqtype[\Psi,x]{A}{A'} &
   \ceqtm{P_0}{P_0'}{\dsubst{A}{0}{x}} &
   \ceqtm{P_1}{P_1'}{\dsubst{A}{1}{x}}}
\]

\[
\infer
  {\ceqtm{\dlam{x}{M}}{\dlam{x}{M'}}{\Id{x.A}{P_0}{P_1}}}
  {\ceqtm[\Psi,x]{M}{M'}{A} &
   \ceqtm{\dsubst{M}{0}{x}}{P_0}{\dsubst{A}{0}{x}} &
   \ceqtm{\dsubst{M}{1}{x}}{P_1}{\dsubst{A}{1}{x}}}
\]

\[
\infer
  {\ceqtm{\dapp{M}{r}}{\dapp{M'}{r}}{\dsubst{A}{r}{x}}}
  {\ceqtm{M}{M'}{\Id{x.A}{P_0}{P_1}}}
\qquad
\infer
  {\ceqtm{\dapp{M}{\e}}{P_\e}{\dsubst{A}{\e}{x}}}
  {\coftype{M}{\Id{x.A}{P_0}{P_1}}}
\]

\[
\infer
  {\ceqtm{\dapp{(\dlam{x}{M})}{r}}{\dsubst{M}{r}{x}}{\dsubst{A}{r}{x}}}
  {\coftype[\Psi,x]{M}{A}}
\]

\[
\infer
  {\ceqtm{M}{\dlam{x}{(\dapp{M}{x})}}{\Id{x.A}{P_0}{P_1}}}
  {\coftype{M}{\Id{x.A}{P_0}{P_1}}}
\]

\paragraph{Booleans}

\[
\infer{\cwftype{\bool}}{}
\]

\[
\infer{\coftype{\true}{\bool}}{}
\qquad
\infer{\coftype{\false}{\bool}}{}
\]

\[
\infer
  {\ceqtm{\ifb{a.A}{M}{T}{F}}{\ifb{a.A'}{M'}{T'}{F'}}{\subst{A}{M}{a}}}
  {\eqtype{\oft{a}{\bool}}{A}{A'} &
   \ceqtm{M}{M'}{\bool} &
   \ceqtm{T}{T'}{\subst{A}{\true}{a}} &
   \ceqtm{F}{F'}{\subst{A}{\false}{a}}}
\]

\[
\infer
  {\ceqtm{\ifb{a.A}{\true}{T}{F}}{T}{\subst{A}{\true}{a}}}
  {\wftype{\oft{a}{\bool}}{A} &
   \coftype{T}{\subst{A}{\true}{a}} &
   \coftype{F}{\subst{A}{\false}{a}}}
\]

\[
\infer
  {\ceqtm{\ifb{a.A}{\false}{T}{F}}{F}{\subst{A}{\false}{a}}}
  {\wftype{\oft{a}{\bool}}{A} &
   \coftype{T}{\subst{A}{\true}{a}} &
   \coftype{F}{\subst{A}{\false}{a}}}
\]

\paragraph{Circle}

\[
\infer{\cwftype{\C}}{}
\]

\[
\infer
  {\coftype{\base}{\C}}
  {}
\qquad
\infer
  {\coftype{\lp{r}}{\C}}
  {}
\qquad
\infer
  {\ceqtm{\lp{\e}}{\base}{\C}}
  {}
\]

\[
\infer
  {\ceqtm{\Celim{a.A}{M}{P}{x.L}}{\Celim{a.A'}{M'}{P'}{x.L'}}{\subst{A}{M}{a}}}
  {\begin{array}{ll}
   &\eqtype{\oft{a}{\C}}{A}{A'} \\
   &\ceqtm{M}{M'}{\C} \\
   &\ceqtm{P}{P'}{\subst{A}{\base}{a}} \\
   &\ceqtm[\Psi,x]{L}{L'}{\subst{A}{\lp{x}}{a}} \\
   (\forall\e) &\ceqtm{\dsubst{L}{\e}{x}}{P}{\subst{A}{\base}{a}}
   \end{array}}
\]

\[
\infer
  {\ceqtm{\Celim{a.A}{\base}{P}{x.L}}{P}{\subst{A}{\base}{a}}}
  {\wftype{\oft{a}{\C}}{A} &
   \coftype[\Psi,x]{L}{\subst{A}{\lp{x}}{a}} &
   (\forall\e)\ \ceqtm{\dsubst{L}{\e}{x}}{P}{\subst{A}{\base}{a}}}
\]

\[
\infer
  {\ceqtm{\Celim{a.A}{\lp{r}}{P}{x.L}}{\dsubst{L}{r}{x}}{\subst{A}{\lp{r}}{a}}}
  {\wftype{\oft{a}{\C}}{A} &
   \coftype[\Psi,x]{L}{\subst{A}{\lp{x}}{a}} &
   (\forall\e)\ \ceqtm{\dsubst{L}{\e}{x}}{P}{\subst{A}{\base}{a}}}
\]

\paragraph{Hcom}

\[
\infer
  {\ceqtm{\hcomgeneric{\etc{r_i}}{A}}
         {\hcom{\etc{r_i}}{A'}{r}{r'}{O}{\etc{y.P^\e_i}}}{A}}
  {\begin{array}{ll}
   &\ceqtype{A}{A'} \\
   &\ceqtm{M}{O}{A} \\
   (\forall i,\e) &\ceqtmres[\Psi,y]{r_i = \e}{N^\e_i}{P^\e_i}{A} \\
   (\forall i,j,\e,\e') &\ceqtmres[\Psi,y]{r_i = \e,r_j = \e'}{N^\e_i}{N^{\e'}_j}{A} \\
   (\forall i,\e) &\ceqtmres{r_i = \e}{\dsubst{N^\e_i}{r}{y}}{M}{A}
   \end{array}}
\]

\[
\infer
  {\ceqtm{\hcom{\etc{r_i}}{A}{r}{r}{M}{\etc{y.N^\e_i}}}{M}{A}}
  {\begin{array}{ll}
   &\cwftype{A} \\
   &\coftype{M}{A} \\
   (\forall i,j,\e,\e') &\ceqtmres[\Psi,y]{r_i = \e,r_j = \e'}{N^\e_i}{N^{\e'}_j}{A} \\
   (\forall i,\e) &\ceqtmres{r_i = \e}{\dsubst{N^\e_i}{r}{y}}{M}{A}
   \end{array}}
\]

\[
\infer
  {\ceqtm{\hcomgeneric{r_1,\dots,r_{i-1},\e,r_{i+1},\dots,r_n}{A}}{\dsubst{N^\e_i}{r'}{y}}{A}}
  {\begin{array}{ll}
   &\cwftype{A} \\
   &\coftype{M}{A} \\
   (\forall i,j,\e,\e') &\ceqtmres[\Psi,y]{r_i = \e,r_j = \e'}{N^\e_i}{N^{\e'}_j}{A} \\
   (\forall i,\e) &\ceqtmres{r_i = \e}{\dsubst{N^\e_i}{r}{y}}{M}{A}
   \end{array}}
\]

\paragraph{Restriction rules}

\[
\infer
  {\judgres{\cdot}{\J}}
  {\judg{\J}}
\qquad
\infer
  {\judgres{\Xi,r=r}{\J}}
  {\judgres{\Xi}{\J}}
\]

\[
\infer
  {\judgres{\Xi,0=1}{\J}}
  {}
\qquad
\infer
  {\judgres[\Psi,x]{x=0,x=1}{\J}}
  {}
\]

\[
\infer
  {\judgres[\Psi,x]{x=\e}{\J}}
  {\judg{\dsubst{\J}{\e}{x}}}
\qquad
\infer
  {\judgres[\Psi,x,y]{x=\e,y=\e'}{\J}}
  {\judg{\dsubst{\dsubst{\J}{\e}{x}}{\e'}{y}}}
\]

\paragraph{Coe}

\[
\infer
  {\ceqtm{\coegeneric{x.A}}{\coe{x.A'}{r}{r'}{N}}{\dsubst{A}{r'}{x}}}
  {\ceqtype[\Psi,x]{A}{A'} &
   \ceqtm{M}{N}{\dsubst{A}{r}{x}}}
\]

\[
\infer
  {\ceqtm{\coe{x.A}{r}{r}{M}}{M}{\dsubst{A}{r}{x}}}
  {\cwftype[\Psi,x]{A} &
   \coftype{M}{\dsubst{A}{r}{x}}}
\]

\paragraph{Not}

\[
\infer{\cwftype{\notb{r}}}{}
\qquad
\infer{\ceqtype{\notb{\e}}{\bool}}{}
\]

\[
\infer
  {\ceqtm{\coe{x.\notb{x}}{\e}{\e}{M}}{M}{\bool}}
  {\coftype{M}{\bool}}
\qquad
\infer
  {\ceqtm{\coe{x.\notb{x}}{\e}{\eb}{M}}{\notf{M}}{\bool}}
  {\coftype{M}{\bool}}
\]


\appendix
\section{Further developments}
\label{appendix}

This appendix contains results obtained since July 2016.

\subsection{Strict booleans}

Extend the term syntax by $\sbool$, and the operational semantics by

\[
\infer
  {\isval{\sbool}}
  {}
\qquad
\infer
  {\hcomgeneric{\etc{r_i}}{\sbool} \steps M}
  {}
\qquad
\infer
  {\coe{x.\sbool}{r}{r'}{M} \steps M}
  {}
\]

A cubical type system \emph{has strict booleans} if $\veqper{\sbool}{\sbool}$ for
all $\Psi$, and $\vinper[-]{-}{-}{\sbool}$ is the least relation such that:
\[
  \vinper{\true}{\true}{\sbool}
  \text{ and }
  \vinper{\false}{\false}{\sbool}.
\]
In the remainder of this subsection, we consider cubical type systems that have
strict booleans.

\paragraph{Pretype}
$\cpretype{\sbool}$.

For all $\psi_1,\psi_2$,
$\td{\sbool}{\psi_1}\evals\sbool$,
$\td{\sbool}{\psi_2}\evals\sbool$,
$\td{\sbool}{\psi_1\psi_2}\evals\sbool$,
and $\veqper[\Psi_2]{\sbool}{\sbool}$.

\paragraph{Introduction}
$\coftype{\true}{\sbool}$ and $\coftype{\false}{\sbool}$.

For all $\msubsts{\Psi_1}{\psi_1}{\Psi}$
and $\msubsts{\Psi_2}{\psi_2}{\Psi_1}$,
$\td{\true}{\psi_1}\evals\true$,
$\td{\true}{\psi_2}\evals\true$,
$\td{\true}{\psi_1\psi_2}\evals\true$,
and $\vinper[\Psi_2]{\true}{\true}{\sbool}$.
The $\false$ case is analogous.

Unlike the ordinary booleans, all strict booleans (at all dimensions) are equal
to either $\true$ or $\false$. This fact justifies our treatment of $\hcomsym$
at $\sbool$, and the extensionality rule proven below.

\begin{lemma}\label{lem:sbool-ceqtm}
If $\coftype{M}{\sbool}$ then either
$\ceqtm{M}{\true}{\sbool}$ or
$\ceqtm{M}{\false}{\sbool}$.
\end{lemma}
\begin{proof}
Then $M\evals M_0$ where $\vinper{M_0}{M_0}{\sbool}$, so either $M_0=\true$ or
$M_0=\false$. In the former case, for any $\psitd$,
$\inper[\Psi']{\td{\true}{\psi}}{\td{M}{\psi}}{\sbool}$, so by
\cref{lem:coftype-ceqtm} and the introduction rule for strict booleans,
$\ceqtm{M}{\true}{\sbool}$. The $\false$ case is analogous.
\end{proof}

\paragraph{Elimination}
If $\ceqtm{M}{M'}{\sbool}$,
$\wftype{\oft{a}{\sbool}}{A}$,
$\ceqtm{T}{T'}{\subst{A}{\true}{a}}$, and
$\ceqtm{F}{F'}{\subst{A}{\false}{a}}$, then
$\ceqtm{\ifb{a.B}{M}{T}{F}}{\ifb{a.B'}{M'}{T'}{F'}}{\subst{A}{M}{a}}$.
(Note that the motive of the $\ifsym$ is ignored in the relevant portion of the
operational semantics, so it need not agree with the desired type.)

By \cref{lem:sbool-ceqtm}, either $\ceqtm{M}{\true}{\sbool}$ or
$\ceqtm{M}{\false}{\sbool}$. We consider the former case; the latter is
symmetric. Then for all $\psitd$, $\td{M}{\psi}\evals\true$ so
$\ifb{a.\td{B}{\psi}}{\td{M}{\psi}}{\td{T}{\psi}}{\td{F}{\psi}} \steps^*
\ifb{a.\td{B}{\psi}}{\true}{\td{T}{\psi}}{\td{F}{\psi}} \steps
\td{T}{\psi}$. Similarly, $\td{M'}{\psi}\evals\true$, and so by
\cref{lem:expansion} on both sides, it suffices to show
$\ceqtm{T}{T'}{\subst{A}{M}{a}}$. This follows from
$\ceqtype{\subst{A}{\true}{a}}{\subst{A}{M}{a}}$.

\paragraph{Extensionality}
If $\oftype{\oft{a}{\sbool}}{M}{A}$,
$\oftype{\oft{a}{\sbool}}{M'}{A}$,
$\ceqtm{\subst{M}{\true}{a}}{\subst{M'}{\true}{a}}{\subst{A}{\true}{a}}$,
$\ceqtm{\subst{M}{\false}{a}}{\subst{M'}{\false}{a}}{\subst{A}{\false}{a}}$, and
$\ceqtm{N}{N'}{\sbool}$,
then $\ceqtm{\subst{M}{N}{a}}{\subst{M'}{N'}{a}}{\subst{A}{N}{a}}$.

By \cref{lem:sbool-ceqtm}, either $\ceqtm{N}{\true}{\sbool}$
or $\ceqtm{N}{\false}{\sbool}$. In the former case, because
$\oftype{\oft{a}{\sbool}}{M}{A}$, we know that
$\ceqtm
{\subst{M}{\true}{a}}
{\subst{M}{N}{a}}
{\subst{A}{\true}{a}}$ and similarly for $M'$. By our assumption
$\ceqtm{\subst{M}{\true}{a}}{\subst{M'}{\true}{a}}{\subst{A}{\true}{a}}$, and
the fact that $\ceqpretype{\subst{A}{\true}{a}}{\subst{A}{N}{a}}$,
we conclude
$\ceqtm
{\subst{M}{N}{a}}
{\subst{M'}{N'}{a}}
{\subst{A}{N}{a}}$ as required. The $\false$ case is similar.

\paragraph{Computation}
If $\wftype{\oft{a}{\sbool}}{A}$,
$\coftype{T}{\subst{A}{\true}{a}}$, and
$\coftype{F}{\subst{A}{\false}{a}}$, then
$\ceqtm{\ifb{a.B}{\true}{T}{F}}{T}{\subst{A}{\true}{a}}$ and
$\ceqtm{\ifb{a.B}{\false}{T}{F}}{F}{\subst{A}{\false}{a}}$.

For all $\psi$,
$\ifb{a.\td{B}{\psi}}{\true}{\td{T}{\psi}}{\td{F}{\psi}} \steps
\td{T}{\psi}$, so the former follows by \cref{lem:expansion} and
$\coftype{T}{\subst{A}{\true}{a}}$. The latter case is analogous.

\paragraph{Kan}
$\cpretype{\sbool}$ is Kan.

The first Kan condition requires that for any $\Psi'$, if
\begin{enumerate}
\item $\ceqtm[\Psi']{M}{O}{\sbool}$,
\item $\ceqtmres[\Psi',y]{r_i=\e,r_j=\e'}{N^\e_i}{N^{\e'}_j}{\sbool}$
for any $i,j,\e,\e'$,
\item $\ceqtmres[\Psi',y]{r_i=\e}{N^\e_i}{P^\e_i}{\sbool}$
for any $i,\e$, and
\item $\ceqtmres[\Psi']{r_i=\e}{\dsubst{N^\e_i}{r}{y}}{M}{\sbool}$
for any $i,\e$,
\end{enumerate}
then $\ceqtm[\Psi']{\hcomgeneric{\etc{r_i}}{\sbool}}
{\hcom{\etc{r_i}}{\sbool}{r}{r'}{O}{\etc{y.P^\e_i}}}
{\sbool}$.

For all $\psitd$, $\td{\hcomsym}{\psi} \steps \td{M}{\psi}$, so by
\cref{lem:expansion} on both sides, it suffices to show
$\ceqtm[\Psi']{M}{O}{\sbool}$, which is the first assumption.

The second Kan condition requires that for any $\Psi'$, if
\begin{enumerate}
\item $\coftype[\Psi']{M}{\sbool}$,
\item $\ceqtmres[\Psi',y]{r_i=\e,r_j=\e'}{N^\e_i}{N^{\e'}_j}
{\sbool}$
for any $i,j,\e,\e'$, and
\item $\ceqtmres[\Psi']{r_i=\e}{\dsubst{N^\e_i}{r}{y}}{M}
{\sbool}$
for any $i,\e$,
\end{enumerate}
then $\ceqtm[\Psi']{\hcom{\etc{r_i}}{\sbool}
{r}{r}{M}{\etc{y.N^\e_i}}}
{M}{\sbool}$.

Again, by \cref{lem:expansion}, this follows from $\ceqtm[\Psi']{M}{M}{\sbool}$.

The third Kan condition requires that for any $\Psi'$,
if $r_i = \e$ for some $i$,
\begin{enumerate}
\item $\coftype[\Psi']{M}{\sbool}$,
\item $\ceqtmres[\Psi',y]{r_i=\e,r_j=\e'}{N^\e_i}{N^{\e'}_j}{\sbool}$
for any $i,j,\e,\e'$, and
\item $\ceqtmres[\Psi']{r_i=\e}{\dsubst{N^\e_i}{r}{y}}{M}{\sbool}$
for any $i,\e$,
\end{enumerate}
then $\ceqtm[\Psi']{\hcom{\etc{r_i}}{\sbool}{r}{r'}{M}{\etc{y.N^\e_i}}}
{\dsubst{N^\e_i}{r'}{y}}{\sbool}$.

The second hypothesis implies that
$\coftyperes[\Psi',y]{\e=\e}{N^\e_i}{\sbool}$, and hence that
$\coftype[\Psi',y]{N^\e_i}{\sbool}$. By \cref{lem:sbool-ceqtm},
$\ceqtm[\Psi',y]{N^\e_i}{\true}{\sbool}$ or
$\ceqtm[\Psi',y]{N^\e_i}{\false}{\sbool}$. Assume the former; then
$\ceqtm[\Psi']{\dsubst{N^\e_i}{r}{y}}{\true}{\sbool}$ and
$\ceqtm[\Psi']{\dsubst{N^\e_i}{r'}{y}}{\true}{\sbool}$, so by the third
hypothesis, $\ceqtm[\Psi']{M}{\true}{\sbool}$. The result follows by
\cref{lem:expansion}.

The fourth Kan condition asserts that for any $\Psi'$, if
$\ceqtm[\Psi']{M}{N}{\sbool}$, then
$\ceqtm[\Psi']{{\coe{x.\sbool}{r}{r'}{M}}}{{\coe{x.\sbool}{r}{r'}{N}}}{\sbool}$.
This follows immediately by \cref{lem:expansion}.

The fifth Kan condition requires that for any $\Psi'$, if
$\coftype[\Psi']{M}{\sbool}$, then
$\ceqtm[\Psi']{{\coe{x.\sbool}{r}{r}{M}}}{M}{\sbool}$.
This again follows immediately by \cref{lem:expansion}.

\paragraph{Cubical}
For any $\Psi'$ and $\vinper[\Psi']{M}{N}{\sbool}$, $\ceqtm[\Psi']{M}{N}{\sbool}$.

Then $M=N=\true$ or $M=N=\false$; in each case, this follows by the introduction
rule for strict booleans.

The above construction justifies adding the following rules to those found in
\cref{sec:proof-theory}:

\[
\infer{\cwftype{\sbool}}{}
\]

\[
\infer{\coftype{\true}{\sbool}}{}
\qquad
\infer{\coftype{\false}{\sbool}}{}
\]

\[
\infer
  {\ceqtm{\ifb{}{M}{T}{F}}{\ifb{}{M'}{T'}{F'}}{\subst{A}{M}{a}}}
  {\wftype{\oft{a}{\sbool}}{A} &
   \ceqtm{M}{M'}{\sbool} &
   \ceqtm{T}{T'}{\subst{A}{\true}{a}} &
   \ceqtm{F}{F'}{\subst{A}{\false}{a}}}
\]

\[
\infer
  {\ceqtm{\ifb{}{\true}{T}{F}}{T}{\subst{A}{\true}{a}}}
  {\wftype{\oft{a}{\sbool}}{A} &
   \coftype{T}{\subst{A}{\true}{a}} &
   \coftype{F}{\subst{A}{\false}{a}}}
\]

\[
\infer
  {\ceqtm{\ifb{}{\false}{T}{F}}{F}{\subst{A}{\false}{a}}}
  {\wftype{\oft{a}{\sbool}}{A} &
   \coftype{T}{\subst{A}{\true}{a}} &
   \coftype{F}{\subst{A}{\false}{a}}}
\]

\[
\infer
  {\ceqtm{\subst{M}{N}{a}}{\subst{M'}{N'}{a}}{\subst{A}{N}{a}}}
  {\begin{array}{ll}
   &\ceqtm{N}{N'}{\sbool} \\
   &\oftype{\oft{a}{\sbool}}{M}{A} \\
   &\oftype{\oft{a}{\sbool}}{M'}{A} \\
   &\ceqtm{\subst{M}{\true}{a}}{\subst{M'}{\true}{a}}{\subst{A}{\true}{a}} \\
   &\ceqtm{\subst{M}{\false}{a}}{\subst{M'}{\false}{a}}{\subst{A}{\false}{a}}
   \end{array}}
\]

\subsection{Fixed point constructions}
\label{ssec:fixed-point}

\Cref{sec:proof-theory} relies on the existence of a cubical type system with
booleans, the circle, all dependent function, pair, and identification types,
and the $\notb{x}$ type. Here we give an explicit fixed point construction of
the least such cubical type system.

Let $\VPER$ denote the collection of ternary relations $R(\Psi,M,N)$ over
$\wfval{M}$ and $\wfval{N}$, such that $R(\Psi,-,-)$ is symmetric and
transitive. Define the \emph{approximation ordering} $R \vperleq S$ by
\[
R \subseteq S \land
\forall \Psi,M. \left( R(\Psi,M,M) \implies
\left( R(\Psi,M,-) = S(\Psi,M,-) \right) \right)
\]
That is, $R\subseteq S$ and any element of both $R$ and $S$ has identical
equivalence classes in each.

\begin{lemma}
$(\VPER,\vperleq)$ is a complete partial order, i.e., $\vperleq$ is a partial
order with a least element and joins of all directed subsets.
\end{lemma}
\begin{proof}
The relation $\vperleq$ is reflexive, antisymmetric, and transitive because
$\subseteq$ is, and the condition on equivalence classes is reflexive and
transitive. The empty relation is $\vperleq$ all other $\VPER$s.

Assume $\{R_i\}_i$ is a directed set of $\VPER$s, i.e., a nonempty set such that
every pair of $R_i,R_j$ have an upper bound $R_k$ also in the set. Then we
construct the least upper bound as:
\[
(\bigsqcup_i R_i)(\Psi,M,N) \iff \exists i. R_i(\Psi,M,N)
\]
This is symmetric because each $R_i$ is. It is transitive because whenever
$R_i(\Psi,M,N)$ and $R_j(\Psi,N,O)$, they have an upper bound $R_k$ with
$R_k(\Psi,M,N)$, $R_k(\Psi,N,O)$, and thus $R_k(\Psi,M,O)$.

To see that it is an upper bound of $\{R_i\}_i$, first observe that by
definition, $R_i \subseteq \bigsqcup_i R_i$. Then, assume $R_i(\Psi,M,M)$ and
$(\bigsqcup_i R_i)(\Psi,M,N)$ (hence $R_j(\Psi,M,N)$ for some $j$). But
$R_i,R_j$ have an upper bound $R_k$ where
$R_k(\Psi,M,-)=R_j(\Psi,M,-)=R_i(\Psi,M,-)$, so in particular, $R_i(\Psi,M,N)$.
Finally, let $S$ be another upper bound of $\{R_i\}_i$, and show $\bigsqcup_i
R_i \vperleq S$. We have $\bigsqcup R_i \subseteq S$ because all $R_i \vperleq
S$ and so all $R_i \subseteq S$. Assume $(\bigsqcup_i R_i)(\Psi,M,M)$ (hence
$R_i(\Psi,M,N)$ for some $i$); because $R_i \vperleq \bigsqcup_i R_i$ and $R_i
\vperleq S$, we have $R_i(\Psi,M,-)=(\bigsqcup_i R_i)(\Psi,M,-)=S(\Psi,M,-)$.
\end{proof}

\begin{corollary}
Every monotone operator on $\VPER$ has a least fixed point.
\end{corollary}

A proof of this corollary can be found in \citet[8.22]{daveypriestleylattices}.
We proceed by using this fact to construct the least $\VPER$s described in the
sections on the booleans and the circle, which we will call
$\mathbb{B}(\Psi,V,V')$ and $\mathbb{C}(\Psi,V,V')$, respectively. (Then, to be
precise, a cubical type system \emph{has booleans} if $\veqper{\bool}{\bool}$
for all $\Psi$, and $\vinper{M}{N}{\bool} \iff \mathbb{B}(\Psi,M,N)$.)

\begin{definition}
For a $\VPER$ $R$, we define the derived relations:
\begin{enumerate}
\item $R^\cube(\Psi,M,N)$ iff
for all $\msubsts{\Psi_1}{\psi_1}{\Psi}$ and $\msubsts{\Psi_2}{\psi_2}{\Psi_1}$,
\begin{enumerate}
\item
$\td{M}{\psi_1}\evals M_1$,
$\td{M_1}{\psi_2}\evals M_2$,
$\td{M}{\psi_1\psi_2}\evals M_{12}$,
\item
$\td{N}{\psi_1}\evals N_1$,
$\td{N_1}{\psi_2}\evals N_2$,
$\td{N}{\psi_1\psi_2}\evals N_{12}$, and
\item
$R(\Psi_2,-,-)$ relates all of $M_2,M_{12},N_2,N_{12}$.
\end{enumerate}

\item $R^\cube(\Psi\mid\Xi,M,N)$ iff for any $\psitd$ satisfying $\Xi$,
$R^\cube(\Psi',\td{M}{\psi},\td{N}{\psi})$.
\end{enumerate}
\end{definition}

Then $\mathbb{B}$ is the least fixed point of the following monotone operator on
$\VPER$s:
\[\begin{aligned}
F(R) &= \{(\Psi,\true,\true),(\Psi,\false,\false)\} \\&\hspace{0.2em}\cup
\{(\Psi,\hcomgeneric{\etc{x_i}}{\bool},\hcom{\etc{x_i}}{\bool}{r}{r'}{O}{\etc{y.P^\e_i}})
\mid \\&\hspace{1.7em}
r\neq r' \land
R^\cube(\Psi,M,O) \land
R^\cube((\Psi,y \mid x_i=\e,x_j=\e'),N^\e_i,N^{\e'}_j) \land{} \\&\hspace{1.7em}
R^\cube((\Psi,y \mid x_i=\e),N^\e_i,P^\e_i) \land
R^\cube((\Psi \mid x_i=\e),\dsubst{N^\e_i}{r}{y},M) \}
\end{aligned}\]
This is monotone because whenever $R\vperleq S$, $F(R)\subseteq F(S)$; and for
any $(F(R))(\Psi,M,M)$, either $M=\true$ or $M=\false$ or $M=\hcomsym$, and in
each case we can show that $(F(R))(\Psi,M,-) = (F(S))(\Psi,M,-)$. One can also
check that this definition agrees with the earlier one. The construction of
$\mathbb{C}$ proceeds similarly.

Now, let $\CTS$ denote the collection of cubical type systems, which we here
denote $(E,\Phi)$ rather than $(\veqper[-]{}{},\vinper[-]{}{}-)$. We define an
approximation ordering $(E,\Phi) \ctsleq (E',\Phi')$ by
\[
E \vperleq E' \land
\forall \Psi,A. \left( E(\Psi,A,A) \implies
\left( \Phi(\Psi,A,-,-) = \Phi'(\Psi,A,-,-) \right) \right)
\]
That is, the elements of $E$ are elements of $E'$ with the same equivalence
classes, and are sent to the same PERs by $\Phi$ and $\Phi'$.

\begin{lemma}
$(\CTS,\ctsleq)$ is a complete partial order.
\end{lemma}
\begin{proof}
The relation $\ctsleq$ is a partial order because $E \vperleq E'$ is. The pair
of empty relations is $\ctsleq$ all other $\CTS$s. Given a directed set
$\{(E_i,\Phi_i)\}_i$ of $\CTS$s, the least upper bound is:
\[
\bigsqcup_i (E_i,\Phi_i) = \left(\bigsqcup_i E_i,
\left\{(\Psi,A,M,N)\mid
\exists i.(E_i(\Psi,A,A) \land \Phi_i(\Psi,A,M,N))\right\}\right)
\]
Because $\{(E_i,\Phi_i)\}_i$ is directed, for any $j,k,\Psi,A$ where
$E_j(\Psi,A,A)$ and $E_k(\Psi,A,A)$, we have $\Phi_j(\Psi,A,-,-) =
\Phi_k(\Psi,A,-,-)$, so the second component is equal to all of these. One can
verify that the above definition yields a least upper bound.
\end{proof}

Once again, every monotone operator on $\CTS$ therefore has a least fixed point.

For every closed or open judgment $\judg{\J}$ defined in \cref{sec:meanings}, we
write $\relcts{(E,\Phi)}{\judg{\J}}$ to explicitly annotate that this judgment
is meant relative to the cubical type system $(E,\Phi)$. We have defined
$(E,\Phi)\ctsleq (E',\Phi')$ in order to ensure that
$\relcts{(E,\Phi)}{\judg{\J}}$ implies $\relcts{(E',\Phi')}{\judg{\J}}$:

\begin{lemma}\label{lem:ctsleq-approx}~
Whenever $(E,\Phi)\ctsleq (E',\Phi')$,
\begin{enumerate}
\item if
$\relcts{(E,\Phi)}{\judg{\J}}$, then
$\relcts{(E',\Phi')}{\judg{\J}}$;

\item if
$\relcts{(E,\Phi)}{\cpretype{A}}$, then
$\relcts{(E,\Phi)}{\ceqpretype{A}{B}}$ iff
$\relcts{(E',\Phi')}{\ceqpretype{A}{B}}$;

\item if
$\relcts{(E,\Phi)}{\coftype{M}{A}}$, then
$\relcts{(E,\Phi)}{\ceqtm{M}{N}{A}}$ iff
$\relcts{(E',\Phi')}{\ceqtm{M}{N}{A}}$; and

\item if
$\relcts{(E,\Phi)}{\cwftype{A}}$, then
$\relcts{(E,\Phi)}{\ceqtype{A}{B}}$ iff
$\relcts{(E',\Phi')}{\ceqtype{A}{B}}$.
\end{enumerate}
\end{lemma}

Finally, we define the monotone operator $F(E,\Phi) = (E',\Phi')$ where
\[\begin{aligned}
E' &= \{(\Psi,\bool,\bool)\} \\&\hspace{0.2em}\cup
\{(\Psi,\C,\C)\} \\&\hspace{0.2em}\cup
\{(\Psi,\picl{a}{A}{B},\picl{a}{A'}{B'}) \mid \\&\hspace{1.7em}
\relcts{(E,\Phi)}{\ceqtype{A}{A'}} \land
\relcts{(E,\Phi)}{\eqtype{\oft{a}{A}}{B}{B'}} \} \\&\hspace{0.2em}\cup
\{(\Psi,\sigmacl{a}{A}{B},\sigmacl{a}{A'}{B'}) \mid \\&\hspace{1.7em}
\relcts{(E,\Phi)}{\ceqtype{A}{A'}} \land
\relcts{(E,\Phi)}{\eqtype{\oft{a}{A}}{B}{B'}} \} \\&\hspace{0.2em}\cup
\{(\Psi,\Id{x.A}{P_0}{P_1},\Id{x.A'}{P_0'}{P_1'}) \mid \\&\hspace{1.7em}
\relcts{(E,\Phi)}{\ceqtype[\Psi,x]{A}{A'}} \land
\relcts{(E,\Phi)}{\ceqtm{P_\e}{P_\e'}{\dsubst{A}{\e}{x}}} \} \\&\hspace{0.2em}\cup
\{((\Psi,x),\notb{x},\notb{x}) \mid
\relcts{(E,\Phi)}{\cwftype[\Psi,x]{\bool}} \}
\\
\Phi' &= \{(\Psi,\bool,V,V') \mid \mathbb{B}(\Psi,V,V')\} \\&\hspace{0.2em}\cup
\{(\Psi,\C,V,V') \mid \mathbb{C}(\Psi,V,V')\} \\&\hspace{0.2em}\cup
\{(\Psi,\picl{a}{A}{B},\lam{a}{M},\lam{a}{M'}) \mid
\relcts{(E,\Phi)}{\eqtm{\oft{a}{A}}{M}{M'}{B}}\} \\&\hspace{0.2em}\cup
\{(\Psi,\sigmacl{a}{A}{B},\pair{M}{N},\pair{M'}{N'}) \mid \\&\hspace{1.7em}
\relcts{(E,\Phi)}{\ceqtm{M}{M'}{A}} \land
\relcts{(E,\Phi)}{\ceqtm{N}{N'}{\subst{B}{M}{a}}}\} \\&\hspace{0.2em}\cup
\{(\Psi,\Id{x.A}{P_0}{P_1},\dlam{x}{M},\dlam{x}{M'}) \mid \\&\hspace{1.7em}
\relcts{(E,\Phi)}{\ceqtm[\Psi,x]{M}{M'}{A}} \land
\relcts{(E,\Phi)}{\ceqtm{\dsubst{M}{\e}{x}}{P_\e}{\dsubst{A}{\e}{x}}}
\} \\&\hspace{0.2em}\cup
\{((\Psi,x),\notb{x},\notel{x}{M},\notel{x}{M'}) \mid
\relcts{(E,\Phi)}{\ceqtm[\Psi,x]{M}{M'}{\bool}}\}
\end{aligned}\]

\begin{theorem}
$F$ is monotone, and any fixed point of $F$ has booleans, the circle, all
dependent function types, all dependent pair types, all identification types,
and the $\notb{x}$ type.
\end{theorem}
\begin{proof}
Monotonicity follows directly from \cref{lem:ctsleq-approx} by inspecting the
definition of $F$.

Let $(E,\Phi)$ be a fixed point of $F$. $(E,\Phi)$ has booleans and the circle
because every cubical type system in the image of $F$ clearly does. Given that
$(E,\Phi)$ has booleans, in order to show that it has the $\notb{x}$ type, we
must show that $E((\Psi,x),\notb{x},\notb{x})$ and $\Phi((\Psi,x),\notb{x},-,-)$
is the least relation relating $\notel{x}{M}$ and $\notel{x}{M'}$ when
$\relcts{(E,\Phi)}{\ceqtm[\Psi,x]{M}{M'}{\bool}}$. These follow by
$\relcts{(E,\Phi)}{\cwftype[\Psi,x]{\bool}}$, $(E,\Phi)=F(E,\Phi)$, and the
definition of $F$.

Assume
$\relcts{(E,\Phi)}{\ceqtype{A}{A'}}$ and
$\relcts{(E,\Phi)}{\eqtype{\oft{a}{A}}{B}{B'}}$; we want to show that $(E,\Phi)$
has their dependent function type. For all $\psitd$,
$\relcts{(E,\Phi)}{\ceqtype[\Psi']{\td{A}{\psi}}{\td{A'}{\psi}}}$ and
$\relcts{(E,\Phi)}
{\eqtype[\Psi']{\oft{a}{\td{A}{\psi}}}{\td{B}{\psi}}{\td{B'}{\psi}}}$. Because
$(E,\Phi)=F(E,\Phi)$, by the definition of $F$,
$E(\Psi',\picl{a}{\td{A}{\psi}}{\td{B}{\psi}},
\picl{a}{\td{A'}{\psi}}{\td{B'}{\psi}})$ and
$\Phi(\Psi',\picl{a}{\td{A}{\psi}}{\td{B}{\psi}},-,-)$ is the least relation
relating $\lam{a}{M}$ and $\lam{a}{M'}$ when
$\relcts{(E,\Phi)}{\eqtm[\Psi']{\oft{a}{\td{A}{\psi}}}{M}{M'}{\td{B}{\psi}}}$.
The argument for dependent pair types and identification types is similar.
\end{proof}

It follows that $F$ has a least fixed point $(E,\Phi)$ which validates the rules
we have given in \cref{sec:proof-theory}. However, by \cref{lem:ctsleq-approx},
any judgments that hold in $(E,\Phi)$ also hold in any larger cubical type
system $(E,\Phi)\ctsleq (E',\Phi')$.

\subsection{Additional lemmas}

Here we prove a number of general-purpose lemmas omitted in \cref{ssec:lemmas}.

\begin{lemma}\label{lem:ceqpretype-cubical}
If $\ceqpretype{A}{B}$ and $A$ is cubical, then $B$ is cubical.
\end{lemma}
\begin{proof}
For any $\psitd$ and $\vinper[\Psi']{M}{N}{B_0}$ for $\td{B}{\psi}\evals B_0$,
show that $\ceqtm[\Psi']{M}{N}{\td{B}{\psi}}$. But
$\veqper[\Psi']{\td{A}{\psi}\evals A_0}{B_0}$ so $\vinper[\Psi']{M}{N}{A_0}$; by
$A$ cubical we have
$\ceqtm[\Psi']{M}{N}{\td{A}{\psi}}$; and the result follows by
\cref{lem:ceqpretype-ceqtm}.
\end{proof}

The following lemmas are variations on head expansion (\cref{lem:expansion}).

\begin{lemma}\label{lem:coherent-expansion}
Assume we have $\wftm{M}$, $\cpretype{A}$, and a family of terms
$\{M^{\Psi'}_\psi\}_{\psitd}$ such that for all $\psitd$,
$\ceqtm[\Psi']
{M^{\Psi'}_{\psi}}
{\td{(M^{\Psi}_{\id})}{\psi}}
{\td{A}{\psi}}$ and
$\td{M}{\psi} \steps^* M^{\Psi'}_\psi$. Then
$\ceqtm{M}{M^\Psi_{\id}}{A}$.
\end{lemma}
\begin{proof}
For any $\msubsts{\Psi_1}{\psi_1}{\Psi}$ and $\msubsts{\Psi_2}{\psi_2}{\Psi_1}$,
we must show that
$\inperfour[\Psi_2]
{\td{M_1}{\psi_2}}
{\td{M}{\psi_1\psi_2}}
{\td{(M^\Psi_{\id})}{\psi_1\psi_2}}
{\td{M'_1}{\psi_2}}
{\td{A}{\psi_1\psi_2}}$
where $\td{M}{\psi_1}\evals M_1$ and $\td{(M^\Psi_{\id})}{\psi_1}\evals M'_1$.
\begin{enumerate}
\item $\inper[\Psi_2]{\td{M_1}{\psi_2}}{\td{M}{\psi_1\psi_2}}
{\td{A}{\psi_1\psi_2}}$.

We know $\td{M}{\psi_1} \steps^* M^{\Psi_1}_{\psi_1}$ and
$\coftype[\Psi_1]{M^{\Psi_1}_{\psi_1}}{\td{A}{\psi_1}}$, so
$M^{\Psi_1}_{\psi_1}\evals M_1$ and
$\inper[\Psi_2]{\td{M_1}{\psi_2}}{\td{(M^{\Psi_1}_{\psi_1})}{\psi_2}}
{\td{A}{\psi_1\psi_2}}$.
By $\ceqtm[\Psi_1]
{M^{\Psi_1}_{\psi_1}}
{\td{(M^{\Psi}_{\id})}{\psi_1}}
{\td{A}{\psi_1}}$ under $\psi_2$,
$\ceqtm[\Psi_2]
{\td{(M^{\Psi}_{\id})}{\psi_1\psi_2}}
{M^{\Psi_2}_{\psi_1\psi_2}}
{\td{A}{\psi_1\psi_2}}$, and transitivity, we have
$\ceqtm[\Psi_2]
{\td{(M^{\Psi_1}_{\psi_1})}{\psi_2}}
{M^{\Psi_2}_{\psi_1\psi_2}}
{\td{A}{\psi_1\psi_2}}$ and thus
$\inper[\Psi_2]{\td{(M^{\Psi_1}_{\psi_1})}{\psi_2}}{M^{\Psi_2}_{\psi_1\psi_2}}
{\td{A}{\psi_1\psi_2}}$. The result
$\inper[\Psi_2]
{\td{M}{\psi_1\psi_2} \steps^* M^{\Psi_2}_{\psi_1\psi_2}}
{\td{M_1}{\psi_2}}
{\td{A}{\psi_1\psi_2}}$ follows.

\item $\inper[\Psi_2]{\td{M}{\psi_1\psi_2}}{\td{(M^\Psi_{\id})}{\psi_1\psi_2}}
{\td{A}{\psi_1\psi_2}}$.

Follows from
$\inper[\Psi_2]
{\td{M}{\psi_1\psi_2} \steps^* M^{\Psi_2}_{\psi_1\psi_2}}
{\td{(M^\Psi_{\id})}{\psi_1\psi_2}}
{\td{A}{\psi_1\psi_2}}$, by
$\ceqtm[\Psi_2]
{M^{\Psi_2}_{\psi_1\psi_2}}
{\td{(M^{\Psi}_{\id})}{\psi_1\psi_2}}
{\td{A}{\psi_1\psi_2}}$.

\item $\inper[\Psi_2]{\td{(M^\Psi_{\id})}{\psi_1\psi_2}}{\td{M'_1}{\psi_2}}
{\td{A}{\psi_1\psi_2}}$.

Follows from $\coftype{M^{\Psi}_{\id}}{A}$.
\qedhere
\end{enumerate}
\end{proof}

\Cref{lem:expansion} is a special case of the above, in which $M^{\Psi'}_\psi =
\td{M'}{\psi}$.

\begin{lemma}\label{lem:pretype-expansion}
If $\ceqpretype{A'}{B}$ and for all $\psitd$,
$\td{A}{\psi} \steps^* \td{A'}{\psi}$, then $\ceqpretype{A}{B}$.
\end{lemma}
\begin{proof}
For any $\msubsts{\Psi_1}{\psi_1}{\Psi}$ and $\msubsts{\Psi_2}{\psi_2}{\Psi_1}$,
we know $\eqperfour[\Psi_2]
{\td{A'}{\psi_1\psi_2}}{\td{A'_1}{\psi_2}}
{\td{B}{\psi_1\psi_2}}{\td{B_1}{\psi_2}}$
where $\td{A'}{\psi_1}\evals A_1$ and $\td{B}{\psi_1}\evals B_1$.
It remains to show $\eqper[\Psi_2]{\td{A'}{\psi_1\psi_2}}{\td{A}{\psi_1\psi_2}}$
and $\eqper[\Psi_2]{\td{A'_1}{\psi_2}}{\td{A_1}{\psi_2}}$ where
$\td{A}{\psi_1}\evals A_1$.
The former holds by $\eqper[\Psi_2]
{\td{A}{\psi_1\psi_2} \steps^* \td{A'}{\psi_1\psi_2}}
{\td{A'}{\psi_1\psi_2}}$.
The latter holds because $\td{A}{\psi_1}\steps^*\td{A'}{\psi_1}\evals A'_1$;
thus $A_1 = A'_1$, and the result follows by
$\eqper[\Psi_2]{\td{A'_1}{\psi_2}}{\td{A'_1}{\psi_2}}$.
\end{proof}

\begin{lemma}\label{lem:type-expansion}
If $\ceqtype{A'}{B}$ and for all $\psitd$,
$\td{A}{\psi} \steps^* \td{A'}{\psi}$, then $\ceqtype{A}{B}$.
\end{lemma}
\begin{proof}
By \cref{lem:pretype-expansion}, $\ceqpretype{A}{A'}$ and $\ceqpretype{A}{B}$.
By the former and \cref{lem:ceqpretype-cubical}, $A$ is cubical. Thus it
suffices to show that $A$ and $B$ are equally Kan. We outline the arguments for
the $\hcomsym$ and $\coesym$ conditions separately.

In the first Kan condition, we are given elements
$\coftype[\Psi']{-}{\td{A}{\psi}}$ and must show
$\ceqtm[\Psi']{\hcomsym_{\td{A}{\psi}}}{\hcomsym_{\td{B}{\psi}}}{\td{A}{\psi}}$.
The given elements are $\coftype[\Psi']{-}{\td{A'}{\psi}}$ also, so by
$\ceqtype{A'}{B}$,
$\ceqtm[\Psi']
{\hcomsym_{\td{A'}{\psi}}}
{\hcomsym_{\td{B}{\psi}}}{\td{A}{\psi}}$. The result follows by
\cref{lem:expansion} on the left, since for any
$\msubsts{\Psi''}{\psi'}{\Psi'}$, $\hcomsym_{\td{A}{\psi\psi'}}\steps^*
\hcomsym_{\td{A'}{\psi\psi'}}$. The second and third Kan conditions follow
similarly.

In the fourth Kan condition, we are given elements
$\coftype[\Psi']{-}{\dsubst{\td{A}{\psi}}{r}{x}}$ and must show
$\ceqtm[\Psi']{\coesym_{x.\td{A}{\psi}}}{\coesym_{x.\td{B}{\psi}}}
{\dsubst{\td{A}{\psi}}{r'}{x}}$. The given elements are
$\coftype[\Psi']{-}{\dsubst{\td{A'}{\psi}}{r}{x}}$ also, so by
$\ceqtype{A'}{B}$,
$\ceqtm[\Psi']{\coesym_{x.\td{A'}{\psi}}}{\coesym_{x.\td{B}{\psi}}}
{\dsubst{\td{A}{\psi}}{r'}{x}}$. The result follows by
\cref{lem:expansion} on the left, since for any
$\msubsts{\Psi''}{\psi'}{\Psi'}$, $\coesym_{x.\td{A}{\psi\psi'}}\steps^*
\coesym_{x.\td{A'}{\psi\psi'}}$. The fifth Kan condition follows similarly.
\end{proof}

Next, we prove a strengthened form of \cref{lem:coftype-ceqtm} for cubical
pretypes.

\begin{lemma}\label{lem:coftype-ceqtm-cubical}
If $\cpretype{A}$ is cubical, $\coftype{M}{A}$, $\coftype{N}{A}$, and
$\inper{M}{N}{A}$, then $\ceqtm{M}{N}{A}$.
\end{lemma}
\begin{proof}
By \cref{lem:coftype-ceqtm}, it suffices to show that for any $\psitd$,
$\inper[\Psi']{\td{M}{\psi}}{\td{N}{\psi}}{\td{A}{\psi}}$.
By $\coftype{M}{A}$, we know that
$\inper[\Psi']{\td{M}{\psi}}{\td{M_0}{\psi}}{\td{A}{\psi}}$ where $M\evals M_0$,
and similarly for $N$.
By $\cpretype{A}$ cubical and $\vinper{M_0}{N_0}{A_0}$, we know that
$\ceqtm{M_0}{N_0}{A}$, and thus,
$\inper[\Psi']{\td{M_0}{\psi}}{\td{N_0}{\psi}}{\td{A}{\psi}}$.
The result follows by transitivity.
\end{proof}

\begin{lemma}\label{lem:coftype-evals-ceqtm}
If $\cpretype{A}$ is cubical and $\coftype{M}{A}$, then
$M\evals V$ and $\ceqtm{M}{V}{A}$.
\end{lemma}
\begin{proof}
By $\coftype{M}{A}$, $M\evals V$ and $\vinper{V}{V}{A_0}$ where $A\evals A_0$.
Because $\cpretype{A}$ is cubical, $\coftype{V}{A}$. The result follows by
\cref{lem:coftype-ceqtm-cubical}.
\end{proof}

\subsection{Univalence for isomorphisms}

Extend the term syntax by $\ia{r}{A,B,F,G}$, $\iain{r}{M,F}$, and
$\iaout{r}{M,G}$, and the operational semantics by

\[
\infer
  {\isval{\ia{x}{A,B,F,G}}}
  {}
\qquad
\infer
  {\ia{0}{A,B,F,G} \steps A}
  {}
\qquad
\infer
  {\ia{1}{A,B,F,G} \steps B}
  {}
\]

\[
\infer
  {\isval{\iain{x}{M,F}}}
  {}
\qquad
\infer
  {\iain{0}{M,F} \steps M}
  {}
\qquad
\infer
  {\iain{1}{M,F} \steps \app{F}{M}}
  {}
\]

\[
\infer
  {\iaout{0}{M,G} \steps M}
  {}
\qquad
\infer
  {\iaout{1}{M,G} \steps \app{G}{M}}
  {}
\]

\[
\infer
  {\iaout{x}{M,G} \steps \iaout{x}{M',G}}
  {M \steps M'}
\qquad
\infer
  {\iaout{x}{\iain{x}{M,F},G} \steps M}
  {}
\]

\[
\infer
  {\coe{x.\ia{x}{A,B,F,G}}{r}{r'}{M}
   \steps
   \iain{r'}
        {\coe{x.A}{r}{r'}{\iaout{r}{M,\dsubst{G}{r}{x}}},\dsubst{F}{r'}{x}}}
  {}
\]

\[
\infer
  {\coe{x.\ia{w}{A,B,F,G}}{r}{r'}{M}
   \steps
   \iain{w}
        {C,\dsubst{F}{r'}{x}}}
  {w \neq x &
   C = \com{w}{x.A}{r}{r'}
           {\iaout{w}{M,\dsubst{G}{r}{x}}}
           {y.\coe{x.A}{r}{y}{M},y.\app{\dsubst{G}{y}{x}}{\coe{x.B}{r}{y}{M}}}}
\]

\[
\infer
  {\hcomgeneric{\etc{r_i}}{\ia{x}{A,B,F,G}}
   \steps
   \iain{x}{\hcom{\etc{r_i},x}{A}{r}{r'}{\iaout{x}{M,G}}{\etc{T}},F}}
  {\etc{T} = \etc{y.\iaout{x}{N^\e_i,G}},
             z.\hcom{\etc{r_i}}{A}{r}{z}{M}{\etc{y.N^\e_i}},
             z.\app{G}{\hcom{\etc{r_i}}{B}{r}{z}{M}{\etc{y.N^\e_i}}}}
\]

If a cubical type system has dependent function types,
\begin{enumerate}
\item $\ceqtype{A}{A'}$,
\item $\ceqtyperes{r=1}{B}{B'}$,
\item $\ceqtmres{r=1}{F}{F'}{\arr{A}{B}}$,
\item $\ceqtmres{r=1}{G}{G'}{\arr{B}{A}}$,
\item $\eqtmres{r=1}{\oft{a}{A}}{\app{G}{\app{F}{a}}}{a}{A}$, and
\item $\eqtmres{r=1}{\oft{b}{B}}{\app{F}{\app{G}{b}}}{b}{B}$,
\end{enumerate}
we say it has their \emph{isomorphism-univalence type} when for all
$\psitd$ such that $\td{r}{\psi} = x$,
$\veqper[\Psi']{\ia{x}{\td{A}{\psi},\td{B}{\psi},\td{F}{\psi},\td{G}{\psi}}}
{\ia{x}{\td{A'}{\psi},\td{B'}{\psi},\td{F'}{\psi},\td{G'}{\psi}}}$, and
$\vinper[\Psi']{-}{-}{\ia{x}{\td{A}{\psi},\td{B}{\psi},\td{F}{\psi},\td{G}{\psi}}}$
is the least relation such that
\[
\vinper[\Psi']
{\iain{x}{M,F''}}
{\iain{x}{M',F'''}}
{\ia{x}{\td{A}{\psi},\td{B}{\psi},\td{F}{\psi},\td{G}{\psi}}}
\]
when $\ceqtm[\Psi']{M}{M'}{\td{A}{\psi}}$,
$\ceqtmres[\Psi']{x=1}{\td{F}{\psi}}{F''}{\arr{\td{A}{\psi}}{\td{B}{\psi}}}$, and
$\ceqtmres[\Psi']{x=1}{\td{F}{\psi}}{F'''}{\arr{\td{A}{\psi}}{\td{B}{\psi}}}$.

We will abbreviate the six premises above as $\Isoeq{r}{A,B,F,G}{A',B',F',G'}$,
and abbreviate the unary version $\Iso{r}{A,B,F,G}$. Note that if $r=0,1$ the
condition of having the isomorphism-univalence type is trivial. In the remainder
of this subsection, we consider cubical type systems that have dependent
function types and the isomorphism-univalence type of
$\Isoeq{r}{A,B,F,G}{A',B',F',G'}$.

Isomorphism-univalence types generalize $\notb{x}$, which constructs an $x$-line
between $\bool$ and $\bool$ corresponding to the isomorphism $\notf{-}$. In
contrast, $\ia{x}{A,B,F,G}$ constructs an $x$-line between $\dsubst{A}{0}{x}$
and $\dsubst{B}{1}{x}$, given an $x$-line $A$ and any (strict) isomorphism
between $\dsubst{A}{1}{x}$ and $\dsubst{B}{1}{x}$.
\[
\xymatrix@=1em{
  {} \ar[r] & x
}
\qquad
\xymatrix@C=8em@R=3em{
  \dsubst{A}{0}{x} \ar[rd]|{\ia{x}{A,B,F,G}\quad} \ar[r]^{A} &
  \dsubst{A}{1}{x} \ar@/^/[d]^{\dsubst{F}{1}{x}} \\
  & \dsubst{B}{1}{x} \ar@/^/[u]^{\dsubst{G}{1}{x}}} \\
\]
Coercion from $0$ to $1$ is implemented by coercing in $A$, then applying the
forward direction of the isomorphism. Analogously to $\notb{x}$, the canonical
elements of $\ia{x}{A,B,F,G}$ are $\iain{x}{M,F}$ where $\coftype{M}{A}$;
computing the right face of such an element applies the isomorphism, but
computing the left face does not. (The opposite is the case with $\notb{x}$.)
Modulo implementation details, however, one may simply replace the type
$\notb{x}$ with $\ia{x}{\bool,\bool,\lam{a}{\notf{a}},\lam{a}{\notf{a}}}$.

A crucial difference between these constructions is that $A$ may have
non-trivial coercion, unlike $\bool$, making it impossible to extract $M$ from
$\iain{x}{M,F}$ by applying $\coesym$ as in \cref{lem:ceqtm-notel-coe}. Instead,
we provide $\iaout{x}{M,G}$ for this explicit purpose.

\paragraph{Pretype}
$\ceqpretype{\ia{r}{A,B,F,G}}{\ia{r}{A',B',F',G'}}$.
If $r=0$ then $\ceqpretype{\ia{0}{A,B,F,G}}{A}$, and
if $r=1$ then $\ceqpretype{\ia{1}{A,B,F,G}}{B}$.

Here and in future proofs, we will abbreviate the arguments to $\ia{r}{-}$ by
$I$ when clear. We focus on the unary case. For any
$\msubsts{\Psi_1}{\psi_1}{\Psi}$ and $\msubsts{\Psi_2}{\psi_2}{\Psi_1}$,
\begin{enumerate}
\item If $\td{r}{\psi_1} = 0$ then
by $\cwftype{A}$, we have
$\ia{0}{\td{I}{\psi_1}} \steps \td{A}{\psi_1} \evals A_1$ and
$\eqper[\Psi_2]{\td{A_1}{\psi_2}}{\td{A}{\psi_1\psi_2}}$.

\item If $\td{r}{\psi_1} = 1$ then
by $\cwftype[\Psi_1]{\td{B}{\psi_1}}$, we have
$\ia{1}{\td{I}{\psi_1}} \steps \td{B}{\psi_1} \evals B_1$ and
$\eqper[\Psi_2]{\td{B_1}{\psi_2}}{\td{B}{\psi_1\psi_2}}$.

\item If $\td{r}{\psi_1} = x$ and $\td{x}{\psi_2} = 0$ then
$\isval{\ia{x}{\td{I}{\psi_1}}}$, and by $\cwftype{A}$ we have
$\ia{0}{\td{I}{\psi_1\psi_2}}\steps \td{A}{\psi_1\psi_2}$,
$\ia{\td{r}{\psi_1\psi_2}}{\td{I}{\psi_1\psi_2}} \steps \td{A}{\psi_1\psi_2}$,
and $\eqper[\Psi_2]{\td{A}{\psi_1\psi_2}}{\td{A}{\psi_1\psi_2}}$.

\item If $\td{r}{\psi_1} = x$ and $\td{x}{\psi_2} = 1$ then
$\isval{\ia{x}{\td{I}{\psi_1}}}$, and by
$\cwftype[\Psi_2]{\td{B}{\psi_1\psi_2}}$ we have
$\ia{1}{\td{I}{\psi_1\psi_2}}\steps\td{B}{\psi_1\psi_2}$,
$\ia{\td{r}{\psi_1\psi_2}}{\td{I}{\psi_1\psi_2}}\steps\td{B}{\psi_1\psi_2}$, and
$\eqper[\Psi_2]{\td{B}{\psi_1\psi_2}}{\td{B}{\psi_1\psi_2}}$.

\item If $\td{r}{\psi_1} = x$ and $\td{x}{\psi_2} = x'$ then
$\isval{\ia{x}{\td{I}{\psi_1}}}$, $\isval{\ia{x'}{\td{I}{\psi_1\psi_2}}}$, and
by assumption,
$\veqper[\Psi_2]{\ia{x'}{\td{I}{\psi_1\psi_2}}}{\ia{x'}{\td{I}{\psi_1\psi_2}}}$.
\end{enumerate}
The equations hold because for any $\psitd$,
$\ia{0}{\td{I}{\psi}}\steps \td{A}{\psi}$ where $\cwftype{A}$, and
$\ia{1}{\td{I}{\psi}}\steps \td{B}{\psi}$ where $\cwftyperes{1=1}{B}$.

\paragraph{Introduction}
If $\ceqtm{M}{M'}{A}$ then
$\ceqtm{\iain{r}{M,F}}{\iain{r}{M',F'}}{\ia{r}{A,B,F,G}}$.
If $r=0$ then $\ceqtm{\iain{0}{M,F}}{M}{A}$, and if $r=1$ then
$\ceqtm{\iain{1}{M,F}}{\app{F}{M}}{B}$.

For any $\msubsts{\Psi_1}{\psi_1}{\Psi}$ and $\msubsts{\Psi_2}{\psi_2}{\Psi_1}$,
\begin{enumerate}
\item If $\td{r}{\psi_1} = 0$ then
by $\ceqtm{M}{M'}{A}$, we know that
$\iain{0}{\td{M}{\psi_1},\td{F}{\psi_1}}\steps\td{M}{\psi_1}\evals M_1$,
$\iain{0}{\td{M}{\psi_1\psi_2},\td{F}{\psi_1\psi_2}}\steps\td{M}{\psi_1\psi_2}$,
and
$\inper[\Psi_2]{\td{M_1}{\psi_2}}{\td{M}{\psi_1\psi_2}}{\td{A}{\psi_1\psi_2}}$
(and similarly for $\iain{r}{M',F'}$).

\item If $\td{r}{\psi_1} = 1$ then
$\ceqtm[\Psi_1]{\app{\td{F}{\psi_1}}{\td{M}{\psi_1}}}
{\app{\td{F'}{\psi_1}}{\td{M'}{\psi_1}}}{\td{B}{\psi_1}}$ by the elimination
rule for dependent functions, so we know that
$\iain{1}{\td{M}{\psi_1},\td{F}{\psi_1}}\steps
\app{\td{F}{\psi_1}}{\td{M}{\psi_1}}\evals N_1$ and
$\inper[\Psi_2]
{\iain{1}{\td{M}{\psi_1\psi_2},\td{F}{\psi_1\psi_2}}\steps
\app{\td{F}{\psi_1\psi_2}}{\td{M}{\psi_1\psi_2}}}
{\td{N_1}{\psi_2}}{\td{B}{\psi_1\psi_2}}$
(and similarly for $\iain{r}{M',F'}$).

\item If $\td{r}{\psi_1} = x$ and $\td{x}{\psi_2} = 0$ then
$\isval{\iain{x}{\td{M}{\psi_1},\td{F}{\psi_1}}}$, and by
$\ceqtm{M}{M'}{A}$ we have
$\inper[\Psi_2]
{\iain{0}{\td{M}{\psi_1\psi_2},\td{F}{\psi_1\psi_2}} \steps
\td{M}{\psi_1\psi_2}}
{\td{M}{\psi_1\psi_2}}{\td{A}{\psi_1\psi_2}}$
(and similarly for $\iain{r}{M',F'}$).

\item If $\td{r}{\psi_1} = x$ and $\td{x}{\psi_2} = 1$ then
$\isval{\iain{x}{\td{M}{\psi_1},\td{F}{\psi_1}}}$, and by
\[ \ceqtm[\Psi_2]{\app{\td{F}{\psi_1\psi_2}}{\td{M}{\psi_1\psi_2}}}
{\app{\td{F'}{\psi_1\psi_2}}{\td{M'}{\psi_1\psi_2}}}{\td{B}{\psi_1\psi_2}} \]
we have
$\inper[\Psi_2]
{\iain{1}{\td{M}{\psi_1\psi_2},\td{F}{\psi_1\psi_2}}\steps
\app{\td{F}{\psi_1\psi_2}}{\td{M}{\psi_1\psi_2}}}
{\app{\td{F}{\psi_1\psi_2}}{\td{M}{\psi_1\psi_2}}}
{\td{B}{\psi_1\psi_2}}$
(and similarly for $\iain{r}{M',F'}$).

\item If $\td{r}{\psi_1} = x$ and $\td{x}{\psi_2} = x'$ then
$\isval{\iain{x}{\td{M}{\psi_1},\td{F}{\psi_1}}}$ and
$\isval{\iain{x'}{\td{M}{\psi_1\psi_2},\td{F}{\psi_1\psi_2}}}$. To see
\[ \vinper[\Psi_2]
{\iain{x'}{\td{M}{\psi_1\psi_2},\td{F}{\psi_1\psi_2}}}
{\iain{x'}{\td{M'}{\psi_1\psi_2},\td{F'}{\psi_1\psi_2}}}
{\ia{x'}{\td{I}{\psi_1\psi_2}}} \]
notice that
$\ceqtm[\Psi_2]{\td{M}{\psi_1\psi_2}}{\td{M'}{\psi_1\psi_2}}
{\td{A}{\psi_1\psi_2}}$
and by \cref{lem:td-judgres}, we have
$\ceqtmres[\Psi_2]{x'=1}{\td{F}{\psi_1\psi_2}}{\td{F'}{\psi_1\psi_2}}
{\arr{\td{A}{\psi_1\psi_2}}{\td{B}{\psi_1\psi_2}}}$.
\end{enumerate}
The first equation holds by \cref{lem:expansion} and our hypothesis; the second
holds by these and the elimination rule for dependent functions.

\paragraph{Cubical}
For any $\psitd$ and $\vinper[\Psi']{M}{N}{I_0}$ where
$\ia{\td{r}{\psi}}{\td{I}{\psi}}\evals I_0$,
$\ceqtm[\Psi']{M}{N}{\ia{\td{r}{\psi}}{\td{I}{\psi}}}$.

If $\td{r}{\psi} = 0$ then this follows from
$\ceqpretype[\Psi']{\ia{\td{r}{\psi}}{\td{I}{\psi}}}{\td{A}{\psi}}$ and the fact
that $\cpretype{A}$ is cubical; and similarly when $\td{r}{\psi} = 1$.
If $\td{r}{\psi} = x$ then this follows from the introduction rule.

\paragraph{Elimination}
If $\ceqtm{M}{M'}{\ia{r}{A,B,F,G}}$ then
$\ceqtm{\iaout{r}{M,G}}{\iaout{r}{M',G'}}{A}$.
If $r=0$ then $\ceqtm{\iaout{0}{M,G}}{M}{A}$, and if $r=1$ then
$\ceqtm{\iaout{1}{M,G}}{\app{G}{M}}{A}$.

We focus on the unary case. For any $\msubsts{\Psi_1}{\psi_1}{\Psi}$ and
$\msubsts{\Psi_2}{\psi_2}{\Psi_1}$,
\begin{enumerate}
\item If $\td{r}{\psi_1} = 0$ then
by $\coftype[\Psi_1]{\td{M}{\psi_1}}{\ia{0}{\td{I}{\psi_1}}}$ and
$\ceqpretype[\Psi_1]{\ia{0}{\td{I}{\psi_1}}}{\td{A}{\psi_1}}$,
$\iaout{0}{\td{M}{\psi_1},\td{G}{\psi_1}}\steps\td{M}{\psi_1}\evals M_1$ and
$\inper[\Psi_2]
{\iaout{0}{\td{M}{\psi_1\psi_2},\td{G}{\psi_1\psi_2}}\steps\td{M}{\psi_1\psi_2}}
{\td{M_1}{\psi_2}}
{\td{A}{\psi_1\psi_2}}$.

\item If $\td{r}{\psi_1} = 1$ then by
$\coftype[\Psi_1]{\td{M}{\psi_1}}{\ia{1}{\td{I}{\psi_1}}}$,
$\ceqpretype[\Psi_1]{\ia{1}{\td{I}{\psi_1}}}{\td{B}{\psi_1}}$,
$\coftype[\Psi_1]{\td{G}{\psi_1}}{\arr{\td{B}{\psi_1}}{\td{A}{\psi_1}}}$, and
the elimination rule for dependent functions, we have
$\coftype[\Psi_1]{\app{\td{G}{\psi_1}}{\td{M}{\psi_1}}}{\td{A}{\psi_1}}$. Thus
$\iaout{1}{\td{M}{\psi_1},\td{G}{\psi_1}}\steps
\app{\td{G}{\psi_1}}{\td{M}{\psi_1}}\evals N_1$ and
\[ \inper[\Psi_2]
{\iaout{1}{\td{M}{\psi_1\psi_2},\td{G}{\psi_1\psi_2}}\steps
\app{\td{G}{\psi_1\psi_2}}{\td{M}{\psi_1\psi_2}}}
{\td{N_1}{\psi_2}}
{\td{A}{\psi_1\psi_2}}. \]

\item If $\td{r}{\psi_1} = x$ then
$\coftype[\Psi_1]{\td{M}{\psi_1}}{\ia{x}{\td{I}{\psi_1}}}$, so
$\td{M}{\psi_1}\evals\iain{x}{N,H}$ where
$\coftype[\Psi_1]{N}{\td{A}{\psi_1}}$ and
$\ceqtmres[\Psi_1]{x=1}
{\td{F}{\psi_1}}{H}{\arr{\td{A}{\psi_1}}{\td{B}{\psi_1}}}$. Thus
$\iaout{x}{\td{M}{\psi_1},\td{G}{\psi_1}}\steps^*
\iaout{x}{\iain{x}{N,H},\td{G}{\psi_1}}\steps N \evals N_0$.

\begin{enumerate}
\item If $\td{x}{\psi_2} = 0$ then
$\iaout{0}{\td{M}{\psi_1\psi_2},\td{G}{\psi_1\psi_2}}\steps
\td{M}{\psi_1\psi_2}$ and we must show
$\inper[\Psi_2]{\td{M}{\psi_1\psi_2}}{\td{N_0}{\psi_2}}
{\td{A}{\psi_1\psi_2}}$. This follows from
$\inper[\Psi_2]
{\iain{0}{\td{N}{\psi_2},\td{H}{\psi_2}}\steps\td{N}{\psi_2}}
{\td{M}{\psi_1\psi_2}}
{\ia{0}{\td{I}{\psi_1\psi_2}}}$,
$\ia{0}{\td{I}{\psi_1\psi_2}}\steps \td{A}{\psi_1\psi_2}$, and
$\inper[\Psi_2]{\td{N}{\psi_2}}{\td{N_0}{\psi_2}}{\td{A}{\psi_1\psi_2}}$.

\item If $\td{x}{\psi_2} = 1$ then
$\iaout{1}{\td{M}{\psi_1\psi_2},\td{G}{\psi_1\psi_2}}\steps
\app{\td{G}{\psi_1\psi_2}}{\td{M}{\psi_1\psi_2}}$ and we must show
$\inper[\Psi_2]{\app{\td{G}{\psi_1\psi_2}}{\td{M}{\psi_1\psi_2}}}
{\td{N_0}{\psi_2}}{\td{A}{\psi_1\psi_2}}$. By
$\inper[\Psi_2]{\td{N}{\psi_2}}{\td{N_0}{\psi_2}}{\td{A}{\psi_1\psi_2}}$,
we instead show that
$\inper[\Psi_2]{\app{\td{G}{\psi_1\psi_2}}{\td{M}{\psi_1\psi_2}}}
{\td{N}{\psi_2}}{\td{A}{\psi_1\psi_2}}$.

For all $\msubsts{\Psi'}{\psi}{\Psi_2}$,
$\inper[\Psi']
{\td{M}{\psi_1\psi_2\psi}}
{\iain{1}{\td{N}{\psi_2\psi},\td{H}{\psi_2\psi}}\steps
\app{\td{H}{\psi_2\psi}}{\td{N}{\psi_2\psi}}}
{\td{B}{\psi_1\psi_2\psi}}$, so by \cref{lem:coftype-ceqtm},
$\coftype[\Psi_2]{\td{M}{\psi_1\psi_2}}{\td{B}{\psi_1\psi_2}}$, and
$\coftype[\Psi_2]{\app{\td{H}{\psi_2}}{\td{N}{\psi_2}}}
{\td{B}{\psi_1\psi_2}}$, we have
\[
\ceqtm[\Psi_2]
{\td{M}{\psi_1\psi_2}}
{\app{\td{H}{\psi_2}}{\td{N}{\psi_2}}}
{\td{B}{\psi_1\psi_2}}.
\]
By the elimination rule for dependent functions and
$\ceqtm[\Psi_2]{\td{F}{\psi_1\psi_2}}{H\psi_2}
{\arr{\td{A}{\psi_1\psi_2}}{\td{B}{\psi_1\psi_2}}}$,
\[
\ceqtm[\Psi_2]
{\td{M}{\psi_1\psi_2}}
{\app{\td{F}{\psi_1\psi_2}}{\td{N}{\psi_2}}}
{\td{B}{\psi_1\psi_2}},
\]
and by the elimination rule for dependent functions and the hypothesis that
$\eqtm[\Psi_2]{\oft{a}{\td{A}{\psi_1\psi_2}}}
{\app{\td{G}{\psi_1\psi_2}}{\app{\td{F}{\psi_1\psi_2}}{a}}}{a}
{\td{A}{\psi_1\psi_2}}$,
\[
\ceqtm[\Psi_2]
{\app{\td{G}{\psi_1\psi_2}}{\td{M}{\psi_1\psi_2}}}
{\td{N}{\psi_2}}
{\td{A}{\psi_1\psi_2}}.
\]
The result follows immediately.

\item If $\td{x}{\psi_2} = x'$ then
$\inper[\Psi_2]
{\td{M}{\psi_1\psi_2}}
{\iain{x'}{\td{N}{\psi_2},\td{H}{\psi_2}}}
{\td{A}{\psi_1\psi_2}}$, so
$\td{M}{\psi_1\psi_2}\evals\iain{x'}{N',H'}$ and
$\ceqtm[\Psi_2]{N'}{\td{N}{\psi_2}}{\td{A}{\psi_1\psi_2}}$. Thus
$\iaout{x'}{\td{M}{\psi_1\psi_2},\td{G}{\psi_1\psi_2}}\steps^*
\iaout{x'}{\iain{x'}{N',H'},\td{G}{\psi_1\psi_2}}\steps N'$, and we must show
$\inper[\Psi_2]{N'}{\td{N_0}{\psi_2}}{\td{A}{\psi_1\psi_2}}$. But this follows
from $\inper[\Psi_2]{N'}{\td{N}{\psi_2}}{\td{A}{\psi_1\psi_2}}$ and
$\inper[\Psi_2]{\td{N}{\psi_2}}{\td{N_0}{\psi_2}}{\td{A}{\psi_1\psi_2}}$.
\end{enumerate}
\end{enumerate}
Once again the equations hold by \cref{lem:expansion} and the elimination rule
for dependent functions.

\paragraph{Computation}
If $\coftype{M}{A}$ then $\ceqtm{\iaout{r}{\iain{r}{M,F},G}}{M}{A}$.

By \cref{lem:coftype-ceqtm} and the introduction and elimination rules we have
just proven, it suffices to show that for any $\psitd$,
$\inper[\Psi']{\iaout{\td{r}{\psi}}{\iain{\td{r}{\psi}}{\td{M}{\psi},\td{F}{\psi}},\td{G}{\psi}}}
{\td{M}{\psi}}{\td{A}{\psi}}$.

\begin{enumerate}
\item If $\td{r}{\psi} = 0$ then
$\iaout{0}{\iain{0}{\td{M}{\psi},\td{F}{\psi}},\td{G}{\psi}} \steps^*
\td{M}{\psi}$ and the result follows by $\coftype{M}{A}$.

\item If $\td{r}{\psi} = 1$ then
by the equations in the introduction and elimination rules,
\[ \ceqtm[\Psi']
{\iaout{1}{\iain{1}{\td{M}{\psi},\td{F}{\psi}},\td{G}{\psi}}}
{\app{\td{G}{\psi}}{\app{\td{F}{\psi}}{\td{M}{\psi}}}}
{\td{A}{\psi}} \]
and by $\eqtmres{r=1}{\oft{a}{A}}{\app{G}{\app{F}{a}}}{a}{A}$ we
have
$\ceqtm[\Psi']
{\app{\td{G}{\psi}}{\app{\td{F}{\psi}}{\td{M}{\psi}}}}
{\td{M}{\psi}}
{\td{A}{\psi}}$. The result follows immediately.

\item If $\td{r}{\psi} = x$ then
$\iaout{x}{\iain{x}{\td{M}{\psi},\td{F}{\psi}},\td{G}{\psi}} \steps
\td{M}{\psi}$ and the result follows by $\coftype{M}{A}$.
\end{enumerate}

\paragraph{Eta}
If $\coftype{M}{\ia{r}{A,B,F,G}}$ then
$\ceqtm{\iain{r}{\iaout{r}{M,G},F}}{M}{\ia{r}{A,B,F,G}}$.

By \cref{lem:coftype-ceqtm} and the introduction and elimination rules we have
just proven, it suffices to show that for any $\psitd$,
$\inper[\Psi']{\iain{\td{r}{\psi}}{\iaout{\td{r}{\psi}}{\td{M}{\psi},\td{F}{\psi}},\td{G}{\psi}}}
{\td{M}{\psi}}{\ia{\td{r}{\psi}}{\td{I}{\psi}}}$.

\begin{enumerate}
\item If $\td{r}{\psi} = 0$ then
$\iain{0}{\iaout{0}{\td{M}{\psi},\td{G}{\psi}},\td{F}{\psi}} \steps^*
\td{M}{\psi}$ and the result follows by $\coftype{M}{A}$.

\item If $\td{r}{\psi} = 1$ then
by the equations in the introduction and elimination rules,
\[ \ceqtm[\Psi']
{\iain{1}{\iaout{1}{\td{M}{\psi},\td{G}{\psi}},\td{F}{\psi}}}
{\app{\td{F}{\psi}}{\app{\td{G}{\psi}}{\td{M}{\psi}}}}
{\td{B}{\psi}} \]
and by assumption,
$\ceqtm[\Psi']
{\app{\td{F}{\psi}}{\app{\td{G}{\psi}}{\td{M}{\psi}}}}
{\td{M}{\psi}}
{\td{B}{\psi}}$. The result follows immediately.

\item If $\td{r}{\psi} = x$ then
$\td{M}{\psi}\evals\iain{x}{N,F''}$ where
$\coftype[\Psi']{N}{\td{A}{\psi}}$ and
$\ceqtmres[\Psi']{x=1}{F''}{\td{F}{\psi}}{\arr{\td{A}{\psi}}{\td{B}{\psi}}}$. By
\cref{lem:coftype-evals-ceqtm}, since $\cpretype[\Psi']{\ia{x}{\td{I}{\psi}}}$
is cubical,
$\ceqtm[\Psi']{\td{M}{\psi}}{\iain{x}{N,F''}}{\ia{x}{\td{I}{\psi}}}$. By the
equations in the introduction and elimination rules,
\[
\ceqtm[\Psi']
{\iain{x}{\iaout{x}{\td{M}{\psi},\td{G}{\psi}},\td{F}{\psi}}}
{\iain{x}{\iaout{x}{\iain{x}{N,\td{F}{\psi}},\td{G}{\psi}},F''}}
{\ia{x}{\td{I}{\psi}}}.
\]
By the computation rule, the right-hand side is
$\ceqtm[\Psi']{-}{\iain{x}{N,F''}}{\ia{x}{\td{I}{\psi}}}$, and so the result
$\inper[\Psi']{\iain{x}{\iaout{x}{\td{M}{\psi},\td{G}{\psi}},\td{F}{\psi}}}
{\td{M}{\psi}}{\ia{x}{\td{I}{\psi}}}$ follows.
\end{enumerate}

\paragraph{Kan}
$\ceqpretype{\ia{r}{I}}{\ia{r}{I'}}$ are equally Kan;
if $r=0$, $\ceqpretype{\ia{0}{I}}{A}$ are equally Kan; and
if $r=1$, $\ceqpretype{\ia{1}{I}}{B}$ are equally Kan.

The latter two equations hold by \cref{lem:type-expansion}, because
when $r=0$, $\ia{0}{\td{I}{\psi}}\steps\td{A}{\psi}$ and $\cwftype{A}$, and when
$r=1$, $\ia{1}{\td{I}{\psi}}\steps\td{B}{\psi}$ and $\cwftyperes{1=1}{B}$.

The first Kan condition (in its unary form) requires that for any $\psitd$, if
\begin{enumerate}
\item $\coftype[\Psi']{M}{\ia{\td{r''}{\psi}}{\td{I}{\psi}}}$,
\item $\ceqtmres[\Psi',y]{r_i=\e,r_j=\e'}{N^\e_i}{N^{\e'}_j}
{\ia{\td{r''}{\psi}}{\td{I}{\psi}}}$ for any $i,j,\e,\e'$, and
\item $\ceqtmres[\Psi']{r_i=\e}{\dsubst{N^\e_i}{r}{y}}{M}
{\ia{\td{r''}{\psi}}{\td{I}{\psi}}}$ for any $i,\e$,
\end{enumerate}
then $\coftype[\Psi']{\hcomgeneric{\etc{r_i}}{\ia{\td{r''}{\psi}}{\td{I}{\psi}}}}%
{\ia{\td{r''}{\psi}}{\td{I}{\psi}}}$.

If $\td{r''}{\psi}=\e$ then this follows from the fact that
$\ceqpretype{\ia{0}{I}}{A}$ and $\ceqpretype{\ia{1}{I}}{B}$ are equally Kan.
If $\td{r''}{\psi}=x$, we appeal to \cref{lem:coherent-expansion} with
\[
\{X^{\Psi''}_{\psi'}\}_{\msubsts{\Psi''}{\psi'}{\Psi'}} =
\begin{cases}
\hcom{\etc{\td{r_i}{\psi'}}}{\td{A}{\psi\psi'}}{\td{r}{\psi'}}{\td{r'}{\psi'}}
{\td{M}{\psi'}}{\etc{y.\td{N^\e_i}{\psi'}}}
& \text{if $\td{x}{\psi'}=0$} \\
\hcom{\etc{\td{r_i}{\psi'}}}{\td{B}{\psi\psi'}}{\td{r}{\psi'}}{\td{r'}{\psi'}}
{\td{M}{\psi'}}{\etc{y.\td{N^\e_i}{\psi'}}}
& \text{if $\td{x}{\psi'}=1$} \\
\iain{x'}{\hcom{\etc{\td{r_i}{\psi'}},x'}{\td{A}{\psi\psi'}}
{\td{r}{\psi'}}{\td{r'}{\psi'}}{\iaout{x'}{\td{M}{\psi'},\td{G}{\psi\psi'}}}
{\etc{T}},\td{F}{\psi\psi'}}
& \text{if $\td{x}{\psi'}=x'$} \\
\end{cases}
\]
\begin{gather*}
\etc{T} = \etc{y.\iaout{x'}{\td{N^\e_i}{\psi'},\td{G}{\psi\psi'}}},
          z.\hcom{\etc{\td{r_i}{\psi'}}}{\td{A}{\psi\psi'}}{\td{r}{\psi'}}{z}
                 {\td{M}{\psi'}}{\etc{y.\td{N^\e_i}{\psi'}}}, \\
          z.\app{\td{G}{\psi\psi'}}
                {\hcom{\etc{\td{r_i}{\psi'}}}{\td{B}{\psi\psi'}}{\td{r}{\psi'}}{z}
                      {\td{M}{\psi'}}{\etc{y.\td{N^\e_i}{\psi'}}}}
\end{gather*}
We begin by showing the following result, which states the well-formedness of
the Kan composition in $X^{\Psi'}_{\id[\Psi']}$; it is the core of the argument
that the first three Kan conditions for $\ia{r}{I}$ hold.

\begin{lemma}\label{lem:ia-hcom}
If $\Iso[\Psi']{x}{\td{I}{\psi}}$,
\begin{enumerate}
\item $\coftype[\Psi']{M}{\ia{x}{\td{I}{\psi}}}$,
\item $\ceqtmres[\Psi',y]{r_i=\e,r_j=\e'}{N^\e_i}{N^{\e'}_j}
{\ia{x}{\td{I}{\psi}}}$ for any $i,j,\e,\e'$, and
\item $\ceqtmres[\Psi']{r_i=\e}{\dsubst{N^\e_i}{r}{y}}{M}
{\ia{x}{\td{I}{\psi}}}$ for any $i,\e$,
\end{enumerate}
then
\begin{enumerate}
\item $\coftype[\Psi']{\iaout{x}{M,\td{G}{\psi}}}{\td{A}{\psi}}$,
\item $\ceqtmres[\Psi',y]{r_i=\e,r_j=\e'}
{\iaout{x}{N^\e_i,\td{G}{\psi}}}
{\iaout{x}{N^{\e'}_j,\td{G}{\psi}}}
{\td{A}{\psi}}$ for any $i,j,\e,\e'$,
\item $\ceqtmres[\Psi']{r_i=\e}
{\iaout{x}{\dsubst{N^\e_i}{r}{y},\td{G}{\psi}}}
{\iaout{x}{M,\td{G}{\psi}}}
{\td{A}{\psi}}$ for any $i,\e$,
\item $\coftyperes[\Psi',y]{x=0}
{\hcom{\etc{r_i}}{\td{A}{\psi}}{r}{y}{M}{\etc{y.N^\e_i}}}
{\td{A}{\psi}}$,
\item $\coftyperes[\Psi',y]{x=1}
{\app{\td{G}{\psi}}
     {\hcom{\etc{r_i}}{\td{B}{\psi}}{r}{y}{M}{\etc{y.N^\e_i}}}}
{\td{A}{\psi}}$,
\item $\ceqtmres[\Psi',y]{r_i=\e,x=0}
{\iaout{x}{N^\e_i,\td{G}{\psi}}}
{\hcom{\etc{r_i}}{\td{A}{\psi}}{r}{y}{M}{\etc{y.N^\e_i}}}
{\td{A}{\psi}}$ for any $i,\e$,
\item $\ceqtmres[\Psi',y]{r_i=\e,x=1}
{\iaout{x}{N^\e_i,\td{G}{\psi}}}
{\app{\td{G}{\psi}}
     {\hcom{\etc{r_i}}{\td{B}{\psi}}{r}{y}{M}{\etc{y.N^\e_i}}}}
{\td{A}{\psi}}$ for any $i,\e$,
\item $\ceqtmres[\Psi']{x=0}
{\hcom{\etc{r_i}}{\td{A}{\psi}}{r}{r}{M}{\etc{y.N^\e_i}}}
{\iaout{x}{M,\td{G}{\psi}}}
{\td{A}{\psi}}$, and
\item $\ceqtmres[\Psi']{x=1}
{\app{\td{G}{\psi}}
     {\hcom{\etc{r_i}}{\td{B}{\psi}}{r}{r}{M}{\etc{y.N^\e_i}}}}
{\iaout{x}{M,\td{G}{\psi}}}
{\td{A}{\psi}}$.
\end{enumerate}
\end{lemma}

\begin{proof}~
\begin{enumerate}
\item[1-3.] \setcounter{enumi}{3}
These follow by the elimination rule.

\item Show that for any $\msubsts{\Psi''}{\psi'}{(\Psi',y)}$ satisfying $x=0$,
$\coftype[\Psi'']
{\hcom{\etc{\td{r_i}{\psi'}}}{\td{A}{\psi\psi'}}{\td{r}{\psi'}}{\td{y}{\psi'}}{\td{M}{\psi'}}{\etc{y.\td{N^\e_i}{\psi'}}}}
{\td{A}{\psi\psi'}}$. The result follows by the first Kan condition of
$\cwftype[\Psi'']{\td{A}{\psi\psi'}}$, because
$\ceqpretype[\Psi'']{\ia{0}{\td{I}{\psi\psi'}}}{\td{A}{\psi\psi'}}$.

\item As in the previous case, for any $\msubsts{\Psi''}{\psi'}{(\Psi',y)}$
satisfying $x=1$,
$\coftype[\Psi'']{\hcomsym_{\td{B}{\psi\psi'}}}{\td{B}{\psi\psi'}}$. The result
follows by $\coftype[\Psi'']{\td{G}{\psi\psi'}}
{\arr{\td{B}{\psi\psi'}}{\td{A}{\psi\psi'}}}$ and the elimination rule for
dependent functions.

\item Under any satisfying $\psi'$, the left-hand side is equal to
$\td{N^\e_i}{\psi'}$ by the elimination rule, and the right-hand side is equal
to $\dsubst{\td{N^\e_i}{\psi'}}{\td{y}{\psi'}}{\td{y}{\psi'}}$ by the third Kan
condition of $\cwftype[\Psi'']{\td{A}{\psi\psi'}}$.

\item Under any satisfying $\psi'$, the left-hand side is equal to
$\app{\td{G}{\psi\psi'}}{\td{N^\e_i}{\psi'}}$ by the elimination rule, and the
right-hand side is equal to
$\app{\td{G}{\psi\psi'}}
{\dsubst{\td{N^\e_i}{\psi'}}{\td{y}{\psi'}}{\td{y}{\psi'}}}$ by the elimination
rule for dependent functions and the third Kan condition of
$\cwftype[\Psi'']{\td{B}{\psi\psi'}}$.

\item Under any satisfying $\psi'$, the left-hand side is equal to
$\td{M}{\psi'}$ by the second Kan condition of
$\cwftype[\Psi'']{\td{A}{\psi\psi'}}$, and the right-hand side is equal to
$\td{M}{\psi'}$ by the elimination rule.

\item Under any satisfying $\psi'$, the left-hand side is equal to
$\app{\td{G}{\psi\psi'}}{\td{M}{\psi'}}$ by the elimination rule for dependent
functions and the second Kan condition of $\cwftype[\Psi'']{\td{B}{\psi\psi'}}$,
and the right-hand side is equal to $\app{\td{G}{\psi\psi'}}{\td{M}{\psi'}}$ by
the elimination rule.
\qedhere
\end{enumerate}
\end{proof}

We now show that for all $\msubsts{\Psi''}{\psi'}{\Psi'}$,
$\ceqtm[\Psi'']
{X^{\Psi''}_{\psi'}}
{\td{(X^{\Psi'}_{\id[\Psi']})}{\psi'}}
{\ia{\td{x}{\psi'}}{\td{I}{\psi\psi'}}}$, from which it follows that
$\ceqtm[\Psi']
{\hcomgeneric{\etc{r_i}}{\ia{x}{\td{I}{\psi}}}}
{X^{\Psi'}_{\id[\Psi']}}
{\ia{x}{\td{I}{\psi}}}$. This completes the proof of the first Kan condition
(including the binary case, which follows by transitivity).
\begin{enumerate}
\item $\td{x}{\psi'}=0$.

By the introduction rule, \cref{lem:ia-hcom}, and the third Kan condition of
$\cwftype[\Psi'']{\td{A}{\psi\psi'}}$,
\[
\ceqtm[\Psi'']
{\td{(X^{\Psi'}_{\id[\Psi']})}{\psi'}}
{\hcom{\etc{\td{r_i}{\psi'}}}{\td{A}{\psi\psi'}}{\td{r}{\psi'}}{\td{r'}{\psi'}}
                 {\td{M}{\psi'}}{\etc{y.\td{N^\e_i}{\psi'}}}}
{\td{A}{\psi\psi'}}
\]
which, by $\ceqpretype[\Psi'']{\td{A}{\psi\psi'}}{\ia{0}{\td{I}{\psi\psi'}}}$,
completes this case.

\item $\td{x}{\psi'}=1$.

By the introduction rule, \cref{lem:ia-hcom}, and the third Kan condition of
$\cwftype[\Psi'']{\td{A}{\psi\psi'}}$,
\[
\ceqtm[\Psi'']{\td{(X^{\Psi'}_{\id[\Psi']})}{\psi'}}
{\app{\td{F}{\psi\psi'}}
     {\app{\td{G}{\psi\psi'}}
          {\hcom{\etc{\td{r_i}{\psi'}}}{\td{B}{\psi\psi'}}
                {\td{r}{\psi'}}{\td{r'}{\psi'}}
                {\td{M}{\psi'}}{\etc{y.\td{N^\e_i}{\psi'}}}}}}
{\td{B}{\psi\psi'}}.
\]
The result follows by the first Kan condition of
$\cwftype[\Psi'']{\td{B}{\psi\psi'}}$,
$\ceqpretype[\Psi'']{\td{B}{\psi\psi'}}{\ia{1}{\td{I}{\psi\psi'}}}$, and
$\eqtm[\Psi'']{\oft{b}{\td{B}{\psi\psi'}}}{\app{\td{F}{\psi\psi'}}
{\app{\td{G}{\psi\psi'}}{b}}}{b}{\td{B}{\psi\psi'}}$.

\item $\td{x}{\psi'}=x'$.

Immediate by the introduction rule, \cref{lem:ia-hcom}, and the first Kan
condition of $\cwftype[\Psi'']{\td{A}{\psi\psi'}}$.
\end{enumerate}

The second Kan condition requires that for any $\psitd$, under the same
hypotheses as before, if $r=r'$, then
$\ceqtm[\Psi']
{\hcom{\etc{r_i}}{\ia{\td{r''}{\psi}}{\td{I}{\psi}}}{r}{r}{M}{\etc{y.N^\e_i}}}
{M}{\ia{\td{r''}{\psi}}{\td{I}{\psi}}}$.

If $\td{r''}{\psi}=\e$ then this follows from the fact that
$\ceqpretype{\ia{0}{I}}{A}$ and $\ceqpretype{\ia{1}{I}}{B}$ are equally Kan.
If $\td{r''}{\psi}=x$ then we have already shown that the left-hand side
\[
\ceqtm[\Psi']{-}
{\iain{x}{\hcom{\etc{r_i},x}{\td{A}{\psi}}
{r}{r}{\iaout{x}{M,\td{G}{\psi}}}
{\dots},\td{F}{\psi}}}
{\ia{x}{\td{I}{\psi}}}.
\]
By the introduction rule, \cref{lem:ia-hcom}, and the second Kan condition of
$\cwftype[\Psi']{\td{A}{\psi}}$, the above is
\[
\ceqtm[\Psi']{-}{\iain{x}{\iaout{x}{M,\td{G}{\psi}},\td{F}{\psi}}}
{\ia{x}{\td{I}{\psi}}}
\]
which by the eta rule $\ceqtm[\Psi']{-}{M}{\ia{x}{\td{I}{\psi}}}$.

The third Kan condition requires that for any $\psitd$, under the same
hypotheses as before, if $r_i=\e$, then
$\ceqtm[\Psi']
{\hcom{\etc{r_i}}{\ia{\td{r''}{\psi}}{\td{I}{\psi}}}{r}{r'}{M}{\etc{y.N^\e_i}}}
{\dsubst{N^\e_i}{r'}{y}}{\ia{\td{r''}{\psi}}{\td{I}{\psi}}}$.

If $\td{r''}{\psi}=\e$ then this follows from the fact that
$\ceqpretype{\ia{0}{I}}{A}$ and $\ceqpretype{\ia{1}{I}}{B}$ are equally Kan.
If $\td{r''}{\psi}=x$ then we have already shown that the left-hand side
\[
\ceqtm[\Psi']{-}
{\iain{x}{\hcom{\etc{r_i},x}{\td{A}{\psi}}
{r}{r'}{\dots}
{\etc{y.\iaout{x}{N^\e_i,\td{G}{\psi}}},z.\dots,z.\dots},\td{F}{\psi}}}
{\ia{x}{\td{I}{\psi}}}.
\]

By the introduction rule, \cref{lem:ia-hcom}, and the third Kan condition of
$\cwftype[\Psi']{\td{A}{\psi}}$, the above is

\[
\ceqtm[\Psi']{-}{\iain{x}
{\iaout{x}{\dsubst{N^\e_i}{r'}{y},\td{G}{\psi}},\td{F}{\psi}}}
{\ia{x}{\td{I}{\psi}}}
\]
which by the eta rule
$\ceqtm[\Psi']{-}{\dsubst{N^\e_i}{r'}{y}}{\ia{x}{\td{I}{\psi}}}$.

The fourth Kan condition (in its unary form) requires that for any
$\msubsts{(\Psi',x)}{\psi}{\Psi}$, if
$\coftype[\Psi']{M}{\ia{\dsubst{\td{r''}{\psi}}{r}{x}}{\dsubst{\td{I}{\psi}}{r}{x}}}$, then
$\coftype[\Psi']{\coe{x.\ia{\td{r''}{\psi}}{\td{I}{\psi}}}{r}{r'}{M}}
{\ia{\dsubst{\td{r''}{\psi}}{r'}{x}}{\dsubst{\td{I}{\psi}}{r'}{x}}}$.

There are three cases.
If $\td{r''}{\psi}=\e$ then this follows from the fact that
$\cpretype[\Psi']{\ia{\e}{\td{I}{\psi}}}$ and $\td{A}{\psi}$ or $\td{B}{\psi}$
(depending on $\e$) are equally Kan.
If $\td{r''}{\psi}=x$ then we know
$\coftype[\Psi']{M}{\ia{r}{\dsubst{\td{I}{\psi}}{r}{x}}}$ and must show
$\coftype[\Psi']{\coe{x.\ia{x}{\td{I}{\psi}}}{r}{r'}{M}}
{\ia{r'}{\dsubst{\td{I}{\psi}}{r'}{x}}}$.
By \cref{lem:expansion} it suffices to show
\[
\coftype[\Psi']
{\iain{r'}{
\coe{x.\td{A}{\psi}}{r}{r'}{\iaout{r}{M,\dsubst{\td{G}{\psi}}{r}{x}}},
\dsubst{\td{F}{\psi}}{r'}{x}}}
{\ia{r'}{\dsubst{\td{I}{\psi}}{r'}{x}}}.
\]
In this case we know $\Iso[\Psi',x]{x}{\td{I}{\psi}}$, so
$\Iso[\Psi']{r}{\dsubst{\td{I}{\psi}}{r}{x}}$, and by the elimination rule and
$\coftype[\Psi']{M}{\ia{r}{\dsubst{\td{I}{\psi}}{r}{x}}}$, we have
$\coftype[\Psi']
{\iaout{r}{M,\dsubst{\td{G}{\psi}}{r}{x}}}
{\dsubst{\td{A}{\psi}}{r}{x}}$.
By the fourth Kan condition of $\cwftype[\Psi',x]{\td{A}{\psi}}$,
$\coftype[\Psi']
{\coe{x.\td{A}{\psi}}{r}{r'}{\iaout{r}{M,\dsubst{\td{G}{\psi}}{r}{x}}}}
{\dsubst{\td{A}{\psi}}{r'}{x}}$. The result follows by the introduction rule
and $\Iso[\Psi']{r'}{\dsubst{\td{I}{\psi}}{r'}{x}}$.

The final case is $\td{r''}{\psi} = w \neq x$: if
$\coftype[\Psi']{M}{\ia{w}{\dsubst{\td{I}{\psi}}{r}{x}}}$, then
$\coftype[\Psi']{\coe{x.\ia{w}{\td{I}{\psi}}}{r}{r'}{M}}
{\ia{w}{\dsubst{\td{I}{\psi}}{r'}{x}}}$. We appeal to
\cref{lem:coherent-expansion} with
\[
\{X^{\Psi''}_{\psi'}\}_{\msubsts{\Psi''}{\psi'}{\Psi'}} =
\begin{cases}
\coe{x.\td{A}{\psi\psi'}}{\td{r}{\psi'}}{\td{r'}{\psi'}}{\td{M}{\psi'}}
& \text{if $\td{w}{\psi'}=0$} \\
\coe{x.\td{B}{\psi\psi'}}{\td{r}{\psi'}}{\td{r'}{\psi'}}{\td{M}{\psi'}}
& \text{if $\td{w}{\psi'}=1$} \\
\iain{w'}{C,\dsubst{\td{F}{\psi\psi'}}{\td{r'}{\psi'}}{x}}
& \text{if $\td{w}{\psi'}=w'$} \\
\end{cases}
\]
\begin{gather*}
C = \com{w'}{x.\td{A}{\psi\psi'}}
        {\td{r}{\psi'}}{\td{r'}{\psi'}}{D}
        {y.\coe{x.\td{A}{\psi\psi'}}{\td{r}{\psi'}}{y}{\td{M}{\psi'}},
         y.\app{\dsubst{\td{G}{\psi\psi'}}{y}{x}}
               {\coe{x.\td{B}{\psi\psi'}}{\td{r}{\psi'}}{y}{\td{M}{\psi'}}}} \\
D = \iaout{w'}{\td{M}{\psi'},\dsubst{\td{G}{\psi\psi'}}{\td{r}{\psi'}}{x}}
\end{gather*}
We begin by showing the following result, which states the well-formedness of
the $\comsym$ in $X^{\Psi'}_{\id[\Psi']}$.

\begin{lemma}\label{lem:ia-coe}
If $\Iso[\Psi',x]{w}{\td{I}{\psi}}$ and
$\coftype[\Psi']{M}{\ia{w}{\dsubst{\td{I}{\psi}}{r}{x}}}$ where
$w\neq x$, then
\begin{enumerate}
\item $\coftype[\Psi']{\iaout{w}{M,\dsubst{\td{G}{\psi}}{r}{x}}}
{\dsubst{\td{A}{\psi}}{r}{x}}$,
\item $\coftyperes[\Psi',y]{w=0}
{\coe{x.\td{A}{\psi}}{r}{y}{M}}
{\dsubst{\td{A}{\psi}}{y}{x}}$,
\item $\coftyperes[\Psi',y]{w=1}
{\app{\dsubst{\td{G}{\psi}}{y}{x}}{\coe{x.\td{B}{\psi}}{r}{y}{M}}}
{\dsubst{\td{A}{\psi}}{y}{x}}$,
\item $\ceqtmres[\Psi']{w=0}
{\coe{x.\td{A}{\psi}}{r}{r}{M}}
{\iaout{w}{M,\dsubst{\td{G}{\psi}}{r}{x}}}
{\dsubst{\td{A}{\psi}}{r}{x}}$, and
\item $\ceqtmres[\Psi']{w=1}
{\app{\dsubst{\td{G}{\psi}}{r}{x}}{\coe{x.\td{B}{\psi}}{r}{r}{M}}}
{\iaout{w}{M,\dsubst{\td{G}{\psi}}{r}{x}}}
{\dsubst{\td{A}{\psi}}{r}{x}}$.
\end{enumerate}
\end{lemma}
\begin{proof}~
\begin{enumerate}
\item Follows by the elimination rule.

\item Let $\msubsts{\Psi''}{\psi'}{(\Psi',y)}$ satisfy $w=0$; we must show
$\coftype[\Psi'']{\coe{x.\td{A}{\psi\psi'}}{\td{r}{\psi'}}{\td{y}{\psi'}}{\td{M}{\psi'}}}
{\td{\dsubst{\td{A}{\psi}}{y}{x}}{\psi'}}$. We know that
$\coftype[\Psi']{M}{\ia{w}{\dsubst{\td{I}{\psi}}{r}{x}}}$, so
$\coftype[\Psi'']{\td{M}{\psi'}}{\ia{0}{\dsubst{\td{I}{\psi\psi'}}{\td{r}{\psi'}}{x}}}$
and thus
$\coftype[\Psi'']{\td{M}{\psi'}}{\dsubst{\td{A}{\psi\psi'}}{\td{r}{\psi'}}{x}}$.
The result follows from the fourth Kan condition of
$\cwftype[\Psi'',x]{\td{A}{\psi\psi'}}$.

\item Let $\msubsts{\Psi''}{\psi'}{(\Psi',y)}$ satisfy $w=1$; we must show
$\coftype[\Psi'']
{\app{\td{\dsubst{\td{G}{\psi}}{y}{x}}{\psi'}}
     {\coe{x.\td{B}{\psi\psi'}}{\td{r}{\psi'}}{\td{y}{\psi'}}{\td{M}{\psi'}}}}
{\td{\dsubst{\td{A}{\psi}}{y}{x}}{\psi'}}$. We know
$\coftype[\Psi'']{\td{M}{\psi'}}{\ia{1}{\dsubst{\td{I}{\psi\psi'}}{\td{r}{\psi'}}{x}}}$,
so
$\coftype[\Psi'']{\td{M}{\psi'}}{\dsubst{\td{B}{\psi\psi'}}{\td{r}{\psi'}}{x}}$.
We know $\cwftyperes[\Psi',x]{w=1}{\td{B}{\psi}}$, so
$\cwftype[\Psi'',x]{\td{B}{\psi\psi'}}$, and by its fourth Kan condition,
$\coftype[\Psi'']
{\coe{x.\td{B}{\psi\psi'}}{\td{r}{\psi'}}{\td{y}{\psi'}}{\td{M}{\psi'}}}
{\td{\dsubst{\td{B}{\psi}}{y}{x}}{\psi'}}$.
We know
$\coftyperes[\Psi',x]{w=1}{\td{G}{\psi}}{\arr{\td{B}{\psi}}{\td{A}{\psi}}}$, so
$\coftype[\Psi'']{\td{\dsubst{\td{G}{\psi}}{y}{x}}{\psi'}}
{\arr{\td{\dsubst{\td{B}{\psi}}{y}{x}}{\psi'}}
     {\td{\dsubst{\td{A}{\psi}}{y}{x}}{\psi'}}}$, and the result follows by the
elimination rule for dependent functions.

\item Under a satisfying $\psi'$, show
$\ceqtm[\Psi'']
{\coe{x.\td{A}{\psi\psi'}}{\td{r}{\psi'}}{\td{r}{\psi'}}{\td{M}{\psi'}}}
{\iaout{0}{\td{M}{\psi'},\td{\dsubst{\td{G}{\psi}}{r}{x}}{\psi'}}}
{\td{\dsubst{\td{A}{\psi}}{r}{x}}{\psi'}}$. By the fifth Kan condition of
$\cwftype[\Psi'',x]{\td{A}{\psi\psi'}}$, the left-hand side is
$\ceqtm[\Psi'']{-}{\td{M}{\psi'}}
{\td{\dsubst{\td{A}{\psi}}{r}{x}}{\psi'}}$; by the elimination rule, so is the
right-hand side.

\item Under a satisfying $\psi'$, show
\[
\ceqtm[\Psi'']
{\app{\td{\dsubst{\td{G}{\psi}}{r}{x}}{\psi'}}
     {\coe{x.\td{B}{\psi\psi'}}{\td{r}{\psi'}}{\td{r}{\psi'}}{\td{M}{\psi'}}}}
{\iaout{1}{\td{M}{\psi'},\td{\dsubst{\td{G}{\psi}}{r}{x}}{\psi'}}}
{\td{\dsubst{\td{A}{\psi}}{r}{x}}{\psi'}}.
\]
By the elimination rule for dependent functions and the fifth Kan condition of
$\cwftype[\Psi'',x]{\td{B}{\psi\psi'}}$, the left-hand side is
$\ceqtm[\Psi'']{-}
{\app{\td{\dsubst{\td{G}{\psi}}{r}{x}}{\psi'}}
     {\td{M}{\psi'}}}
{\td{\dsubst{\td{A}{\psi}}{r}{x}}{\psi'}}$; by the elimination rule, so is the
right-hand side.
\qedhere
\end{enumerate}
\end{proof}

Now, for each $\msubsts{\Psi''}{\psi'}{\Psi'}$, we show that
$\ceqtm[\Psi'']
{X^{\Psi''}_{\psi'}}
{\td{(X^{\Psi'}_{\id[\Psi']})}{\psi'}}
{\ia{\td{w}{\psi'}}{\dsubst{\td{I}{\psi\psi'}}{\td{r'}{\psi'}}{x}}}$, from which
it follows that
$\ceqtm[\Psi']{\coe{x.\ia{w}{\td{I}{\psi}}}{r}{r'}{M}}
{X^{\Psi'}_{\id[\Psi']}}
{\ia{w}{\dsubst{\td{I}{\psi}}{r'}{x}}}$. This completes the proof of the fourth
Kan condition (including the binary case, by transitivity).
\begin{enumerate}
\item $\td{w}{\psi'}=0$.

By the introduction rule, \cref{lem:ia-coe}, and \cref{thm:com},
$\ceqtm[\Psi'']{\td{(X^{\Psi'}_{\id[\Psi']})}{\psi'}}
{\coe{x.\td{A}{\psi\psi'}}{\td{r}{\psi'}}{\td{r'}{\psi'}}{\td{M}{\psi'}}}
{\dsubst{\td{A}{\psi\psi'}}{\td{r'}{\psi'}}{x}}$.
The result follows by
$\ceqpretype[\Psi'']
{\dsubst{\td{A}{\psi\psi'}}{\td{r'}{\psi'}}{x}}
{\ia{0}{\dsubst{\td{I}{\psi\psi'}}{\td{r'}{\psi'}}{x}}}$.

\item $\td{w}{\psi'}=1$.

By the introduction rule, \cref{lem:ia-coe}, and \cref{thm:com},
\[
\ceqtm[\Psi'']{\td{(X^{\Psi'}_{\id[\Psi']})}{\psi'}}
{\app{\dsubst{\td{F}{\psi\psi'}}{\td{r'}{\psi'}}{x}}
     {\app{\dsubst{\td{G}{\psi\psi'}}{\td{r'}{\psi'}}{x}}
          {\coe{x.\td{B}{\psi\psi'}}{\td{r}{\psi'}}{\td{r'}{\psi'}}{\td{M}{\psi'}}}}}
{\dsubst{\td{B}{\psi\psi'}}{\td{r'}{\psi'}}{x}}.
\]
By $\Iso[\Psi',x]{w}{\td{I}{\psi}}$, we have
$\Iso[\Psi'']{1}{\td{\dsubst{\td{I}{\psi}}{r'}{x}}{\psi'}}$, so
$\cwftype[\Psi'',x]{\td{\dsubst{\td{B}{\psi}}{r'}{x}}{\psi'}}$ and
\[
\eqtm[\Psi'']{\oft{b}{\td{\dsubst{\td{B}{\psi}}{r'}{x}}{\psi'}}}
{\app{\td{\dsubst{\td{F}{\psi}}{r'}{x}}{\psi'}}
     {\app{\td{\dsubst{\td{G}{\psi}}{r'}{x}}{\psi'}}{b}}}{b}
{\td{\dsubst{\td{B}{\psi}}{r'}{x}}{\psi'}}.
\]
By $\coftype[\Psi']{M}{\ia{w}{\dsubst{\td{I}{\psi}}{r}{x}}}$ and the fourth Kan
condition of $\cwftype[\Psi'',x]{\td{\dsubst{\td{B}{\psi}}{r'}{x}}{\psi'}}$, we
have $\coftype[\Psi'']
{\coe{x.\td{B}{\psi\psi'}}{\td{r}{\psi'}}{\td{r'}{\psi'}}{\td{M}{\psi'}}}
{\dsubst{\td{B}{\psi\psi'}}{\td{r'}{\psi'}}{x}}$. The result follows by
instantiating the open equation.

\item $\td{w}{\psi'}=w'$.

Immediate by the introduction rule, \cref{lem:ia-coe}, and \cref{thm:com}.
\end{enumerate}

The fifth Kan condition requires that for any
$\msubsts{(\Psi',x)}{\psi}{\Psi}$, if
$\coftype[\Psi']{M}{\ia{\dsubst{\td{r''}{\psi}}{r}{x}}{\dsubst{\td{I}{\psi}}{r}{x}}}$, then
$\ceqtm[\Psi']{\coe{x.\ia{\td{r''}{\psi}}{\td{I}{\psi}}}{r}{r}{M}}{M}
{\ia{\dsubst{\td{r''}{\psi}}{r}{x}}{\dsubst{\td{I}{\psi}}{r}{x}}}$.

There are three cases.
If $\td{r''}{\psi}=\e$ then this follows from the fact that
$\cpretype[\Psi']{\ia{\e}{\td{I}{\psi}}}$ and $\td{A}{\psi}$ or $\td{B}{\psi}$
(depending on $\e$) are equally Kan.
If $\td{r''}{\psi}=x$ then we know
$\coftype[\Psi']{M}{\ia{r}{\dsubst{\td{I}{\psi}}{r}{x}}}$ and must show
$\ceqtm[\Psi']{\coe{x.\ia{x}{\td{I}{\psi}}}{r}{r}{M}}{M}
{\ia{r}{\dsubst{\td{I}{\psi}}{r}{x}}}$.
By \cref{lem:expansion} it suffices to show
\[
\ceqtm[\Psi']
{\iain{r}{
\coe{x.\td{A}{\psi}}{r}{r}{\iaout{r}{M,\dsubst{\td{G}{\psi}}{r}{x}}},
\dsubst{\td{F}{\psi}}{r}{x}}}{M}
{\ia{r}{\dsubst{\td{I}{\psi}}{r}{x}}}.
\]
In this case we know $\Iso[\Psi',x]{x}{\td{I}{\psi}}$, so
$\Iso[\Psi']{r}{\dsubst{\td{I}{\psi}}{r}{x}}$, and by the elimination rule and
$\coftype[\Psi']{M}{\ia{r}{\dsubst{\td{I}{\psi}}{r}{x}}}$, we have
$\coftype[\Psi']
{\iaout{r}{M,\dsubst{\td{G}{\psi}}{r}{x}}}
{\dsubst{\td{A}{\psi}}{r}{x}}$.
By the fifth Kan condition of $\cwftype[\Psi',x]{\td{A}{\psi}}$,
$\ceqtm[\Psi']
{\coe{x.\td{A}{\psi}}{r}{r}{\iaout{r}{M,\dsubst{\td{G}{\psi}}{r}{x}}}}
{\iaout{r}{M,\dsubst{\td{G}{\psi}}{r}{x}}}
{\dsubst{\td{A}{\psi}}{r}{x}}$. The result follows from the introduction and eta
rules.

The final case is $\td{r''}{\psi} = w \neq x$: if
$\coftype[\Psi']{M}{\ia{w}{\dsubst{\td{I}{\psi}}{r}{x}}}$, then
$\ceqtm[\Psi']{\coe{x.\ia{w}{\td{I}{\psi}}}{r}{r}{M}}{M}
{\ia{w}{\dsubst{\td{I}{\psi}}{r}{x}}}$. We have already shown that the left-hand
side
\[
\ceqtm[\Psi']{-}
{\iain{w}{\com{w}{x.\td{A}{\psi}}{r}{r}
{\iaout{w}{M,\dsubst{\td{G}{\psi}}{r}{x}}}{y.\dots,y.\dots}
,\dsubst{\td{F}{\psi}}{r}{x}}}
{\ia{w}{\dsubst{\td{I}{\psi}}{r}{x}}}.
\]
By the introduction rule, \cref{lem:ia-coe}, and \cref{thm:com}, this in turn
\[
\ceqtm[\Psi']{-}
{\iain{w}{\iaout{w}{M,\dsubst{\td{G}{\psi}}{r}{x}},\dsubst{\td{F}{\psi}}{r}{x}}}
{\ia{w}{\dsubst{\td{I}{\psi}}{r}{x}}}
\]
and the result follows by the eta rule.

\paragraph{Fixed point construction}
We obtain a cubical type system with all isomorphism-univalence types (in
addition to the previous connectives) by extending the construction in
\cref{ssec:fixed-point}. Define an operator on cubical type systems $F(E,\Phi) =
(E',\Phi')$ as before with the additional clauses:

\[\begin{aligned}
E' &= \dots\\&\hspace{0.2em}\cup
\{((\Psi,x),\ia{x}{A,B,F,G},\ia{x}{A',B',F',G'}) \mid \\&\hspace{1.7em}
\relcts{(E,\Phi)}{\ceqtype[\Psi,x]{A}{A'}} \land
\relcts{(E,\Phi)}{\ceqtyperes[\Psi,x]{x=1}{B}{B'}} \land \\&\hspace{1.7em}
\relcts{(E,\Phi)}{\ceqtyperes[\Psi,x]{x=1}{\arr{A}{B}}{\arr{A'}{B'}}} \land \\&\hspace{1.7em}
\relcts{(E,\Phi)}{\ceqtyperes[\Psi,x]{x=1}{\arr{B}{A}}{\arr{B'}{A'}}} \land \\&\hspace{1.7em}
\relcts{(E,\Phi)}{\ceqtmres[\Psi,x]{x=1}{F}{F'}{\arr{A}{B}}} \land \\&\hspace{1.7em}
\relcts{(E,\Phi)}{\ceqtmres[\Psi,x]{x=1}{G}{G'}{\arr{B}{A}}} \land \\&\hspace{1.7em}
\relcts{(E,\Phi)}{\eqtmres[\Psi,x]{x=1}{\oft{a}{A}}{\app{G}{\app{F}{a}}}{a}{A}} \land \\&\hspace{1.7em}
\relcts{(E,\Phi)}{\eqtmres[\Psi,x]{x=1}{\oft{b}{B}}{\app{F}{\app{G}{b}}}{b}{B}} \}
\\
\Phi' &= \dots\\&\hspace{0.2em}\cup
\{((\Psi,x),\ia{x}{A,B,F,G},\iain{x}{M,F'},\iain{x}{M',F''}) \mid \\&\hspace{1.7em}
\relcts{(E,\Phi)}{\ceqtm[\Psi,x]{M}{M'}{A}} \land \\&\hspace{1.7em}
\relcts{(E,\Phi)}{\ceqtmres[\Psi,x]{x=1}{F}{F'}{\arr{A}{B}}} \land \\&\hspace{1.7em}
\relcts{(E,\Phi)}{\ceqtmres[\Psi,x]{x=1}{F}{F''}{\arr{A}{B}}} \}
\end{aligned}\]

The resulting $F$ is monotone, so it has a least fixed point. Let $(E,\Phi)$ be
a fixed point of $F$; by our previous argument, $(E,\Phi)$ has all dependent
function types. To see that $(E,\Phi)$ has all isomorphism-univalence types,
suppose that
$\relcts{(E,\Phi)}{\Isoeq{r}{I}{I'}}$. Then for any
$\msubsts{(\Psi',x)}{\psi}{\Psi}$ where $\td{r}{\psi} = x$,
$\relcts{(E,\Phi)}{\Isoeq[\Psi',x]{x}{\td{I}{\psi}}{\td{I'}{\psi}}}$. By the
definition of $F$,
$E((\Psi',x),\ia{x}{\td{I}{\psi}},\ia{x}{\td{I'}{\psi}})$ and
$\Phi((\Psi',x),\ia{x}{\td{I}{\psi}},-,-)$ is the least relation relating
$\iain{x}{M,F'}$ and $\iain{x}{M',F''}$ when:
\begin{gather*}
\relcts{(E,\Phi)}{\ceqtm[\Psi',x]{M}{M'}{\td{A}{\psi}}} \\
\relcts{(E,\Phi)}{\ceqtmres[\Psi',x]{x=1}{\td{F}{\psi}}{F'}
{\arr{\td{A}{\psi}}{\td{B}{\psi}}}} \\
\relcts{(E,\Phi)}{\ceqtmres[\Psi',x]{x=1}{\td{F}{\psi}}{F''}
{\arr{\td{A}{\psi}}{\td{B}{\psi}}}}
\end{gather*}

\paragraph{Proof theory}
The above construction justifies adding the following rules to our proof theory:

\[
\Iso{r}{A,B,F,G} :=
\left\{\begin{array}{l}
  \cwftype{A} \\
  \cwftyperes{r=1}{B} \\
  \coftyperes{r=1}{F}{\arr{A}{B}} \\
  \coftyperes{r=1}{G}{\arr{B}{A}} \\
  \eqtmres{r=1}{\oft{a}{A}}{\app{G}{\app{F}{a}}}{a}{A} \\
  \eqtmres{r=1}{\oft{b}{B}}{\app{F}{\app{G}{b}}}{b}{B}
\end{array}\right.
\]

\[
\infer
  {\ceqtype{\ia{r}{A,B,F,G}}{\ia{r}{A',B',F',G'}}}
  {\begin{array}{l}
   \ceqtype{A}{A'} \\
   \ceqtyperes{r=1}{B}{B'} \\
   \ceqtmres{r=1}{F}{F'}{\arr{A}{B}} \\
   \ceqtmres{r=1}{G}{G'}{\arr{B}{A}} \\
   \eqtmres{r=1}{\oft{a}{A}}{\app{G}{\app{F}{a}}}{a}{A} \\
   \eqtmres{r=1}{\oft{b}{B}}{\app{F}{\app{G}{b}}}{b}{B}
   \end{array}}
\]

\[
\infer
  {\ceqtype{\ia{0}{A,B,F,G}}{A}}
  {\cwftype{A}}
\qquad
\infer
  {\ceqtype{\ia{1}{A,B,F,G}}{B}}
  {\cwftype{B}}
\]

\[
\infer
  {\ceqtm{\iain{r}{M,F}}{\iain{r}{M',F'}}{\ia{r}{A,B,F,G}}}
  {\ceqtm{M}{M'}{A}
  &\Iso{r}{A,B,F,G}
  &\ceqtmres{r=1}{F}{F'}{\arr{A}{B}}}
\]

\[
\infer
  {\ceqtm{\iain{0}{M,F}}{M}{A}}
  {\coftype{M}{A}}
\qquad
\infer
  {\ceqtm{\iain{1}{M,F}}{\app{F}{M}}{B}}
  {\coftype{M}{A}
  &\coftype{F}{\arr{A}{B}}}
\]

\[
\infer
  {\ceqtm{\iaout{r}{M,G}}{\iaout{r}{M',G'}}{A}}
  {\ceqtm{M}{M'}{\ia{r}{A,B,F,G}}
  &\Iso{r}{A,B,F,G}
  &\ceqtmres{r=1}{G}{G'}{\arr{B}{A}}}
\]

\[
\infer
  {\ceqtm{\iaout{0}{M,G}}{M}{A}}
  {\coftype{M}{A}}
\qquad
\infer
  {\ceqtm{\iaout{1}{M,G}}{\app{G}{M}}{A}}
  {\coftype{M}{B}
  &\coftype{G}{\arr{B}{A}}}
\]

\[
\infer
  {\ceqtm{\iaout{r}{\iain{r}{M,F},G}}{M}{A}}
  {\coftype{M}{A}
  &\Iso{r}{A,B,F,G}}
\qquad
\infer
  {\ceqtm{\iain{r}{\iaout{r}{M,G},F}}{M}{\ia{r}{A,B,F,G}}}
  {\coftype{M}{\ia{r}{A,B,F,G}}
  &\Iso{r}{A,B,F,G}}
\]

\newpage
\bibliographystyle{plainnat}
\bibliography{dep}

\end{document}
